\documentclass{article}
\usepackage{times}
\usepackage[pdftex]{graphicx}
\usepackage{amssymb,amsfonts,amsmath,amsthm}
\usepackage{float}
\usepackage{xcolor}

\newtheorem{prop}{Proposition}
\newtheorem{lemma}{Lemma}

\newtheorem{definition}{Definition}
\newtheorem{remark}{Remark}
\newtheorem{cor}{Corollary}

\newcommand\EE {\mathbb E}

\newcommand\RR {\mathbb R}

\def\bone{\mathbf{1}}

\newcommand{\belief}{\mathcal{B}}

\newcommand{\K}{\mathbb{K}}

\usepackage{color}

\newcommand{\hitu}{H_U}

\begin{document}

\title{Robust Trading of Implied Skew}
\author{Sergey~Nadtochiy\footnote{Department of Mathematics, University of Michigan.} and Jan~Ob\l\'oj\footnote{Mathematical Institute, University of Oxford. The research has received funding from the European Research Council under the European Union's Seventh Framework Programme (FP7/2007-2013) / ERC grant agreement no. 335421. The author is also grateful to the Oxford-Man Institute of Quantitative Finance and St John's College in Oxford for their support.}}
\date{Current version: Nov 16, 2016}

\maketitle

\begin{abstract}
In this paper, we present a method for constructing a (static) portfolio of co-maturing European options whose price sign is determined by the skewness level of the associated implied volatility. 
This property holds regardless of the validity of a specific model -- i.e. the method is robust. The strategy is given explicitly and depends only on one's beliefs about the future values of implied skewness, which is an observable market indicator. As such, our method allows to use the existing statistical tools to formulate the beliefs, providing a practical interpretation of the more abstract mathematical setting, in which the belies are understood as a family of probability measures.
One of the applications of the results established herein is a method for trading one's views on the future changes in implied skew, largely independently of other market factors.
Another application of our results provides a concrete improvement of the model-independent super- and sub- replication strategies for barrier options proposed in \cite{BHR}, which exploits the given beliefs on the implied skew. Our theoretical results are tested empirically, using the historical prices of S\&P 500 options.
\end{abstract}

%\vskip 30pt

\section{Introduction}

In this paper, we present a method for constructing a (static) portfolio of co-maturing European options whose price sign is determined by the skewness level of the associated implied volatility. We define the \emph{implied skewness} (or \emph{implied skew}) as a ratio of the implied volatilities of out-of-the-money (OTM) calls to those of co-maturing OTM puts. The target portfolio is constructed so that its price has a specified sign, whenever the level of implied skew satisfies a given upper or lower bound, regardless of the exact values of implied volatilities (e.g. regardless of the overall level of volatility).

The existing literature on implied skew can be split into three main categories, depending on the interpretation of what this term actually means. The first interpretation of implied skew is as a measure of asymmetry of the risk-neutral distribution of the underlying value at a fixed future time (i.e. maturity). This measure is defined through the first three moments of the associated distribution, and each of these moments, in turn, can be computed as a price of a corresponding (static) portfolio of European options (with the same maturity and multiple strikes). Such definition has been adopted by CBOE in the construction of its Skew Index, and its empirical performance is analyzed in \cite{BakshiKapadiaMadan}. Nevertheless, it is worth mentioning that the CBOE Skew Index does not correspond to a price of a tradable portfolio: its value is given by a non-linear function of the prices of several static portfolios of European options. The second approach is to measure the implied skew through the correlation between the underlying returns and the returns of its (spot) volatility. Indeed, in a classical stochastic volatility model, the derivative of a short-term at-the-money implied volatility, with respect to the log-moneyness variable, is determined by the aforementioned correlation. Such a relation may fail to hold in practice, leading to potential trading opportunities: e.g. \cite{CarrWu} proposes a trading strategy that benefits from such deviations. Finally, another interpretation of implied skew is the one adopted in the present paper: namely, the implied skew is defined as a ratio of OTM implied volatilities (calls to puts). It is a popular approach, both among practitioners and among academics (cf. \cite{BaliMurray}, \cite{Bueno}), to attempt to profit from changes in such implied skew by opening a long or short position in a ``risk reversal" portfolio: i.e. a static portfolio consisting of a long position in an OTM put and a short position in (several shares of) a co-maturing OTM call. This portfolio can be combined with a static position in the underlying. Then, choosing the weights of the three instruments appropriately, one can ensure that the price of the portfolio is locally (i.e. asymptotically, for an infinitesimal time period) insensitive to changes in the underlying and to parallel shifts in the implied volatility surface (or its logarithm). At the same time, the price of such a portfolio is locally (for an infinitesimal time period) monotone in the implied skew. Note that such a portfolio can be constructed for any choice of the two strikes: one for an OTM call and one for an OTM put. In the present paper, we show how to construct a static risk reversal portfolio whose price sign is insensitive to \emph{arbitrary} changes in the implied volatility surface, as long as the implied skew stays above (below) a chosen lower (upper) bound. These properties of the portfolio hold for a sufficiently small, but not infinitesimal, time period. It turns out that, in order to achieve the desired properties of a risk reversal portfolio, in addition to fixing the weights of the two OTM options, one has to restrict the choice of the two strikes. In particular, given an arbitrary OTM put (call) strike, our method shows how to choose the appropriate strike of an OTM call (put).

The results presented herein do not require any assumptions on the distribution of the underlying or the implied volatility surface. As such, they present an explicit example of \emph{robust}, or \emph{model-independent}, methods in Financial Mathematics.
%The classical approach in financial mathematics consists in specifying a probabilistic model for future evolution of stock prices. No arbitrage then implies that the price of a contingent claim has to be equal to the initial capital required to replicate the payoff and these prices are most often computed through risk-neutral expectations. 
%This approach, however successful, has important drawbacks, see \cite{HouObloj:15} for a broader discussion. 
%Most importantly, the classical approach is based on the assumption of a specific stochastic description of reality, \emph{the model}, and offers little understanding of what happens if this description is inaccurate, i.e.\ of \emph{model} (or \emph{Knightian}) \emph{uncertainty}. 
Many researchers (cf. the seminal work of \cite{Merton:73}, as well as e.g. \cite{HouObloj:15}, for a broader discussion) have been interested in properties of option prices which are independent of any model specification. In particular, many of such works aim to incorporate existing market constraints or information in order to obtain interesting properties without specific assumptions. Chief examples of such constraints are given by market prices of liquidly traded options on the underlying. The resulting properties may include constraints on the prices of illiquid options, such as barrier options, and the associated sub- and super-replicating strategies: see e.g. \cite{BreedenLitzenberger:78}, \cite{Hobson:98}, \cite{BHR}, \cite{CoxObloj:11}, \cite{CoxWang:11}, \cite{BHLP:13}, \cite{DolinskySoner:14}, \cite{ABPS:16}, \cite{BurzoniFrittelliMaggis:16}.
%\cite{BreedenLitzenberger:78} argued that having market prices of co-maturing European call or put options with all strikes is equivalent to fixing the marginal distribution of the underlying under any martingale measure. \cite{Hobson:98} combined this with probabilistic time-change techniques of Skorokhod embeddings, to obtain model-independent no-arbitrage bounds, along with hedging strategies which enforced them, for a lookback option, given prices of co-maturing European call options.  
%The pioneering work of \cite{Hobson:98} initiated an active field of research. Analogous model-independent considerations were developed for barrier options (see \cite{BHR}) for double barrier options (see \cite{CoxObloj:11}), or for variance options (see \cite{CoxWang:11}), to name just a few. More recently, the problem was re-interpreted as a constrained (martingale) optimal transportation problem with papers focusing on abstract duality results between the cost of the cheapest super-hedging strategy and the supremum over feasible model prices, see \cite{BHLP:13, DolinskySoner:14}. Significant focus has also been placed on developing suitable notions of arbitrage and the associated, model-independent, variants of the Fundamental Theorem of Asset Pricing, see e.g.\ \cite{ABPS:16,BurzoniFrittelliMaggis:16}.
The works on model-independent pricing and hedging mentioned above, while of fundamental interest, have been criticized for the lack of practical relevance due to typically wide interval of no-arbitrage prices they produce.
%Indeed, this was the original reason why \cite{Merton:73} considered the Black-Scholes model in the second part of his paper.
And so, even if the model-independent hedging strategies were shown to be competitive to classical delta-vega hedging, see e.g.\ \cite{OblojUlmer:10}, it is of vital interest to develop robust framework which is capable of interpolating between the model-independent and the model-specific setups. 
A classical probabilistic way to address this problem is by considering families of models (i.e.\ probability measures), which may vary from a single model to the class of all possible models: see e.g. \cite{DenisMartini:06}, \cite{BouchardNutz:14}.
Another approach is to focus not on any probabilistic models but rather on the set of possible paths of observed quantities (e..g prices of financial assets). The latter may vary between a very large (all paths) and a small (the support of some fixed model) set of paths. The crucial insight here is that the \emph{feasible paths} should be defined using market observable and meaningful quantities. This should allow one to apply existing statistical methods to the time series data, together with expert views, to restrict the universe of paths: see e.g. \cite{Lyons:95}, \cite{AvellanedaLevyParas:95}, \cite{Mykland:00,Mykland:03}, \cite{Spoida:14}. 
%There seems to be two natural ways to build such a framework: one which starts with models and another which starts with hedging strategies. We discuss both approaches briefly.
%First, if we think about probabilistic models, then it might be natural to start with a classical model-specific setup and introduce \emph{model uncertainty} by considering simultaneously richer and richer families of models (i.e.\ probability measures) in a way which mirrors the representation of robust utility functionals in decision theory going back to \cite{gilboa89}. This point of view was taken in e.g.\ \cite{DenisMartini:06} and more recently in \cite{BouchardNutz:14}, where a comprehensive discrete time framework was obtained.
%We contribute to the above literature in several ways. First, 
Herein, we consider a framework in which the set of tradable instruments includes the underlying and a set of European options written on it. This has been previously done in works on the so-called \emph{market models}, see e.g. \cite{SchweizerWissel:08}, \cite{CarmonaNadtochiy:09}, \cite{carmona2010tangentlevy}, \cite{carmona2011tangentlevy}, \cite{CarmonaYiNadtochiy}. While these works focused on a fixed probabilistic setting, herein, we pursue the robust approach. Instead of postulating any specific probabilistic features, we prescribe constraints on the price paths of the assets. More specifically, we consider beliefs on the \emph{implied skew}, as the latter is a market observable quantity of interest, and it is reasonable to expect that many practitioners have a view on its future behavior.
The resulting theoretical framework is very parsimonious: it allows for a continuous interpolation between an entirely model--free and (the classical) model-specific setting.
%Our aim in this work is to show how the beliefs about future values of the market \emph{implied skew} can be exploited in order to construct robust trading strategies. The resulting theoretical framework is very parsimonious: it allows for a continuous interpolation between an entirely model--free and (the classical) model-specific setting. It is a pathwise approach which does not require specification of any probability measure but instead offers a consistent way of incorporating beliefs about possible pathwise behaviour of prices. 

As indicated above, one of the applications of the results established herein is a method for trading one's views on the future changes in implied skew, largely independently of other market factors. 
Another application of our results provides a concrete improvement of the model-independent super- and sub- replication strategies for barrier options proposed in \cite{BHR}, which exploits a given set of beliefs on the implied skew. Our theoretical results are tested empirically: herein, we implement the proposed risk reversal portfolio using the historical prices of S\&P 500 options, and verify that its price possesses the desired properties.

Finally, in order to achieve our goal, we introduce and analyze a new family of models, called Piecewise Constant Local Variance Gamma models (see \cite{CN.LVG} for related results). These turn out to be rather remarkable.
First, they allow for explicit computation of European options' prices and of the short-term implied volatility. Second, these models allow for skew in the implied volatility they generate, which is flexible enough to provide a good fit, as well as the upper and lower bounds, on the empirically observed implied skew. Third, any model from this family has an explicit \emph{static hedging} strategy, constructed via the \emph{weak reflection principle} (cf. \cite{CN.SH}, \cite{BN.SH}). We believe that these results are of independent interest.

The remainder of this paper is organized as follows. In Section \ref{se:uop}, we discuss an example which provides a rationale behind the proposed trading strategies and presents one of their applications. Section \ref{sec:setup} formalizes the notion of beliefs on implied skewness. Section \ref{sec:PCLVG} introduces the class of auxiliary models, needed to construct the desired portfolios, for given beliefs. These portfolios are constructed in Section \ref{se:RTIS}, where it is also shown how to use them for trading the implied skew. Section \ref{se:alg} presents another application of the proposed portfolios, concerned with super- and sub-replication of barrier options. Section \ref{se:emp} contains the empirical analysis. Finally, we summarize our results in Section \ref{se:summary}.

\section{Case study: super-hedging of Up-and-Out Put}
\label{se:uop}

We start with a motivating example which illustrates how the beliefs on implied skewness can be used in practice. We consider a trader who wants to hedge a short position in an up-and-out put (UOP) option with the payoff 
\begin{equation}\label{eq.UOP.payoff}
(K-S_T)^+\mathbf{1}_{\{\hitu>T\}},
\end{equation}
where $T>0$, $0<K\vee S_0<U$ and $\hitu= \inf\{t\in[0,T]:\,S_t\geq U\}$. We suppose the trader does not want to make any specific probabilistic assumptions about the dynamics of the risky assets, but believes asset prices are continuous and that the sign of the implied skew will remain constant through time. We show here how such beliefs can be used to improve the mode--independent approach of \cite{BHR}. For simplicity, assume that European options with just two strikes $K_1<U$ and $K_2=U^2/K_1>U$ are liquidly traded at all times and denote their Black-Scholes implied volatilities at time $t$ as $\Sigma^{mkt}_t(K_i,T-t)$, $i=1,2$. 
The model--independent super-replication strategy proposed in \cite{BHR}  is based on the following inequality:
\begin{equation}\label{eq:robust_uop}
(K-S_T)^+\mathbf{1}_{\{\hitu>T\}}
\leq \frac{U-K}{U-K_1}(K_1-S_{T})^+ - \frac{K-K_1}{U-K_1}(S_{T\wedge \hitu}-U),
\end{equation}
which is satisfied as long as $S_{\hitu}=U$. This strategy is discussed in more detail in Section \ref{se:alg}.
The first term in the right hand side of (\ref{eq:robust_uop}) is a static position in $(U-K)/(U-K_1)$ shares of a put option with strike $K_1$ and maturity $T$.
The second term on the RHS in \eqref{eq:robust_uop} is a static position in $(K-K_1)/(U-K_1)$ shares of a forward struck at $U$, with maturity $T$ and which is liquidated at time $\hitu$. This is done at no cost since $S_{\hitu}=U$. The initial capital needed to set up this strategy is
$$
\frac{U-K}{U-K_1}P^{mkt}_0(K_1,T) - \frac{K-K_1}{U-K_1}(S_0-U).
$$
It is easy to see from (\ref{eq:robust_uop}) that the above strategy super-replicates the payoff of the UOP option, provided the paths of the underlying are continuous. Hence, it allows one to hedge the risk of a short position in a UOP option.

Let us now construct an improvement of this strategy by specifying some rather weak beliefs about future implied skew. Namely, we assume that the trader believes that, whenever $S_t=U$, $t\leq T$, the market implied volatility will exhibit a \emph{non-negative skew}:
\begin{equation}\label{eq:uop_beliefs_sign}
\Sigma^{mkt}_t(K_1,T-t)\geq \Sigma^{mkt}_t(K_2,T-t).
\end{equation}
Recall that the well known static hedging formula in the Black-Scholes model (cf. \cite{BC.SH}) implies:
\begin{equation}\label{eq:BS_parity}
P^{bs}(U,K_1,T-t,\sigma) = \frac{K_1}{U} C^{bs}(U,U^2/K_1=K_2,T-t,\sigma), \quad \forall \sigma\geq 0,
\end{equation}
where $P^{bs}, C^{bs}$ denote the put and call prices in the Black-Scholes model. Combining this with the beliefs above, instantly yields
$$
P^{mkt}_t(K_1,T-t) - \frac{K_1}{U} C^{mkt}_t(K_2,T-t) \geq P^{bs}_t(K_1,T-t,\sigma_t) - \frac{K_1}{U} C^{bs}_t(K_2,T-t,\sigma_t) \geq 0,
$$
which has to hold for any $t$ with $S_t=U$ and where we take, e.g.,\ $\sigma_t=(\Sigma^{mkt}_t(K_1,T-t)+\Sigma^{mkt}_t(K_2,T-t))/2$. Then, the improvement of the super-replicating strategy proposed in \cite{BHR} consists in short-selling, initially, additional $\frac{(U-K)K_1}{(U-K_1)U}$ shares of co-maturing calls struck at $K_2$, and closing the position at $\hitu$ (if $\hitu< T$).
If the underlying does not hit $U$ before time $T$ (i.e. $\hitu\geq T$), then the payoff of the additional calls is zero. If the underlying does hit $U$ before $T$ (i.e. $\hitu<T$), then, at the time $\hitu$, the super-replication portfolio is closed at a profit:
$$
\frac{U-K}{U-K_1}P^{mkt}_{\hitu}(K_1,T-\hitu) - \frac{(U-K)K_1}{(U-K_1)U} C^{mkt}_t(K_3,T-\hitu) \geq 0.
$$ 
Clearly, the initial value of the new strategy is smaller than the initial capital of the original one, as we shorted  call options. The strategy exploits the beliefs \eqref{eq:uop_beliefs_sign} to improve the model-independent approach of \cite{BHR}. In hindsight, the explicit construction was possible due to two factors. First, we could identify the boundary case of beliefs in \eqref{eq:uop_beliefs_sign}, which in some sense ``dominates" all the other consistent market scenarios. Second, this boundary case corresponded to a flat implied volatility surface, i.e.\ the Black-Scholes model, which allows for an explicit static hedging formula (\ref{eq:BS_parity}). 
These two observations turn out to be crucial. We will see below that they can be extended to allow for more general beliefs, leading, in particular, to a more general super-replication algorithm described in Section \ref{se:alg}. In order to develop a general approach extending the case study above, we need to describe a convenient, flexible family of martingale models for the underlying, which, on the one hand, can produce a dominating surface for general beliefs, and, on the other hand, allow for explicit static hedging formulas in the spirit of (\ref{eq:BS_parity}). We develop such a family of models in Section \ref{sec:PCLVG}.

Finally, we emphasize that \eqref{eq:uop_beliefs_sign}, and also the superhedging strategy, were invariant with respect to the scaling of the implied volatilities and only depended on the implied skew. We show later in the paper that the same phenomenon occurs in general: the super-replicating strategy depends on the skewness of the dominating surface, but it is robust with respect to scaling of this surface. Thus, a super-replication strategy for a UOP option is an example of \emph{trading the implied skew}.

\section{The Market Setup and Beliefs on Implied Skewness}
\label{sec:setup}

\subsection{The Market Setup}
\label{se:mrkt}

We consider a market consisting of an underlying tradable asset, whose price process we denote by $(S_t)_{t\geq0}$ and a family of European call and put options, with $N\geq 2$ available strikes, $\K=\{K_1,\ldots,K_N\}$, and with arbitrary maturities, whose price processes are denoted
$$
\{(C^{mkt}_t(K_j,T-t), P^{mkt}_t(K_j,T-t))_{t\geq0}\}.
$$ 
Namely, at any time $t\geq0$, the underlying can be purchased or sold in any quantity at the price $S_t$, and a European call or put option with any strike $K_j$ and any maturity $T>t$ can purchased or sold in any quantity at the price $C^{mkt}_t(K_j,T-t)$ or $P^{mkt}_t(K_j,T-t)$, respectively. There is no specific probabilistic model underlying these processes. Instead, we work directly with the set of their admissible paths. In fact, for convenience, we express the prices of European options at time $t$ through their \emph{implied volatility surface}:
\begin{equation}\label{eq.Sigma.mkt.def}
\Sigma^{mkt}_t = \{\Sigma^{mkt}_t(K_j,\cdot):\,(0,\infty)\rightarrow(0,\infty)\}_{j=1}^N.
\end{equation}
For simplicity, we assume that the carrying costs are zero (or, equivalently, that all prices are in units of a numeraire). Then, the call and put prices at time $t$ are defined through the implied volatility as follows:
$$
C^{mkt}_t(K_j,\tau) = C^{bs}(S_t,K_j,\tau,\Sigma^{mkt}_t(K_j,\tau)),
\quad P^{mkt}_t(K_j,\tau) = P^{bs}(S_t,K_j,\tau,\Sigma^{mkt}_t(K_j,\tau)),
$$
for all $j=1,\ldots,N$ and $\tau\geq0$, where $C^{bs}(S,K,\tau,\sigma)$ and $C^{bs}(S,K,\tau,\sigma)$ are the Black-Scholes prices of the European call and put options, respectively, with the current level of underlying $S$, strike $K$, time to maturity $\tau$, volatility $\sigma$, and zero interest and dividend rates. Note that the above definition in particular implies that the put-call parity is satisfied:  
%\comment{not sure if we need to talk about ``forward contract" and if this is actually appropriate?}
$$
C^{mkt}_t(K_j,\tau) - P^{mkt}_t(K_j,\tau) = S_t - K_j.
$$
Clearly, the above represents the time $t$ price of a forward contract, which depends only on $S_t$ and but not on $\Sigma^{mkt}$.
We say that a call option with strike $K$ is in-the-money (ITM) at time $t$ if $K<S_t$, it is at-the-money (ATM) if $K=S_t$, and it is out-of-the-money (OTM) if $K>S_t$. The terminology is similar for the put options, except that the inequalities are reversed.

We now impose conditions on each {\bf market value} $(S_t,\Sigma^{mkt}_t)$ to ensure it does not admit \emph{any static arbitrage}.

\begin{definition}\label{def:admissible}
The market value $(S_t,\Sigma^{mkt}_t)$, where $S_t$ is a nonnegative number and $\Sigma^{mkt}_t$ is as given in (\ref{eq.Sigma.mkt.def}), is {\bf admissible} if the following conditions hold.
\begin{itemize}
\item For any $j=1,\ldots,N$, there exists a strictly positive limit $\Sigma^{mkt}_t(K_j,0+)=\lim_{\tau\downarrow0} \Sigma^{mkt}_t(K_j,\tau)$.

\item For any $j=1,\ldots,N-1$, we have: $C^{mkt}_t(K_j,\tau)\geq C^{mkt}_t(K_{j+1},\tau)$, for all $\tau>0$.

\item For any $0<\tau\leq \tau'$ and any $j=1,\ldots,N$, we have: $C^{mkt}_t(K_j,\tau)\leq C^{mkt}_t(K_j,\tau')$.

\item For any $j=1,\ldots,N-1$, we have: 
$$
\frac{C^{mkt}_t(K_{j+1},\tau) - C^{mkt}_t(K_j,\tau)}{K_{j+1}-K_j} 
\geq \frac{C^{mkt}_t(K_{j},\tau) - C^{mkt}_t(K_{j-1},\tau)}{K_{j}-K_{j-1}},
\quad\forall\,\tau>0,
$$ 
where $K_0=0$.
An {\bf admissible path} of $(S,\Sigma^{mkt})$ is a function from $t\in[0,\infty)$ into the set of all admissible market values, $\{(S_t,\Sigma^{mkt}_t)\}$.
\end{itemize}
\end{definition}

\begin{remark}[On static arbitrages]
\label{rk:no_static_arb}
The conditions of Definition \ref{def:admissible} exclude static arbitrage opportunities using options, however they do not imply that the option prices are compatible with some classical probabilistic model as the option prices may still include some \emph{weak arbitrage opportunities}, see \cite{DavisHobson:07, CoxObloj:11}. 
\end{remark}

\begin{remark}[On dynamic arbitrages]
Note that, with the definition of admissible paths of $(S,\Sigma^{mkt})$ above, we are allowing paths which may admit dynamic arbitrage opportunities. We could further restrict the set of admissible paths to, essentially, those which are supported by \emph{some} martingale measure (for $S$ and the options) and satisfy additional constraints, e.g.\ pathwise beliefs, introduced further in the paper. All of the statements in Sections \ref{se:RTIS} and \ref{se:alg} would hold with such a modified definition of admissibility. In some abstract setups such an approach is necessary to have the pricing--hedging duality, see e.g.\ \cite{BurzoniFrittelliMaggis:16}. However in our setup, similarly to e.g.\ \cite{HouObloj:15}, this is not necessary, as we do not analyze the optimality of the proposed trading strategies in the classical probabilistic sense.
%only prove the optimality of trading strategies constructed herein in a fairly small class of strategies.
%Intuitively speaking, our setup admits enough regularity, so that when a market path starts to admit a dynamic arbitrage, such an arbitrage can be exploited to preserve the superhedging property at no extra cost.
\end{remark}

We work with the sets of admissible paths of $(S,\Sigma^{mkt})$ and, in many cases, introduce additional continuity assumption on $t\mapsto S_t$. Every path of $(S,\Sigma^{mkt})$ is a realization of the future states of the market: it determines the future prices of all traded instruments. Crucially, later on we further restrict the set of possible market values, so that they are admissible and satisfy some additional constraints stated via \emph{beliefs} on the implied volatility.
The beliefs are expected to arise from statistical observations of the implied volatility and may e.g. take a form of confidence intervals. They are discussed in detail in the following subsection.

The traded instruments can be used to construct portfolios, which arise as the linear combinations of the European options and the underlying, with the time varying (adapted) weights. A static portfolio has constant weights and we restrict ourselves to portfolios which trade finitely many times. This is both realistic and allows to circumvent issues with pathwise definitions of a stochastic integral, seen in e.g.\ \cite{Lyons:95, DavisOblojRaval:14}.
We assume that the pricing operator is linear: the price of any portfolio at time $t$ is the associated linear combination of the market prices of its elements. Definition \ref{def:admissible} and Remark \ref{rk:no_static_arb} imply that the price of any static portfolio of co-terminal calls, puts and underlying, whose terminal payoff function (i.e. the associated linear combination of the payoffs of its elements) is nonnegative, has a nonnegative price at any time before the maturity. We will use this observation implicitly in the proofs in Sections \ref{se:RTIS} and \ref{se:alg}.

\subsection{Beliefs on Implied Skewness}
\label{se:beliefs}

We now turn to the specific type of beliefs we want to study in this paper. 
%We consider a (single) barrier option written on the underlying $S$, with the barrier $U$. Our aim is to find a static portfolio of available European options whose price dominates (or is dominated by) the price of the given barrier option at all times up to and including the time of hitting the barrier, provided the implied volatilities of liquid calls and puts stay within specified intervals.
%Without loss of generality, we restrict our analysis to the \textbf{upper} barrier \textbf{knock-out} options, and consider \textbf{superreplicating} strategies only.
In order to define the beliefs, we consider a family of functions:
\begin{equation}\label{eq:beliefs}
a_j,b_j:\,(0,\infty)\rightarrow(0,\infty),\,\,\,\text{s.t.}\,\,\,
\exists\,\, a_j(0+)=\lim_{\tau\downarrow0} a_j(\tau),
\,\,b_j(0+)=\lim_{\tau\downarrow0} b_j(\tau)\,\in(0,\infty),
\quad j=1,\ldots,N.
\end{equation}
which determine the bounds for the \emph{skewness} of the market implied volatility.
Typically, we have $a_j\leq b_j$, but there is no need to enforce this inequality.
%$$
%\Sigma_t(K_j,T) = \frac{1}{\sqrt{T-t}}BSIV\left(C_t(K_j,T),S_t,K_j\right),
%$$
%where $BSIV\left(.,S,K\right)$ is the inverse of the mapping $\sigma \mapsto BS\left(\sigma,S,K\right)$, with $BS\left(\sigma,S,K\right)$ being the time-zero price of a call option with strike $K$, maturity one, volatility level $\sigma$ and the underlying level $S$ (interest and dividend rates are assumed to be zero).

\begin{definition}\label{def:lower.beliefs}
Given a set of functions $\{a_j,b_j\}$, as in (\ref{eq:beliefs}), and a barrier $U>0$, we define the beliefs $\belief_*(T)$, for any $T>0$, as the set of all implied volatility surfaces $\Sigma$, which admit a constant $c>0$ (depending on $\Sigma$), satisfying
\begin{equation}\label{eq.B.lower.def}
\Sigma(K_i,\tau) \leq c b_i(c^2\tau),
\quad \Sigma(K_j,\tau) \geq c a_j(c^2\tau)
\quad \forall\, K_i<U<K_j,\quad \forall\, \tau\in(0,T]. 
\end{equation}
%The family $\belief_*$ is the intersection of $\belief_*(T)$ over all $T>0$.
\end{definition}

Notice that it is natural to make beliefs about the values of implied volatility in the \emph{moneyness} variable, $K/S_t$, as opposed to the strike variable, $K$. This is why the beliefs are formulated w.r.t. a given barrier $U$: it is often convenient to assume that the beliefs are expected to be satisfied when $S_t=U$. Even though it is implicitly assumed that $S_t$ is close to $U$, so that the above inequalities have a correct interpretation, in some cases, we need to consider the above beliefs for $S_t\neq U$. The latter is stated explicitly in the main results of the paper.

Notice also that, for $\Sigma\in\belief_*(T)$, we can estimate the implied skewness from below:
\begin{equation}\label{eq.skew.def}
\frac{\Sigma(K_j,\tau)}{\Sigma(K_i,\tau)} \geq \frac{a_j(c^2\tau)}{b_i(c^2\tau)}
\quad \forall\, K_i<U<K_j,\quad \forall\, \tau\in(0,T].
\end{equation}
Moreover, for any $c>0$, the beliefs generated by $\{c a_j(c^2\cdot),c b_j(c^2\cdot)\}$, are the same as the beliefs generated by $\{a_j,b_j\}$. In other words, the beliefs $\belief_*(T)$ restrict the skewness of a member implied volatility w.r.t. the barrier $U$ (defined as the left hand side of (\ref{eq.skew.def})), via the skewness of the input functions. However, these beliefs are invariant w.r.t. scaling of the input functions.\footnote{In the cases where there are only two available strikes, one on each side of the barrier, and if $\{a_j,b_j\}$ are flat, $\belief_*(T)$ is defined uniquely by specifying $b/a$, as opposed to $a$ and $b$ separately. However, if any of these assumptions are violated, (\ref{eq.skew.def}) fails to imply (\ref{eq.B.lower.def}), even though the two are meant to estimate the same characteristic of implied volatility.}

The beliefs $\belief_*$ provide \emph{lower} bounds on the implied skewness. Similarly, we introduce the beliefs $\belief^*$, which provide \emph{upper} bounds on the implied skew.

\begin{definition}\label{def:upper.beliefs}
Given a set of functions $\{a_j,b_j\}$, as in (\ref{eq:beliefs}), and a barrier $U>0$, we define the beliefs $\belief^*(T)$, for any $T>0$, as the set of all implied volatility surfaces $\Sigma$, which admit a constant $c>0$ (depending on $\Sigma$), satisfying
$$
\Sigma(K_i,\tau) \geq c a_i(c^2\tau),
\quad \Sigma(K_j,\tau) \leq c b_j(c^2\tau)
\quad \forall\, K_i<U<K_j,\quad \forall\, \tau\in(0,T]. 
$$
%The family $\belief^*$ is the intersection of $\belief^*(T)$ over all $T>0$.
\end{definition}

Clearly, any $\Sigma\in\belief^*(T)$, satisfies the upper bound on its skewness:
\begin{equation*}\label{eq.skew.def.2}
\frac{\Sigma(K_j,\tau)}{\Sigma(K_i,\tau)} \leq \frac{b_j(c^2\tau)}{a_i(c^2\tau)}
\quad \forall\, K_i<U<K_j,\quad \forall\, \tau\in(0,T].
\end{equation*}
Notice that the beliefs $\belief_*(T)$ and $\belief^*(T)$ can be defined with a single set of functions -- either $\{a_j\}$ or $\{b_j\}$.
The reason we introduce both sets is that, if $a_j\leq b_j$, the interval $[ca_j(c^2\tau),\tilde cb_j(\tilde c^2\tau)]$ can be interpreted as a confidence interval for the values of $\Sigma(K_j,\tau)$. Then, when $\{a_j\}$ approach zero and $\{b_j\}$ grow to infinity, the beliefs $\belief_*(T)$ and $\belief^*(T)$ include more and more surfaces, converging to the set of all admissible implied volatilities, with arbitrary skewness. On the other hand, when $a_j=b_j$, for all $j$, the set $\belief_*(T)\cap\belief^*(T)$ is either empty or includes implied volatilities with the same short-term skewness (i.e.\ as $\tau\searrow 0$). In this sense, the methodology we develop herein interpolates between a specific market setting (where the short-term implied skewness is determined uniquely) and the model-independent setting (where the implied skewness may take arbitrary values).

Whenever $\belief_*(T)$ or $\belief^*(T)$ is invoked, we assume that it is created with some barrier $U>0$ and some input functions $\{a_j,b_j\}$, as in (\ref{eq:beliefs}). The latter functions may be left unspecified if it causes no ambiguity.
We say that the market implied volatility satisfies the beliefs $\belief_*(T)$ or $\belief^*(T)$ at time $t$, if $\Sigma^{mkt}_t\in\belief_*(T)$ or $\Sigma^{mkt}_t\in\belief^*(T)$, respectively.

\begin{remark}
It may often be convenient to specify the beliefs on implied skewness for a smaller set of strikes than the one that is actually available in the market. For example, the results in Sections \ref{se:RTIS} and \ref{se:alg} are formulated for a fixed strike $K_i<U$, and they describe trading strategies which never include any options with other strikes below the barrier. In this case, the input functions $\{a_j,b_j\}$ (and, hence, the beliefs) may be constructed for the set of strikes that includes all $K_j>U$ and only one $K_i<U$.
\end{remark}

It is often convenient to extend the set $\belief_*$ or $\belief_*$, in order to make it more tractable. This is done via the \emph{dominating implied volatility surfaces}, or, the \emph{dominating surfaces}, for short.

\begin{definition}\label{def:domSurf}
An implied volatility $\Sigma$ is a {\bf lower dominating} surface for beliefs $\belief_*(T)$, if $(U,\Sigma)$ is an admissible market value and
$$
\Sigma(K_i,\tau)\geq b_i(\tau),\quad \Sigma(K_j,\tau)\leq a_j(\tau),\quad \forall\, K_i<U<K_j,\quad \forall\,\tau\in(0,T].
$$
Similarly, an implied volatility $\Sigma$ is an {\bf upper dominating} surface for beliefs $\belief^*(T)$, if $(U,\Sigma)$ is an admissible market value and
$$
\Sigma(K_i,\tau)\leq a_i(\tau),\quad \Sigma(K_j,\tau)\geq b_j(\tau),\quad \forall\, K_i<U<K_j,\quad \forall\,\tau\in(0,T].
$$
\end{definition}

In particular, the lower dominating surface has a lower short-term implied skew than any implied volatility in $\belief_*(T)$. Similarly, the upper dominating surface has a higher short-term implied skew than any implied volatility in $\belief^*(T)$.
Corollary \ref{prop:prop1} below shows that a dominating surface always exists, provided $U\notin\{K_j\}$ and $T>0$ is sufficiently small.
A dominating surface, in general, does not belong to $\belief_*(T)$ or $\belief^*(T)$. However, if it does so for arbitrarily large $T>0$, we say that it \emph{generates} the beliefs $\belief_*=\{\belief_*(T)\}_{T>0}$ or $\belief^*(T)=\{\belief^*(T)\}_{T>0}$.

\begin{definition}
A {\bf lower dominating} surface $\Sigma$ {\bf generates} the beliefs $\belief_*$ if
$$
\Sigma(K_i,\tau) = b_i(\tau),\quad \Sigma(K_j,\tau) = a_j(\tau),\quad \forall\, K_i<U<K_j,\quad \forall\,\tau>0.
$$
Similarly, an {\bf upper dominating} surface $\Sigma$ {\bf generates} the beliefs $\belief^*$ if 
$$
\Sigma(K_i,\tau)= a_i(\tau),\quad \Sigma(K_j,\tau)= b_j(\tau),\quad \forall\, K_i<U<K_j,\quad \forall\,\tau>0.
$$
\end{definition}

It turns out that it is very convenient to analyze the beliefs generated by dominating surfaces. This is due to the fact that a dominating surface can be constructed via a specific martingale model, and this is what allows us to use the classical probabilistic tools of Financial Mathematics in the present robust (pathwise) analysis. We achieve this in Section \ref{se:ext.models} below, where we construct a convenient parametric family of dominating surfaces, produced by a specific class of martingale models for the underlying. The main results are, then, formulated in Sections \ref{se:RTIS} and \ref{se:alg} for the beliefs generated by such dominating surfaces.

\section{Piecewise Constant Local Variance Gamma (PCLVG) Models}
\label{sec:PCLVG}

In this section, we introduce and analyze the relevant family of time-changed local volatility models, which posses three important features. First, they allow for explicit computation of European options' prices and of the short-term implied volatility. Second, the models allow for skew in the implied volatility they generate, which is flexible enough to bound the market implied skew. Third, any model form this family has an explicit static hedging formula, in the spirit of (\ref{eq:BS_parity}). We believe this is a remarkable set of properties and the family of models is of independent interest. The proofs of the results in this section are technical and are all relegated to Appendix \ref{app:proofs} to allow for a more concise presentation. 

The models presented herein are closely related to the Local Variance Gamma (LVG) models, introduced in \cite{CN.LVG}.
%Consider a standard Brownian motion $(W_t)$ defined on a filtered probability space. 
Let the process $D$, taking values in $[0,\infty)$, be defined for $t\geq0$ as the unique weak solution to
\begin{equation}\label{eq.1}
dD_t = \sqrt{2}\sigma(D_t) dW_t,\quad t\leq \zeta=\inf\{t\geq0:\,D_t=0\},\quad D_0=x,
\end{equation}
absorbed at zero. In the above, $W$ is a Brownian motion, and the function $\sigma:[0,\infty)\rightarrow (0,\infty)$ is piece-wise constant, of the form
$$
\sigma(x) = \sigma_{1} \bone_{\left[0,U\right)}(x) + \sigma_{2} \bone_{\left[U,\infty\right)}(x),
$$
with some constants $\sigma_1,\sigma_2>0$. The existence and uniqueness of the solution to (\ref{eq.1}) is discussed in \cite{CN.LVG}.
Finally, we define the stochastic process $X$ as a time change of $D$. Consider a random variable $\xi$, such that $\xi$ is independent of $X$ and has an exponential distribution with mean one. Then, we set 
\begin{equation}\label{eq.S.def}
X_t = D_{t^2\xi},\,\,\,\,\,\,t\geq0.
\end{equation}
It is clear that $X$ is a continuous nonnegative martingale.
Equation (\ref{eq.S.def}), parameterized by $\sigma=(\sigma_1,\sigma_2>0)$, describes a plausible risk-neutral evolution of the underlying (recall that the carrying costs, including interest and dividend rates, are assumed to be zero). 
We will refer to the parametric family of models given by (\ref{eq.S.def}) as the \textbf{PCLVG} (Piecewise Constant Local Variance Gamma) models.
The name is motivated by the similarity of the above construction and the LVG models introduced in \cite{CN.LVG}. Indeed, the main difference between the two is the different choice of a time change: the latter is assumed to be a Gamma process in \cite{CN.LVG}. However, it is also mentioned in \cite{CN.LVG} that any other independent time change will produce models with similar features, as long as the marginal distributions of the time change are exponential. Herein, we choose $(t^2\xi)$ as the desired time change process, so that it has exponential distribution at any time $t>0$. This particular choice is motivated by the desire to have a non-trivial short-term implied volatility produced by the models, and it is explained by the results of the next subsection.

%We consider barrier options with an upper barrier $U$ and assume that $U$, $L$, $M$, and $\sigma$ satisfy:
%\begin{equation}\label{eq.length.cond}
%\int_{L}^U \frac{dx}{\sigma(x)} \leq \in\hitu^M \frac{dx}{\sigma(x)}.
%\end{equation}
%In fact, to construct the dominating models, we only need to use PCLVG with $R=1$ (i.e. only one discontinuity point for $\sigma$) and $L=0$, $M=\infty$. 
%Namely, we consider a partition of the interval $[L,M)=(0,\infty)$ into $[0,U)\cup[U,\infty)$, and a PCLVG model associated with this partition (cf. (\ref{eq.1})) whose diffusion coefficient is given by
%$$
%\sigma(x) = \sigma_{1} \bone_{\left[0,U\right)}(x) + \sigma_{2} \bone_{\left[U,\infty\right)}(x),
%$$
%with some constants $\sigma_1,\sigma_2>0$. 

\subsection{Dominating PCLVG Surfaces}
\label{se:ext.models}

%For any set of beliefs $\belief$, as in \eqref{eq:beliefs}, we aim to show that there exists a maximum maturity $T>0$ and {\bf lower and upper dominating} PCLVG models, $M_*$ and $M^*$. Namely, whenever $S_0=U$, the implied volatility functions produced by $M_*$ and $M^*$, respectively,
%$$
%\left\{\tau\mapsto \Sigma_{*}(K_j,\tau),\,\,\,\,\tau\in[0,T]\right\}_{j=1}^N,
%\,\,\,\,\,\,\,\,\text{and}\,\,\,\,\,\,\,\,\,\left\{\tau\mapsto \Sigma^{*}(K_j,\tau),\,\,\,\,\tau\in[0,T]\right\}_{j=1}^N,
%$$ 
%must satisfy:
%\begin{equation}\label{eq.domIV.def}
% \frac{\Sigma_{*}(K_j,\tau)}{\Sigma_{*}(K_i,\tau)} \leq a^j_i(\tau),
%\,\,\,\,\,\frac{\Sigma^{*}(K_j,\tau)}{\Sigma^{*}(K_i,\tau)} \geq b^j_i(\tau),
%\,\,\,\,\,\forall \tau\in[0,T],\,\,\forall K_i<U<K_j
%\end{equation}

In this subsection, we compute the implied volatility in a PCLVG model and show that it can be used to generate  a dominating surface for general beliefs, provided the barrier does not coincide with any of the strikes and the maturity is sufficiently small. Importantly, our proof is constructive -- it provides a method for computing the dominating surfaces numerically.

Denote the time zero price of a European call option produced by a PCLVG model with parameter $\sigma=(\sigma_1,\sigma_2)$ by
$$
C^{\sigma}(x,K,\tau) = \EE (X_{\tau} - K)^+,\quad X_0=x,
$$
where $x$ is the initial level of underlying, $K$ denotes the strike, and $\tau$ is the time to maturity.
Similarly, we define the price of a European put, denoted $P^{\sigma}(x,K,\tau)$.
%In addition, for any given triplet $(x,K,\tau)$ we denote by $IV^{BS}(\cdot,x,K,\tau)$ the implied volatility function, which is defined the inverse of $\sigma\mapsto C^{BS}(x,K,\tau,\sigma)$.

\begin{lemma}\label{le:Csigma.def}
For any $K\in(0,U)$ and any $\tau\geq0$, we have:
\begin{equation}\label{eq.Csigma.smallK.def}
C^{\sigma}(U,K,\tau) - (U-K) 
= \tau \exp\left(-(U-K)/(\sigma_1\tau)\right)
\frac{ 1 - \exp\left(-2K/(\sigma_1\tau)\right) }
{ \frac{1}{\sigma_1} + \frac{1}{\sigma_2} 
- \left( \frac{1}{\sigma_2} - \frac{1}{\sigma_1}\right) \exp\left(-2U/(\sigma_1\tau)\right)}.
%\sigma_2\tau 
%\frac{ \exp\left(K/(\sigma_1\tau)\right) - \exp\left(-K/(\sigma_1\tau)\right) }
%{ \left( 1+ \frac{\sigma_2}{\sigma_1}\right) \exp\left(U/(\sigma_1\tau)\right) 
%- \left( 1 - \frac{\sigma_2}{\sigma_1}\right) \exp\left(-U/(\sigma_1\tau)\right)}
\end{equation}
For any $K\geq U$ and any $\tau\geq0$, we have:
\begin{equation}\label{eq.Csigma.largeK.def}
C^{\sigma}(U,K,\tau)
=  \tau \exp\left(-(K-U)/(\sigma_2 \tau)\right) 
\frac{ 1 - \exp\left(-2U/(\sigma_1\tau)\right) }
{ \frac{1}{\sigma_1}+ \frac{1}{\sigma_2}
- \left( \frac{1}{\sigma_2} - \frac{1}{\sigma_1} \right) \exp\left(-2U/(\sigma_1\tau)\right)}.
%\frac{\sigma_2\tau}{2} \exp\left(-(K-U)/(\sigma_2 \tau)\right) \left[ 1 +
%\frac{ \left( 1 - \frac{\sigma_2}{\sigma_1}\right) \exp\left(U/(\sigma_1\tau)\right) 
%- \left( 1+ \frac{\sigma_2}{\sigma_1}\right) \exp\left(-U/(\sigma_1\tau)\right) }
%{ \left( 1+ \frac{\sigma_2}{\sigma_1}\right) \exp\left(U/(\sigma_1\tau)\right) 
%- \left( 1 - \frac{\sigma_2}{\sigma_1}\right) \exp\left(-U/(\sigma_1\tau)\right)}\right]
\end{equation}
\end{lemma}

For any $K>0$ and $\tau>0$, denote by $\Sigma^{\sigma}(U,K,\tau)$ the Black-Scholes implied volatility of $C^{\sigma}(U,K,\tau)$. 
%$$
%\Sigma^{\sigma}(U,K,\tau) = \left(C^{bs}(U,K,\tau,\cdot)\right)^{-1}(C^{\sigma}(U,K,\tau)).
%$$
It is easy to check, using equations (\ref{eq.Csigma.smallK.def})--(\ref{eq.Csigma.largeK.def}), that the associated call prices always lie in the interval $(0,U)$, hence, the implied volatility is always well defined. In addition, it is not difficult to check that the call prices, given by (\ref{eq.Csigma.smallK.def})--(\ref{eq.Csigma.largeK.def}), satisfy all the static no-arbitrage conditions of Definition \ref{def:admissible} except the first condition on the existence of a short-term limit of implied volatility. The following proposition verifies the first condition of Definition \ref{def:admissible}, thus, showing that $(U,\Sigma^{\sigma}(U,\cdot,\cdot))$ is an admissible market value, and it provides an explicit expression for the short-term implied volatility.

\begin{prop}\label{prop:IV.asymp}
For any $K>U>0$ and any $\sigma_1,\sigma_2>0$, we have
$$
\lim_{\tau\downarrow0} \Sigma^{\sigma}(U,K,\tau) = \frac{\sqrt{\sigma_2} \log K/U}{\sqrt{2(K-U)}}.
$$
For any $U>K>0$ and any $\sigma_1,\sigma_2>0$, we have
$$
\lim_{\tau\downarrow0} \Sigma^{\sigma}(U,K,\tau) = \frac{\sqrt{\sigma_1} \log U/K}{\sqrt{2(U-K)}}.
$$
\end{prop}

Several graphs of $ \Sigma^{\sigma}(U,\cdot,\tau)$, for small but strictly positive $\tau$, are given in Figure \ref{fig:emp.1}.
%\vskip 6pt
%{\color{red}Graph of $ \Sigma^{\sigma}$.}
%\vskip 6pt
The above proposition has an important corollary which shows that, by choosing $\sigma$, we can construct a dominating surface $\Sigma^{\sigma}(U,\cdot,\cdot)$ for any beliefs, provided the maturity is sufficiently small.

\begin{cor}\label{prop:prop1}
Consider any $U>0$, s.t. $U\notin\{K_j\}$, and any $\{a_j,b_j\}$ as in (\ref{eq:beliefs}). Then, there exist $T>0$ and $\sigma^*=(\sigma^*_1,\sigma^*_2>0)$, such that
\begin{equation}\label{eq.Sigma.upperStar}
\Sigma^{\sigma^*}(U,K_i,\tau) \leq a_i(\tau),
\quad \Sigma^{\sigma^*}(U,K_j,\tau) \geq b_j(\tau)
\quad \forall\, \tau\in(0,T],\,\,\forall K_i<U<K_j.
\end{equation}
In particular, $\Sigma^{\sigma^*}(U,\cdot,\cdot)$ is an upper dominating model for the beliefs $\belief^*$ generated by $U$ and $\{a_j,b_j\}$.
Similarly, there exists $T>0$ and $\sigma_*=(\sigma_{*1},\sigma_{*2}>0)$, such that
\begin{equation}\label{eq.Sigma.lowerStar}
\Sigma^{\sigma_*}(U,K_i,\tau) \geq b_i(\tau),
\quad \Sigma^{\sigma_*}(U,K_j,\tau) \leq a_j(\tau)
\quad \forall\, \tau\in(0,T],\,\,\forall K_i<U<K_j.
\end{equation}
In particular, $\Sigma^{\sigma_*}(U,\cdot,\cdot)$ is a lower dominating model for the beliefs $\belief_*$ generated by $U$ and $\{a_j,b_j\}$.
\end{cor}

\begin{proof}
Proposition \ref{prop:IV.asymp} implies that, by choosing a sufficiently small $\sigma^*_1$ and a sufficiently large $\sigma^*_2$, we can ensure that (\ref{eq.Sigma.upperStar}) holds for $\tau=0^+$. Similarly, we verify (\ref{eq.Sigma.lowerStar}) for $\tau=0^+$. Finally, by continuity, we obtain the statement of the corollary.
\end{proof}

For brevity, when the level of the barrier $U$ is fixed, we simply refer to $\Sigma^{\sigma}$ as a dominating surface instead of writing $\Sigma^{\sigma}(U,\cdot,\cdot)$.

\subsection{Exact Static Hedge in PCLVG Models}
\label{sec:exact_static_PCLVG}

As follows from the results of \cite{CN.SH} (further analyzed and extended in \cite{BN.SH}), the static hedge of a UOP with strike $K>0$ and barrier $U>K$, in a time-homogeneous diffusion model with coefficient $\sigma$, and with zero carrying costs, is given by a European option with the payoff $(K-x)^+ - g(x)$, where
\begin{equation}\label{eq:g_static_generic}
g\left(x\right) = \frac{1}{\pi i} \int_{\varepsilon-\infty i}^{\varepsilon + \infty i}
\frac{\psi^1\left(x,z\right) \psi^1(K,z) }{\psi^1_x(U,z)-\psi^2_x(U,z)} \frac{dz}{z},
\,\,\,\,\,\,\,x>U,
\end{equation}
and $g(x)=0$, for $x\leq U$. In the above, $\varepsilon>0$ is an arbitrary, sufficiently large, constant, and $\psi^i$'s are the two fundamental solutions of the following ODE, normalized to take value $1$ at $x=U$:
$$
\sigma^2\left(x\right) \psi_{xx}\left(x,z\right) - z^2 \psi\left(x,z\right) = 0.
$$
It is easy to verify that, in the present case, we have:
$$
\psi^1(x,z) = \left( \frac{1}{2} + \frac{\sigma_2}{2\sigma_1} \frac{1 + e^{-2Uz/\sigma_1}}{1 - e^{-2Uz/\sigma_1}} \right) 
e^{(x-U)z/\sigma_2} 
+\left( \frac{1}{2} - \frac{\sigma_2}{2\sigma_1} \frac{1 + e^{-2Uz/\sigma_1}}{1 - e^{-2Uz/\sigma_1}} \right)
e^{-(x-U)z/\sigma_2},
\,\,\,\,\,\,\,\,x>U,
$$
%$$
%\psi^1(x,z) =
%\frac{2\frac{\sigma_1}{\sigma_2} e^{(x-\nu)z/\sigma_1}}
%{\left( \frac{\sigma_1}{\sigma_2} + 1 \right) e^{(U-\nu)z/\sigma_2} + \left( \frac{\sigma_1}{\sigma_2} - 1 \right) e^{-(U-\nu)z/\sigma_2}},\,\,\,\,\,\,\,\,x\leq\nu,
%$$
$$
\psi^1_x(U,z) = \frac{z}{\sigma_1}\left( \frac{2}{1 - e^{-2Uz/\sigma_1}} - 1 \right),
\,\,\,\,\,\,\,\,\,\,\,\,
\psi^2_x(U,z) = -\frac{z}{\sigma_2},
$$
$$
\psi^1_x(U,z) - \psi^2_x(U,z)
= z\left( \frac{1}{\sigma_1} \frac{1 + e^{-2Uz/\sigma_1}}{1 - e^{-2Uz/\sigma_1}} + \frac{1}{\sigma_2} \right)
= \frac{2z}{\sigma_2} c_2.
$$
The static hedging property means that, with $X$ denoting the aforementioned time-homogeneous diffusion, we have:
$$
\EE\left( (K-X_T)^+\bone_{\left\{ \sup_{t\in[0,T]} X_t < U \right\}}  | \mathcal{F}_{t\wedge \hitu}\right) 
= \EE\left( (K-X_T)^+ - g(X_T)^+  | \mathcal{F}_{t\wedge \hitu}\right),
\quad \forall\,t\in(0,T],
$$
where $\hitu = \inf\{t\geq0\,:\,X_t\geq U\}$ and $(\mathcal{F}_t)$ is the filtration w.r.t. which $X$ is defined.
The left hand side of the above is the price of a UOP, and the right hand side is the price of a co-maturing European option, evaluated at or before the underlying hits the barrier (if the barrier is not hit, the payoffs of the two options coincide).
The above result also applies to the processes $X$ obtained as an independent continuous time change of a time-homogeneous diffusion, such as the PCLVG processes.
However, strictly speaking, the PCLVG process does not satisfy the assumptions made in \cite{CN.SH}. Indeed, the coefficient $\sigma$ is discontinuous: it is piecewise constant, taking values $\sigma_1$, for $x\in(0,U)$, and $\sigma_2$, for $x\geq U$. Nevertheless, we can still use (\ref{eq:g_static_generic}) to compute a candidate $g$, which is expected to produce a static hedge in the PCLVG model.
Assuming $K\leq U$ and $x>U$, we obtain
$$
\frac{\psi^1\left(x,z\right) \psi^1(K,z) }{\psi^1_x(U,z)-\psi^2_x(U,z)}
= \frac{\sigma_2}{2z} 
\left( e^{(x-U)z/\sigma_2} +  (1/c_2 - 1)e^{-(x-U)z/\sigma_2} \right)
\left( \frac{e^{Kz/\sigma_1} - e^{-Kz/\sigma_1}}{{e^{Uz/\sigma_1} - e^{-Uz/\sigma_1}}} \right),
$$
and, in turn,
$$
g\left(x\right) = \frac{\sigma_2}{2 \pi i} \int_{\varepsilon-\infty i}^{\varepsilon + \infty i}
e^{(x-U)z/\sigma_2} \frac{e^{Kz/\sigma_1} - e^{-Kz/\sigma_1}}{e^{Uz/\sigma_1} - e^{-Uz/\sigma_1}} \frac{dz}{z^2}
$$
$$
+ \frac{\sigma_2}{2\pi i}
\int_{\varepsilon-\infty i}^{\varepsilon + \infty i}
(1/c_2 - 1)e^{-(x-U)z/\sigma_2} \frac{e^{Kz/\sigma_1} - e^{-Kz/\sigma_1}}{{e^{Uz/\sigma_1} - e^{-Uz/\sigma_1}}}
 \frac{dz}{z^2}
$$
By closing the contour of integration on the right, we conclude that the second integral in the above expression is zero.
Thus,
$$
g\left(x\right) = \frac{\sigma_2}{2 \pi i} \int_{\varepsilon-\infty i}^{\varepsilon + \infty i}
e^{(x-U)z/\sigma_2} \frac{e^{Kz/\sigma_1} - e^{-Kz/\sigma_1}}{e^{Uz/\sigma_1} - e^{-Uz/\sigma_1}} \frac{dz}{z^2}
$$
$$
= \sum_{n=0}^{\infty} \frac{\sigma_2}{2 \pi i} \int_{\varepsilon-\infty i}^{\varepsilon + \infty i}
e^{(x-U)z/\sigma_2 - Uz/\sigma_1} \left(e^{Kz/\sigma_1} - e^{-Kz/\sigma_1}\right) e^{-2Unz/\sigma_1} \frac{dz}{z^2}
$$
$$
= \sum_{n=0}^{\infty} \frac{\sigma_2}{2 \pi i} \int_{\varepsilon-\infty i}^{\varepsilon + \infty i}
e^{(x-U)z/\sigma_2 - (U(2n+1) - K)z/\sigma_1} \frac{dz}{z^2}
$$
$$
- \sum_{n=0}^{\infty} \frac{\sigma_2}{2 \pi i} \int_{\varepsilon-\infty i}^{\varepsilon + \infty i}
e^{(x-U)z/\sigma_2 - (U(2n+1) + K)z/\sigma_1} \frac{dz}{z^2}.
$$
By closing the contour of integration on the right, we conclude that the integrand inside the $n$th integral in the first sum is zero when $x\leq U + ((2n+1)U-K)\sigma_2/\sigma_1$. For the other values of $x$, we close the contour of integration on the left and compute the integral via the residue calculus. Similarly, we proceed with the second sum. As a result, we obtain
\begin{equation}\label{eq.g.def}
g\left(x\right) 
= \sum_{n=0}^{\infty}
\left(x-U - (U(2n+1) - K)\sigma_2/\sigma_1\right)^+ - \left(x-U - (U(2n+1) + K)\sigma_2/\sigma_1\right)^+.
\end{equation}

A graph of $g$ is given in Figure \ref{fig:g.funcs}.
%\vskip 6pt
%{\color{red}Graph of $g$.}
%\vskip 6pt
Having an explicit candidate, we can use Markov property and explicit option prices in \eqref{eq.Csigma.smallK.def}--\eqref{eq.Csigma.largeK.def} to show that this is indeed the correct function. The latter is stated precisely in the following proposition, whose proof is given in Appendix \ref{app:proofs}.
\begin{prop}\label{prop:exactHedge}
Assume that $X$ follows a PCLVG model, with $0<X_0\leq U$, $\sigma(x) =  \sigma_{1} \bone_{\left[0,U\right)}(x) + \sigma_{2} \bone_{\left[U,\infty\right)}(x)$, $\sigma_1,\sigma_2>0$. Then, for any $0<K\leq U$, any $T>0$, and all $t\in[0,T]$, we have:
\begin{equation}\label{eq.prop.exactHedge}
\EE\left( (K-X_T)^+\bone_{\left\{ \sup_{t\in[0,T]} X_t < U \right\}}  | \mathcal{F}_{t\wedge \hitu}\right) 
= \EE\left( (K-X_T)^+ - g(X_T)^+  | \mathcal{F}_{t\wedge \hitu}\right),
\end{equation}
where $g$ is given by (\ref{eq.g.def}) and $\hitu = \inf\{t\geq0\,:\,X_t\geq U\}$.
\end{prop}

Assuming $X_0=U$, taking expectations in (\ref{eq.prop.exactHedge}) and using \eqref{eq.g.def}, we obtain the following formula, which holds for all $\tau\geq0$, any $K<U$, and any $\sigma=(\sigma_1,\sigma_2>0)$:
\begin{equation}\label{eq.ESH.putViaCalls}
P^{\sigma}(U,K,\tau) = \sum_{n=0}^{\infty}
C^{\sigma}\left(U,U + (U(2n+1) - K)\frac{\sigma_2}{\sigma_1},\tau\right) 
- C^{\sigma}\left(U,U + (U(2n+1) + K)\frac{\sigma_2}{\sigma_1},\tau\right).
\end{equation}
The above formula is the analogue of (\ref{eq:BS_parity}), but it holds in a model in which the underlying is \emph{not symmetric} w.r.t. the barrier (and, hence, the implied volatility may have a skew).

Notice also that
\begin{equation}\label{eq.gStar.def}
\underline{g}(x) \leq g(x) \leq \overline{g}(x),\,\,\,\,\,\,\,\,\,\forall x\in\RR,
\end{equation}
where
\begin{equation}\label{eq.ESH.K.def}
\underline{g}(x) = \frac{K}{U}\left( x - \mathbf{K}^{\sigma,U}(K) \right)^+,
\,\,\,\,\,\,\,\,\, \overline{g}(x) = \left( x - \mathbf{K}^{\sigma,U}(K) \right)^+,
\,\,\,\,\,\,\,\,\,\mathbf{K}^{\sigma,U}(K) = U + (U-K)\sigma_2/\sigma_1.
\end{equation}
The function $\underline{g}$ is the largest convex function dominating $g$ from below.
%\comment{Is this clear/obvious here?}
The function $\overline{g}$ is one of the minimal convex functions dominating $g$ from above: i.e. there exists no other convex function which dominates $g$ from above and is dominated by $\overline{g}$ from above.
The graphs of $g$, $\overline{g}$, and $\underline{g}$, are given in Figure \ref{fig:g.funcs}. We can see that $\overline{g}$ provides a good approximation of $g$ for the values of the argument that re not too far from the barrier $U$ (which are the most important values, if the maturity is small and the underlying is close to $U$). Moreover, the right hand side of Figure \ref{fig:g.funcs} shows that, if $K$ is close to $U$ (this graph uses $K=0.95 U$, which is the value used in our empirical analysis), $g$ is very close to the payoff of a call struck at $\mathbf{K}^{\sigma,U}(K)$, and, in particular, both $\overline{g}$, and $\underline{g}$ provide a very good approximation of $g$ in such a regime.
%\vskip 6pt
%{\color{red}Graph of $g$, $\overline{g}$, $\underline{g}$.}
%\vskip 6pt

Thus, in a PCLVG model, conditional on the level of underlying, we can bound the price of an OTM put uniformly, from above and from below, by the prices of respective static portfolios, consisting of a single OTM call each:
\begin{equation}\label{eq.exactSH.1}
\frac{K}{U} C^{\sigma}(U,\mathbf{K}^{\sigma,U}(K),\tau) < P^{\sigma}(U,K,\tau)
< C^{\sigma}(U,\mathbf{K}^{\sigma,U}(K),\tau),
\quad \forall\,\tau>0,\,\,\forall\, 0<K\leq U.
\end{equation}
Note that the above inequalities hold for all $\tau$ (not only asymptotically, for $\tau\rightarrow0$) and, hence, for all times until the maturity, whenever the underlying is at $U$.
Considering a symmetric problem of static hedging a down-and-out-call (DOC), we can derive the following, equivalent, system:
\begin{equation}\label{eq.exactSH.2}
P^{\sigma}(U,(\mathbf{K}^{\sigma,U})^{-1}(K),\tau) < C^{\sigma}(U,K,\tau)
< \frac{U}{(\mathbf{K}^{\sigma,U})^{-1}(K)}P^{\sigma}(U,(\mathbf{K}^{\sigma,U})^{-1}(K),\tau),
\quad \forall\,\tau>0,\, 0<K\leq U.
\end{equation}

\begin{remark}\label{re:exactSH.tight}
It is easy to deduce from Lemma \ref{le:Csigma.def} that
$$
P^{\sigma}(U,K,\tau)
\sim C^{\sigma}(U,\mathbf{K}^{\sigma,U}(K),\tau),
\quad \tau\rightarrow0.
$$
In particular, the second inequality in (\ref{eq.exactSH.1}) is ``infinitely tighter" than the first one, for small $\tau$. 
The above equivalence, in principle, is well known in a general diffusion model (cf. \cite{Gesell}) and can be obtained without establishing the exact static hedge: one simply needs to notice that $K$ and $\mathbf{K}^{\sigma,U}(K)$ are symmetric w.r.t. the geodesic distance given by $dx/\sigma(x)$.
However, such an equivalence, alone, does not yield the second inequality in (\ref{eq.exactSH.1}), even if it is restricted to small $\tau$, which explains the need for Proposition \ref{prop:exactHedge}. In addition, Proposition \ref{prop:exactHedge} yields the first inequality in (\ref{eq.exactSH.1}) and shows that both of them hold for all $\tau\geq0$.
\end{remark}

\section{Trading the Deviations of Implied Skew}
\label{se:RTIS}

Let us show how (\ref{eq.exactSH.1}) can be used to construct a portfolio for \emph{robust trading} of the \emph{implied skew} -- further referred to as RTIS. Buying such a portfolio, whenever the implied skewness deviates from its typical range of values, and selling it when the implied skewness comes back to its normal range, should yield a positive return. More precisely, for a given barrier and a given set of beliefs, $\belief_*$ or $\belief^*$, on the implied skewness, the RTIS portfolio is a static portfolio of vanilla options, which satisfies the following two properties:
\begin{enumerate}
\item if the implied skewness satisfies the beliefs, and if the underlying is on the appropriate side of the barrier, the portfolio has a positive price;
\item if the implied skewness deviates sufficiently far from the beliefs, and if the underlying is not too far from the barrier, the portfolio has a negative price.
\end{enumerate}
Note that, whenever the underlying is on the appropriate side of the barrier, the price of an RTIS portfolio is guaranteed to have a positive sign, depending on the level of implied skew, but regardless of the overall level of implied volatility, and independent of the exact value of the underlying. Of course, the exact price of an RTIS portfolio may depend on the entire shape of the implied volatility and on the underlying value.
The precise meaning of the two properties of an RTIS portfolio is given below, in Propositions \ref{prop:RTSIS1} and \ref{prop:RTSIS2}.

However, the above formulation is sufficient to see the practical benefit of RTIS. For example, assume that, at time $t$, the implied skewness measured relative to the current level of underlying (i.e. using $S_t$ as a barrier) deviates sufficiently far from the beliefs (constructed w.r.t. $S_t$). Then, one can choose a new barrier $U$ and construct the corresponding beliefs w.r.t. $U$. The new barrier needs to be close to $S_t$, so that the implied skewness measured relative to $U$ is also sufficiently far from the beliefs, and, in turn, the price of the associated RTIS portfolio is negative at time $t$.
On the other hand, the barrier $U$ should not be too close to $S_t$, so that the underlying is likely to remain on the same side of $U$ for some time (the choice of optimal $U$ is the ``art" of implementing such a strategy successfully). One can, then, open a long position in the RTIS portfolio at time $t$. If, at some future time, the skewness returns to a level consistent with the beliefs (one has to assume that this will occur eventually, as this is the principle by which the beliefs are constructed), and the underlying remains on the appropriate side of the barrier, the price of the RTIS portfolio becomes positive, yielding a positive return.
Whenever this strategy is implementable, it yields a positive return. Its implementability, in turn, depends on the behavior of the implied skew and the underlying, but not on any other changes in the implied volatility. 
%Heuristically, the average return of the portfolio due to changes in the underlying is expected to be negligible (this reflects the absence of obvious arbitrage), and its sensitivity to changes in the implied skew should introduce a positive bias in its returns (provided the portfolio is purchased when the skewness is sufficiently far from the beliefs). 
The sensitivity of the strategy with respect to changes in the underlying can be mitigated by opening an additional static position in a forward contract struck at $S_t$ (which has initial price zero and is insensitive to changes in the implied volatility), so that the initial Black-Scholes delta of the portfolio is zero.

Let us show how to construct an RTIS portfolio satisfying the two defining properties stated above.
The choice of a portfolio depends on whether we are given beliefs $\belief_*$ or $\belief^*$.
Herein, we consider the case of $\belief_*$, with the other case discussed in Remark \ref{rem:RTIS.upperBeliefs}.
Consider a barrier $U$, and beliefs $\belief_*$, generated by $\Sigma^{\sigma_*}$, with some $\sigma_*=(\sigma_{*1},\sigma_{*2}>0)$ (cf. Definition \ref{def:domSurf}). 
%Denote by $M_*$ a lower dominating PCLVG  model, whose diffusion coefficient has at most one discontinuity, defined for the maximum maturity $T$ (such $M_*$ and $T>0$ always exist, according to Proposition \ref{prop:prop1}). Denote the values of the associated diffusion coefficient by $\sigma_*=(\sigma_{*1},\sigma_{*2})$.
Consider any time $t$, s.t. $(S_t,\Sigma^{mkt}_t)$ is admissible, $\Sigma^{mkt}_t\in\belief_*(T)$, with some $T>0$, and $S_t=U$. 
%Assume also that $S_t=U$.
%Let us fix an index $i$, s.t. $K_i<U$.
Notice that, scaling $\sigma_*$ with a positive constant $c^2$, scales the $\Sigma^{\sigma_*}$ in the natural way: 
$$
\Sigma^{c^2\sigma_*}(U,K,\tau) = c \Sigma^{\sigma_*}(U,K,c^2\tau).
$$
Then, multiplying $\sigma_*$ by a positive constant, we obtain $\tilde{\sigma}=(\tilde{\sigma}_1,\tilde{\sigma}_2)$, such that $\tilde{\sigma}_2/\tilde{\sigma}_1=\sigma_{*2}/\sigma_{*1}$ and
\begin{equation}\label{eq.extIV.dom}
\Sigma^{\tilde{\sigma}}(U,K_i,\tau) \geq \Sigma^{mkt}_t(K_i,\tau),
\quad  \Sigma^{\tilde{\sigma}}(U,K_j,\tau) \leq \Sigma^{mkt}_t(K_j,\tau),
\quad \forall K_i<U<K_j,\,\,\forall \tau\in(0,T].
\end{equation}

\begin{remark}
%\comment{A priori this is not the same as saying that these other beliefs do not matter, it is just saying that our portfolio is not making use/exploiting these other beliefs?}
In what follows, we fix an arbitrary index $i$, s.t. $K_i<U$, and avoid using options with other strikes below the barrier. In particular, the subsequent statements hold true if we replace the set $\belief_*(T)$ by its extension, obtained by requiring (\ref{eq.B.lower.def}) to hold only for the fixed $K_i<U$ (and for all $K_j>U$).
\end{remark}

%Denote by $C^{bs}(U,K,\tau,\sigma)$ and $P^{bs}(U,K,\tau,\sigma)$ the Black-Scholes prices of a European call and put option, respectively, with the underlying level $U$, strike $K$, time to maturity $\tau$, volatility $\sigma$, and with zero carrying costs.
%Denote also by $C^{mkt}(U,K,\tau)$ and $P^{mkt}(U,K,\tau)$ the market prices of European calls and puts, respectively.
It is easy to see that (\ref{eq.extIV.dom}) implies, for all $K_i<U<K_j$ and all $\tau\in(0,T]$,
\begin{equation}\label{eq.TS.Ctildesigma.mkt}
C^{\tilde{\sigma}}(U,K_j,\tau)
= C^{bs}(U,K_j,\tau,\Sigma^{\tilde{\sigma}}(U,K_j,\tau))
\leq C^{bs}(U,K_j,\tau,\Sigma^{mkt}_t(K_j,\tau))
= C^{mkt}_t(K_j,\tau),
\end{equation}
\begin{equation}\label{eq.TS.Ptildesigma.mkt}
P^{\tilde{\sigma}}(U,K_i,\tau)
= P^{bs}(U,K_i,\tau,\Sigma^{\tilde{\sigma}}(U,K_i,\tau))
\geq P^{bs}(U,K_i,\tau,\Sigma^{mkt}_t(K_i,\tau))
= P^{mkt}_t(K_i,\tau).
\end{equation}
Let us introduce the indices $\overline{j}^{\sigma_*}(i)$ and $\underline{j}^{\sigma_*}(i)$:
\begin{equation}\label{eq.overj.def}
\overline{j}^{\sigma_*}(i) = \min\{j\in\{1,\ldots,N\}:\,K_j\geq \mathbf{K}^{\sigma_*,U}(K_i)\},
\end{equation}
\begin{equation}\label{eq.underj.def}
\underline{j}^{\sigma_*}(i) = \max\{j\in\{1,\ldots,N\}:\,U<K_j\leq \mathbf{K}^{\sigma_*,U}(K_i)\}.
\end{equation}
We assume that the above sets are non-empty so that $\overline{j}(i)$ and $\underline{j}(i)$ are well defined.
Then, we have:
$$
C^{\tilde{\sigma}}(U,\mathbf{K}^{\tilde{\sigma},U}(K_i),\tau)
= C^{\tilde{\sigma}}(U,\mathbf{K}^{\sigma_*,U}(K_i),\tau)
\leq C^{\tilde{\sigma}}(U,K_{\underline{j}^{\sigma_*}(i)},\tau).
$$
Collecting the above inequalities and making use of (\ref{eq.exactSH.1}), we obtain:
\begin{equation}\label{eq.robustSH.1}
P^{mkt}_t(K_i,\tau)
< C^{mkt}_t(K_{\underline{j}^{\sigma_*}(i)},\tau),
\quad \forall\,\tau\in(0,T].
\end{equation}
It only remains to notice that, if we fix $\Sigma^{mkt}_t$, increasing the value of $S_t$ will decrease the left hand side of (\ref{eq.robustSH.1}) and increase its left hand side.
Thus, we have proved the following proposition, which shows how to construct an RTIS portfolio and formalizes the first defining property of RTIS.

\begin{prop}\label{prop:RTSIS1}
Consider a barrier $U>0$ and beliefs $\belief_*$, generated by $\Sigma^{\sigma_*}$, with an arbitrary (but fixed) $\sigma_*=(\sigma_{*1},\sigma_{*2}>0)$.
Assume that there exists an index $i$, s.t. $K_i<U$ and $\underline{j}^{\sigma_*}(i)$ is well defined via (\ref{eq.underj.def}).
Then, at any time $t$, at which $(S_t,\Sigma^{mkt}_t)$ is admissible, $S_t \geq U$, and $\Sigma^{mkt}_t\in\belief_*(T)$, with some $T>0$, the portfolio consisting of a long position in one share of a European call, with the strike $K_{\underline{j}^{\sigma_*}(i)}$ and with any time to maturity less than $T$, and a short position in one share of a co-maturing European put, with the strike $K_{i}$, has a {\bf positive} price: i.e.
\begin{equation}\label{eq.RTIS.def}
C^{mkt}_t(K_{\underline{j}^{\sigma_*}(i)},\tau) - P^{mkt}_t(K_i,\tau) > 0,
\quad \forall\tau\in(0,T].
\end{equation}
\end{prop}

%For a set of upper bounds on implied skewness, $\belief^*$, and an associated upper dominating PCLVG model given by $\sigma^*=(\sigma^*_1,\sigma^*_2)$, repeating the above derivations (with the opposite inequalities), we easily establish
%\begin{equation}\label{eq.robustSH.2}
%\frac{K_i}{U} C^{mkt}(U,K_{\overline{j}^{\sigma^*}(i)},\tau) 
%< P^{mkt}(U,K_i,\tau),
%\quad \forall\,\tau\in(0,T),
%\end{equation}
%whenever $\Sigma^{mkt}$ satisfies the beliefs.

Proposition \ref{prop:RTSIS1} defines the candidate RTIS portfolio via (\ref{eq.RTIS.def}) and shows that it satisfies the first defining property of RTIS (in the sense stated in the proposition).
Proposition \ref{prop:RTSIS2}, below, clarifies the second defining property of RTIS and shows that the portfolio given by (\ref{eq.RTIS.def}) satisfies this property. Namely, Proposition \ref{prop:RTSIS2} shows that the price of the proposed portfolio becomes negative when the implied skew is sufficiently small (and the underlying is at the barrier).
%shows that the static portfolio of OTM options, constructed therein, has a positive price whenever the implied skew satisfies the given lower bounds (and the underlying is at the barrier). Then, if the skew deviates significantly below these bounds, one may naturally expect that the portfolio's price becomes negative. Then, one can buy the portfolio at a negative price and liquidate it at a positive profit, when the skew returns to its normal range of values (and the underlying stays around the same level).

\begin{prop}\label{prop:RTSIS2}
Consider a barrier $U$ and any two pairs of positive numbers, $\sigma_*=(\sigma_{*1},\sigma_{*2})$ and $\sigma^*=(\sigma^*_1,\sigma^*_2)$. 
%Assume that there exists an index $i$, s.t. $K_i<U$, and the index $\underline{j}^{\sigma_*}(i)$ is well defined via (\ref{eq.underj.def}).
Assume that there exist indices $(i,j)$, s.t. $K_i<U$ and
\begin{equation}\label{eq.TS.prop2.1}
U + (U-K_i)\frac{\sigma^*_{2}}{\sigma^*_{1}} = \mathbf{K}^{\sigma^*,U}(K_i) < K_j \leq \mathbf{K}^{\sigma_*,U}(K_i) = U + (U-K_i)\frac{\sigma_{*2}}{\sigma_{*1}}.
\end{equation}
%Consider any set of upper bounds on implied skewness, $\belief^*$, as in \eqref{eq:beliefs}, with an upper dominating PCLVG  model given by $\sigma^*$.
Denote by $\belief^*$ the beliefs generated by $\Sigma^{\sigma^*}$.
Then, at any time $t$, at which $(S_t,\Sigma^{mkt}_t)$ is admissible, $S_t=U$, and $\Sigma^{mkt}_t\in\belief^*(T)$, with some $T>0$, the portfolio described in Proposition \ref{prop:RTSIS1} has a {\bf negative} price, provided its maturity is sufficiently small: i.e.
$$
C^{mkt}_t(K_{\underline{j}^{\sigma_*}(i)},\tau) - P^{mkt}_t(K_i,\tau) < 0,
\quad \forall\tau\in(0,T'],
$$
for some $T'\in(0,T]$.
\end{prop}

\begin{remark}
The condition (\ref{eq.TS.prop2.1}) is satisfied whenever $\underline{j}^{\sigma_*}(i)$ and $\overline{j}^{\sigma^*}(i)$ are well defined, $\sigma^*_{2}/\sigma^*_{1} < \sigma_{*2}/\sigma_{*1}$, and the grid of available strikes, above the barrier, is sufficiently rich.
\end{remark}

%\begin{remark}
%The set of upper bounds on implied skewness, $\belief^*$, appearing in Proposition \ref{prop:RTSIS2}, can be constructed by simply setting
%$$
%b^j_i(\tau) = \Sigma^{\sigma^*}(U,K_j,\tau)/\Sigma^{\sigma^*}(U,K_i,\tau),
%$$
%for a fixed $K_i$ and all $K_j>U$.
%\end{remark}

\begin{proof}
%The statement of Proposition \ref{prop:RTSIS1} follows directly from the derivations in the preceding paragraphs. 
To prove Proposition \ref{prop:RTSIS2}, we notice that its assumptions imply that $\underline{j}^{\sigma_*}(i)$ and $\overline{j}^{\sigma^*}(i)$ are well defined, and that there exists $j$, s.t.
$$
\mathbf{K}^{\sigma^*,U}(K_i) < K_j \leq K_{\underline{j}^{\sigma_*}(i)}.
$$
The latter, along with the first inequality in (\ref{eq.exactSH.1}) and Lemma \ref{le:Csigma.def}, implies that, for any $\tilde{\sigma}=(\tilde{\sigma}_1,\tilde{\sigma}_2>0)$, with $\tilde{\sigma}_2/\tilde{\sigma}_1 = \sigma^*_2/\sigma^*_1$, there exists $T'>0$, s.t.
\begin{equation*}
C^{\tilde{\sigma}}(U,K_{j},\tau) 
\leq \frac{K_i}{U} C^{\tilde{\sigma}}(U,\mathbf{K}^{\sigma^*,U}(K_i),\tau)
< P^{\tilde{\sigma}}(U,K_i,\tau),
\quad \forall \tau\in(0,T').
\end{equation*}
In addition,
$$
C^{\tilde{\sigma}}(U,K_{\underline{j}^{\sigma_*}(i)},\tau) \leq C^{\tilde{\sigma}}(U,K_{j},\tau),
\quad \forall \tau>0.
$$
Finally, we can choose $\tilde{\sigma}$, s.t. $\tilde{\sigma}_2/\tilde{\sigma}_1 = \sigma^*_2/\sigma^*_1$ and
\begin{equation*}
C^{\tilde{\sigma}}(U,K_{\underline{j}^{\sigma_*}(i)},\tau) \geq C^{mkt}(U,K_{\underline{j}^{\sigma_*}(i)},\tau),
\quad P^{\tilde{\sigma}}(U,K_i,\tau) \leq P^{mkt}(U,K_i,\tau),
\end{equation*}
for all $\tau\in(0,T]$.
Collecting the above inequalities, we obtain the statement of the proposition.
\end{proof}

Let us clarify the behavior of the proposed RTIS portfolio.
Notice that, ideally, the sign of the price of an RTIS portfolio should indicate whether the implied skew is outside of its normal range or not. 
In accordance with this argument, Proposition \ref{prop:RTSIS1} shows that the price of the proposed RTIS portfolio is positive whenever the beliefs are satisfied (and the underlying is at or above the barrier).
Similarly, it is desirable that the price of the portfolio becomes negative whenever the beliefs are not satisfied (and the underlying is at the barrier).
However, this ideal behavior seems very difficult to achieve, and it certainly cannot be accomplished using our current methods (unless the market implied volatility is always given precisely by a PCLVG model). 
Therefore, we focus on the deviations of implied skew that can be detected via the proposed family of dominating surfaces $\{\Sigma^{\sigma}\}$.
In particular, Proposition \ref{prop:RTSIS2} states that the price of the proposed RTIS portfolio becomes negative when the market implied skewness can be ``separated" from the beliefs, generated by $\Sigma^{\sigma_*}$, by an upper dominating surface $\Sigma^{\sigma^*}$ whose skewness is lower: i.e. $\sigma^*_{2}/\sigma^*_{1}<\sigma_{*2}/\sigma_{*1}$ (and the underlying is at the barrier).\footnote{These statements hold for all sufficiently small maturities.}

It is easy to see that the statements of Propositions \ref{prop:RTSIS1} and \ref{prop:RTSIS2} apply to other portfolios, in addition to the one given by (\ref{eq.RTIS.def}).
%\begin{equation}\label{eq.RTIS.def}
%C^{mkt}(U,K_{\underline{j}^{\sigma_*}(i)},\tau) - P^{mkt}(U,K_i,\tau). 
%\end{equation}
For example, if we modify the proposed portfolio by adding one more share of the call option, the statements of Propositions \ref{prop:RTSIS1} and \ref{prop:RTSIS2} will still hold for the new portfolio. Proposition \ref{prop:RTSIS3}, below, shows that such a modification (i.e. adding more shares of the call option) describes all possible static portfolios, consisting of co-maturing OTM options with two distinct strikes, for which Propositions \ref{prop:RTSIS1} and \ref{prop:RTSIS2} hold. The proposed RTIS portfolio is special because it is the \emph{minimal} one.
%A natural question is, then: why do we choose the specific portfolio given by (\ref{eq.RTIS.def}) as an RTIS portfolio? The reason for this choice is the \emph{efficiency} and \emph{minimality} of the proposed portfolio, as described in the Proposition \ref{prop:RTSIS3} below.
%A good RTIS portfolio, clearly, should have a price that is more sensitive to the deviations of the implied skew, so that fewer trading opportunities are lost (i.e. not reflected in the sign of the portfolio). In the context of Proposition \ref{prop:RTSIS2}, it means that the price of a good RTIS portfolio should becomes negative, whenever $\sigma^*_{2}/\sigma^*_{1}<\sigma_{*2}/\sigma_{*1}$. it is natural to call such a portfolio \emph{efficient}.
%It turns out that the portfolio given by (\ref{eq.RTIS.def}) satisfies this optimality property in a natural class of static portfolios, provided the beliefs $\belief_*$ are generated by $\sigma_*$ (cf. {\color{red}Definition ???}).
%The latter assumption reflects the \emph{efficiency} of the lower dominating model given by $\sigma_*$, and it ensures that a cheaper portfolio cannot be constructed by simply choosing a more efficient (i.e. ``tighter") lower dominating model for the given beliefs.
%The class of available static portfolios is restricted as all linear combinations of an OTM call and a co-maturing OTM put, with the two strikes chosen arbitrarily. 
The options allowed in the portfolio are co-maturing because the portfolio is meant to reflect the implied skew, which, in turn, depends only on the shape of the implied volatility across strikes.
The number of strikes is limited to two, as the methods developed herein are designed for small maturities, hence, the prices of deeper OTM options become relatively negligible.
%It is natural to define the class of available static portfolios as all (finite) linear combination of co-maturing European options (with the fixed finite set of available strikes) and forward contracts written on the same asset.\footnote{The options are co-maturing because the portfolio is meant to reflect the implied skew, which, in turn, depends only on the shape of the implied volatility across strikes.} However, the ITM options can be represented through the OTM options and forwards. 
In addition, we do not include forward contracts, as their prices are insensitive to the changes in implied volatility.\footnote{In practice, it may be useful to include forwards struck at the barrier in the portfolio: e.g. it may decrease the sensitivity of an RTIS portfolio to changes in underlying. Propositions \ref{prop:RTSIS1}, \ref{prop:RTSIS2} and \ref{prop:RTSIS3} remain true in this case.}
%The optimality of the proposed RTIS portfolio is summarized in the following proposition.

\begin{prop}\label{prop:RTSIS3}
Consider a barrier $U$, and beliefs $\belief_*$, generated by $\Sigma^{\sigma_*}$, with an arbitrary (but fixed) $\sigma_*=(\sigma_{*1},\sigma_{*2}>0)$.
Assume that there exists an index $i$, s.t. $K_i<U$, $\underline{j}^{\sigma_*}(i)$ is well defined via (\ref{eq.underj.def}), and, moreover, $K_{\underline{j}^{\sigma_*}(i)} = \mathbf{K}^{\sigma_*,U}(K_i)$.
Consider any weights $\alpha,\beta\in\RR$ and any index $j$, satisfying: 
\begin{itemize}
\item $K_j>U$;
\item at any time $t$, at which $(S_t,\Sigma^{mkt}_t)$ is admissible, $S_t=U$, and $\Sigma^{mkt}_t\in\belief_*(T)$, with some $T>0$, we have
\begin{equation}\label{eq.RTIS.cand.def}
\alpha C^{mkt}_t(K_{j},\tau) + \beta P^{mkt}_t(K_{i},\tau)>0,\quad \forall \tau\in(0,T];
\end{equation}
\item at any time $t$, at which $(S_t,\Sigma^{mkt}_t)$ is admissible, $S_t=U$, and $\Sigma^{mkt}_t\in\belief^*(T)$, with some $T>0$ and with $\belief^*$ generated by $\Sigma^{\sigma^*}$, for some $\sigma^*=(\sigma^*_1,\sigma^*_2>0)$, s.t. $\sigma^*_2/\sigma^*_1 < \sigma_{*2}/\sigma_{*1}$, we have
%$$
%K_{\underline{j}^{\sigma_*}(i)-1}<\mathbf{K}^{\sigma^*,U}(K_i)
%<\mathbf{K}^{\sigma_*,U}(K_i)=K_{\underline{j}^{\sigma_*}(i)}
%$$ 
%there exists $T'\in(0,T]$, s.t.
\begin{equation}\label{eq.RTIS.cand.def.2}
\alpha C^{mkt}_t(K_{j},\tau) + \beta P^{mkt}_t(K_{i},\tau)<0,
\quad \forall \tau\in(0,T'],
\end{equation}
with some $T'\in(0,T]$.
\end{itemize}
Then, we must have: $j=\underline{j}^{\sigma_*}(i)$ and $\alpha\geq-\beta>0$.
\end{prop}

\begin{remark}
The restriction $K_{\underline{j}^{\sigma_*}(i)} = \mathbf{K}^{\sigma_*,U}(K_i)$ simply means that $K_i$ and $\sigma_{*2}/\sigma_{*1}$ are such that $\mathbf{K}^{\sigma_*,U}(K_i) = U + (U-K_i)\sigma_{*2}/\sigma_{*1}$ belongs to the grid of available strikes. This can be guaranteed by restricting the values of $\sigma_{*2}/\sigma_{*1}$ to a finite grid. Alternatively, this restriction can be viewed as an assumption on the richness of the grid of available strikes, above the barrier.
\end{remark}

\begin{proof}
First, we show that $\alpha\geq0$ and $\beta\leq 0$. Assume that $\alpha<0$. Then, Proposition \ref{prop:IV.asymp} implies that we can choose $\sigma=(\sigma_1,\sigma_2)$, with $\sigma_2=\sigma_{*2}$ and with small enough $\sigma_1>0$, to ensure that $\Sigma^{\sigma}(U,K_i,0+)$ is arbitrarily small, while $\Sigma^{\sigma}(U,K_j,0+)=\Sigma^{\sigma_*}(U,K_j,0+)$.
Proposition \ref{prop:IV.asymp} also implies that, for a small enough $\sigma_1>0$, there exists $T>0$, s.t. $\Sigma^{\sigma}\in\belief_*(T)$.
In addition, (\ref{eq.exactSH.1}) implies
\begin{equation*}
P^{\sigma}(U,K_i,\tau) < C^{\sigma}(U,\mathbf{K}^{\sigma,U}(K_i),\tau),
\quad \forall\,\tau>0,
\end{equation*}
with $\mathbf{K}^{\sigma,U}$ defined in (\ref{eq.ESH.K.def}).
Choosing $\sigma_1>0$ to be sufficiently small, we ensure that $\mathbf{K}^{\sigma,U}(K_i) > K_j$.
Then, Lemma \ref{le:Csigma.def} implies that
\begin{equation*}
C^{\sigma}(U,K_j,\tau)
\geq \left|\frac{\beta}{\alpha}\right| C^{\sigma}(U,\mathbf{K}^{\sigma,U}(K_i),\tau)
\geq \left|\frac{\beta}{\alpha}\right| P^{\sigma}(U,K_i,\tau),
\quad \forall\,\tau\in(0,T],
\end{equation*}
with a, possibly, different $T>0$, which contradicts (\ref{eq.RTIS.cand.def}).
Similarly, we obtain a contradiction if $\beta>0$.

Notice that, if $\beta=0$, then we must have $\alpha>0$, due to (\ref{eq.RTIS.cand.def}), but, in this case, (\ref{eq.RTIS.cand.def.2}) fails to hold.
Thus, rescaling the portfolio, we conclude that it suffices to set $\beta=-1$ and prove that $K_j = K_{\underline{j}^{\sigma_*}(i)}$ and $\alpha\geq1$.
%\begin{equation*}
%\gamma C^{mkt}(U,K_{j},\tau) - P^{mkt}(U,K_{i},\tau)>0,\quad \forall \tau\in(0,T'),
%\end{equation*}
%Note that, if $\alpha>1$, then, the statement also follows trivially. Hence, we assume that $\alpha\in(0,1)$ and $\beta=-1$.

Assume that $K_j> K_{\underline{j}^{\sigma_*}(i)}=\mathbf{K}^{\sigma_*,U}(K_i)$.
Notice that $\Sigma^{\sigma_*}\in\belief_*(T)$ (for any $T>0$). 
Then (\ref{eq.exactSH.1}) and Lemma \ref{le:Csigma.def} imply that there exists $T'>0$, s.t.
\begin{equation*}
P^{\sigma_*}(U,K_i,\tau) \geq \frac{K_i}{U} C^{\sigma_*}(U,\mathbf{K}^{\sigma_*,U}(K_i),\tau)
\geq \alpha C^{\sigma_*}(U,K_j,\tau),
\quad \forall\,\tau\in(0,T'],
\end{equation*}
which contradicts (\ref{eq.RTIS.cand.def}).

Assume that $K_j = K_{\underline{j}^{\sigma_*}(i)}=\mathbf{K}^{\sigma_*,U}(K_i)$. 
Then, (\ref{eq.ESH.putViaCalls}) and Lemma \ref{le:Csigma.def} yield:
$$
P^{\sigma_*}(U,K,\tau) = C^{\sigma_*}\left(U,K_j,\tau\right)(1+\overline{\overline{o}}(1)),
\quad\tau\rightarrow0,
$$
which, in view of (\ref{eq.RTIS.cand.def}), implies that $\alpha\geq1$. 

Assume that $K_j < K_{\underline{j}^{\sigma_*}(i)}=\mathbf{K}^{\sigma_*,U}(K_i)$ and consider any $\sigma^*$ as in the statement of the proposition. 
Notice that $\Sigma^{\sigma^*}\in\belief^*(T)$ (for any $T>0$).
Notice also that we can choose $\sigma^*$ to be close to $\sigma_*$, so that $K_j<\mathbf{K}^{\sigma^*,U}(K_i)< \mathbf{K}^{\sigma_*,U}(K_i)$.
Then, (\ref{eq.exactSH.1}) and Lemma \ref{le:Csigma.def} imply that there exists $T'>0$, s.t.
\begin{equation*}
P^{\sigma^*}(U,K_i,\tau) \leq C^{\sigma^*}(U,\mathbf{K}^{\sigma^*,U}(K_i),\tau)
\leq \alpha C^{\sigma^*}(U,K_j,\tau),
\quad \forall\,\tau\in(0,T'],
\end{equation*}
which contradicts (\ref{eq.RTIS.cand.def.2}).
\end{proof}

Propositions \ref{prop:RTSIS1} and \ref{prop:RTSIS2} imply that $j=\underline{j}^{\sigma_*}(i)$ and any $\alpha=-\beta>0$ satisfy the three properties stated in Proposition \ref{prop:RTSIS3} and, hence, yield an admissible RTIS portfolio, given by (\ref{eq.RTIS.def}). This portfolio contains a \emph{minimal} number of call shares (i.e. $|\alpha|$), for a given number of put shares (i.e. $|\beta|$). Since $|\beta|$ can be viewed as a simple normalization constant, we can set its value to one and obtain the proposed RTIS portfolio, given by (\ref{eq.RTIS.def}). This observation demonstrates the minimality of the proposed portfolio and, thus, justifies our choice.

\begin{remark}
The portfolio given by (\ref{eq.RTIS.def}) is constructed in such a way that the strike of a put option is chosen a priori, and the appropriate strike of a call option is computed via the dominating model. Clearly, one can also start by prescribing the strike of a call option, and obtain the corresponding strike of a put via the derivations preceding Proposition \ref{prop:RTSIS1}, using (\ref{eq.exactSH.2}) rather than (\ref{eq.exactSH.1}). This may result in a different portfolio, although the difference should become negligible as the grid of available strikes becomes sufficiently rich. It is straightforward to formulate the analogues of Propositions \ref{prop:RTSIS1}, \ref{prop:RTSIS2}, \ref{prop:RTSIS3} for the new portfolio.
\end{remark}

\begin{remark}\label{rem:RTIS.upperBeliefs}
The strategy described by Propositions \ref{prop:RTSIS1} and \ref{prop:RTSIS2} is based on the expectation that the implied skew will always return above a chosen lower threshold (before the underlying drops below the barrier). Such a strategy is designed to benefit from the abnormal deviations of the implied skew {\bf below} its typical values. Roughly speaking, a position is opened when the skew drops below a critical level, and it is closed when the skew returns above this level. Similarly, one can design a strategy that benefits from the deviations of the implied skew {\bf above} its typical values. Namely, it is straightforward to formulate the analogues of Propositions \ref{prop:RTSIS1} and \ref{prop:RTSIS2}, which show that the portfolio
\begin{equation}\label{eq.RTIS.upper.def}
P^{mkt}(K_i,\tau) - \frac{K_i}{U} C^{mkt}(K_{\overline{j}^{\sigma^*}(i)},\tau)
\end{equation}
has a positive price when the implied volatility satisfies $\belief^*$, generated by $\Sigma^{\sigma^*}$ (and the underlying stays below the barrier), and its price becomes negative when the implied skew is sufficiently large (and the underlying is at the barrier).
\end{remark}

\section{Robust Hedging of a UOP Option with Beliefs on Implied Skewness}
\label{se:alg}

The second RTIS portfolio, given by (\ref{eq.RTIS.upper.def}), can also be used to construct a robust super-replication strategy of a UOP option, given beliefs on the implied skew. A simple example of such a strategy is presented in Subsection \ref{se:uop} as our motivating case study. We now extend that reasoning to allow for more general beliefs.
Assume that we need to hedge a UOP option with strike $K>0$ (which may not belong to the grid of available strikes for the European options), barrier $U\geq K\vee S_0$, and maturity $T>0$. The payoff of this option at time $T$ is given in (\ref{eq.UOP.payoff}).
%$$
%(K-S_T)^+ \bone_{\{\sup_{t\in[0,T]}S_t < U\}}.
%$$
%At the initial time $t=0$, the market provides current prices of European call and put options with the same maturity $T$ and strikes $K_1,\ldots,K_N$ (put-call parity holds). 
%Assume that, for every $K_i<U$, we have formed beliefs on the implied skewness with respect to $K_i$, given by $\{a^i_j,b^i_j\}_j$, as described in Subsection \ref{se:beliefs}. In addition, for every $K_i<U$, we have chosen the upper and lower dominating models, $M^{*i}$ and $M^i_*$, given by the associated pairs $(\sigma^{*i}_1,\sigma^{*i}_2)$ and $(\sigma^i_{*1},\sigma^i_{*2})$.
Below, we recall the algorithm presented in \cite{BHR}, which provides a robust semi-static super-replication strategy for the UOP option with the available co-maturing European options (and forwards), given that the underlying has continuous paths. We refer to this strategy as BHR, following the names of its authors.
\begin{enumerate}
\item Choose any $K_i\leq K$ from the grid of available strikes, and consider a portfolio that consists of 
\begin{itemize}
\item $(K_i-K)/(U-K_i)$ shares of a forward contract struck at $U$ and maturing at $T$, 
\item and $(U-K)/(U-K_i)$ shares of a co-maturing European put option struck at $K_i$. 
\end{itemize}

\item Iterate the construction in step 1, over all admissible indices $i$, to find the cheapest portfolio, according to the current market prices.

\item If there exists no $K_i\leq K$ in the grid of available strikes, skip steps 1 and 2 and buy one share of a European put with the smallest available strike.

\item Liquidate this portfolio at the time $\hitu\wedge T$, where 
\begin{equation}\label{eq.Tu.def}
\hitu=\inf\{t\in[0,T]:\,S_t \geq U\},
\end{equation}
with $\inf\emptyset=\infty$.
\end{enumerate}

The payoff function of the BHR strategy is shown in Figure \ref{fig:BHR}.
It is easy to see that, if $\hitu\geq T$, then, the liquidation price of the above portfolio at $T$ is simply the payoff of the initial portfolio, which dominates the payoff of the UOP option, given by $(K-S_T)^+$. If $\hitu<T$, then, the payoff of the UOP option becomes zero, while the market price of the BHR portfolio, at the time $\hitu$, is equal to the price of $1 - (K-K_i)/(U-K_i)$ shares of put options, which is nonnegative, for any admissible implied volatility surface. Thus, the payoff of this portfolio is never smaller than the payoff of the UOP option. 
It is also shown in \cite{BHR} that there exists a martingale model for the underlying, which is consistent with the currently observed market prices of European options, and in which the payoff of the BHR portfolio is a.s. equal to the payoff of the UOP option. The latter demonstrates that one cannot super-replicate this barrier option at a cheaper price, without adding more assumptions on the market.

At the same time, the ``extremal" martingale model used in \cite{BHR} is quite irregular, and may not be consistent with the typical behavior of the underlying and the prices of European options. This leads one to believe that it is natural to introduce further assumptions (i.e. beliefs) on the market, which would allow one to construct more efficient super-replication strategies and to reduce the super-replication price.
Herein, we construct an improvement of the BHR strategy, which super-replicates the payoff of the UOP option, provided the beliefs on the implied skewness are satisfied. The price of this strategy never exceeds the price of BHR, and, if the grid of available strikes above the barrier is sufficiently rich, the price of our strategy is strictly lower than the price of BHR.

Consider any beliefs $\belief^*$ on the implied skewness generated by $\Sigma^{\sigma^*}$, with some $\sigma^*=(\sigma^*_{1},\sigma^*_{2}>0)$.
Below, we present the algorithm for robust static super-replication of the UOP option with the available co-maturing European options (and forwards), given that the underlying has continuous paths and that the beliefs $\belief^*$ are satisfied whenever $S_t=U$.

\begin{enumerate}
\item Choose any $K_i\leq K$ from the grid of available strikes, and consider a portfolio that consists of 
\begin{itemize}
\item $(K_i-K)/(U-K_i)$ shares of a forward contract struck at $U$ and maturing at $T$, 
\item $(U-K)/(U-K_i)$ shares of a co-maturing European put option struck at $K_i$,
\item and $(-(U-K)/(U-K_i))\frac{K_i}{U}$ shares of a co-maturing European call struck at $K_{\overline{j}^{\sigma^*}(i)}$, provided $\overline{j}^{\sigma^*}(i)$ is well defined via (\ref{eq.overj.def}) (otherwise, skip this part).
\end{itemize}

\item If there exists no $K_i\leq K$ in the grid of available strikes, skip step 1, and
\begin{itemize}
\item buy one share of a European put with strike $K_i<U$; 
\item short $K_i/U$ shares of a co-maturing European call with strike $K_{\overline{j}^{\sigma^*}(i)}$, if $\overline{j}^{\sigma^*}(i)$ is well defined via (\ref{eq.overj.def}) (otherwise, skip this part).
\end{itemize}

\item Iterate the construction in steps 1 and 2, over all admissible indices $i$, to find the cheapest portfolio, according to the current market prices.

\item Liquidate this portfolio at the time $\hitu\wedge T$.
\end{enumerate}

The payoff function of the proposed super-replication strategy is shown in Figure \ref{fig:BHR.mod}.
The above strategy differs from the BHR in the parts which add extra short positions in a call option to the portfolio. It is easy to see that the above strategy super-replicates the payoff of the UOP option. Indeed, if the barrier is never hit, the payoff of the above strategy dominates the payoff of the associated put option. If the barrier is hit, then, at the time $\hitu$, the forward price is zero and the price of the remaining portfolio is given by (\ref{eq.RTIS.upper.def}), scaled by a positive constant. The latter portfolio has a positive price, due to (\ref{eq.exactSH.1}), as the beliefs $\belief^*$ are satisfied.

\begin{prop}\label{prop:RSH}
Consider a maturity $T>0$, a barrier $U\geq S_0 >0$, a strike $K\in(0,U]$, and the beliefs $\belief^*$, generated by $\Sigma^{\sigma^*}$, with some $\sigma^*=(\sigma^*_1,\sigma^*_2>0)$.
Assume that there exists an index $i$, s.t. $K_i\leq K$ and $\overline{j}^{\sigma^*}(i)$ is well defined via (\ref{eq.overj.def}). Then for any admissible path of $(S,\Sigma^{mkt})$, which further satisfies: 
\begin{itemize}
\item $t\mapsto S_t$ is continuous,
\item and at any time $t\in[0,T)$, at which $S_t=U$, we have $\Sigma^{mkt}_t\in\belief^*(T-t)$,
\end{itemize}
the payoff of the proposed super-replicating strategy is at least as large as the payoff of the UOP option, i.e.
$$
\frac{K_i-K}{U-K_i} (S_{\hitu\wedge T} - U) 
+ \left(1 - \frac{K-K_i}{U-K_i}\right) 
\left(P^{mkt}_{\hitu\wedge T}(K_i,T-\hitu\wedge T) - \frac{K_i}{U} C^{mkt}_{\hitu\wedge T}(K_{\overline{j}^{\sigma^*}(i)},T-\hitu\wedge T) \right)
$$
\begin{equation}\label{eq.RSH.main}
\geq \left(K-S_{\hitu\wedge T}\right)^+ \bone_{\{\hitu > T\}},
\end{equation}
where $\hitu$ is defined in (\ref{eq.Tu.def}).
\end{prop}

\begin{remark}
If there exists no index $i$, s.t. $K_i\leq K$, but there exists $i$, s.t. $K_i<U$ and $\overline{j}^{\sigma^*}(i)$ is well defined via (\ref{eq.underj.def}), then, step 2 of the super-replication strategy becomes relevant, and the inequality (\ref{eq.RSH.main}) changes to
$$
P^{mkt}_{\hitu\wedge T}(K_i,T-\hitu\wedge T) - \frac{K_i}{U} C^{mkt}_{\hitu\wedge T}(K_{\overline{j}^{\sigma^*}(i)},T-\hitu\wedge T)
\geq \left(K-S_{\hitu\wedge T}\right)^+ \bone_{\{\hitu > T\}}.
$$
The above inequality is also satisfied on the set of paths defined in Proposition \ref{prop:RSH}, which follows easily form the proof of the proposition.
\end{remark}

\begin{proof}
If $\hitu\geq T$, then (\ref{eq.RSH.main}) turns into
$$
\frac{K_i-K}{U-K_i} (S_{T} - U) + \left(1 - \frac{K-K_i}{U-K_i}\right) (K_i - S_T)^+ 
\geq \left(K-S_{T}\right)^+,
$$
which, clearly, holds for any $S_T\in[0,U]$. 
Assume that $\hitu < T$. Then, we need to show that
$$
\left(1 - \frac{K-K_i}{U-K_i}\right) 
\left(P^{mkt}_{\hitu}(K_i,T-\hitu) - \frac{K_i}{U} C^{mkt}_{\hitu}(K_{\overline{j}^{\sigma^*}(i)},T-\hitu) \right)
\geq 0.
$$
Multiplying $\sigma^*$ by a positive constant, we obtain $\tilde{\sigma}=(\tilde{\sigma}_1,\tilde{\sigma}_2)$, such that $\tilde{\sigma}_2/\tilde{\sigma}_1=\sigma^*_{2}/\sigma^*_{1}$ and
\begin{equation*}
\Sigma^{\tilde{\sigma}}(U,K_i,T-\hitu) \leq \Sigma^{mkt}_{\hitu}(K_i,T-\hitu),
\quad  \Sigma^{\tilde{\sigma}}(U,K_j,T-\hitu) \geq \Sigma^{mkt}_{\hitu}(K_j,T-\hitu),
\quad \forall K_j>U.
\end{equation*}
Then, we have:
\begin{equation*}
P^{\tilde{\sigma}}(U,K_i,T-\hitu)
\leq P^{mkt}_{\hitu}(K_i,T-\hitu),
\quad
C^{\tilde{\sigma}}(U,K_{\overline{j}^{\sigma^*}(i)},T-\hitu)
\geq C^{mkt}_{\hitu}(K_{\overline{j}^{\sigma^*}(i)},T-\hitu).
\end{equation*}
In addition,
$$
C^{\tilde{\sigma}}(U,\mathbf{K}^{\tilde{\sigma},U}(K_i),T-\hitu)
= C^{\tilde{\sigma}}(U,\mathbf{K}^{\sigma^*,U}(K_i),T-\hitu)
\geq C^{\tilde{\sigma}}(U,K_{\overline{j}^{\sigma^*}(i)},T-\hitu).
$$
The above inequalities, along with (\ref{eq.exactSH.1}), yield
$$
P^{mkt}_{\hitu}(K_i,T-\hitu) - \frac{K_i}{U} C^{mkt}_{\hitu}(K_{\overline{j}^{\sigma^*}(i)},T-\hitu) > 0
$$
and complete the proof.
\end{proof}

Clearly the initial capital of the proposed super-replicating strategy never exceeds the initial capital of BHR, and is strictly lower if $\overline{j}^{\sigma^*}(i)$ is well defined and $C^{mkt}_{0}(K_{\overline{j}^{\sigma^*}(i)},T)>0$.

\begin{remark}
Proposition 3.4 in \cite{BHR} also provides a robust static sub-replication strategy for a UOP option. In the cases where the latter strategy is simply zero (which may occur, depending on the initial market implied volatility), one can use the ideas presented in this section to construct a similar improvement of this strategy, using beliefs $\belief_*$. Namely, the improved strategy would consist in short-selling the first RTIS portfolio, given by (\ref{eq.RTIS.def}), with $K_i\leq K < U$, and closing this position at $\hitu$ (if $\hitu<T$).
\end{remark}

\section{Empirical Study of RTIS Portfolios}
\label{se:emp}

In this section, we compute the empirically observed prices of the RTIS portfolios (\ref{eq.RTIS.def}) and (\ref{eq.RTIS.upper.def}) and demonstrate that they have the predicted signs. In order to construct these portfolios, we need to choose the associated $\sigma_*=(\sigma_{*1},\sigma_{*2})$ and $\sigma^*=(\sigma^*_1,\sigma^*_2)$. In fact, to construct the portfolios, we only need to know the skewness parameters 
$$
\kappa_* = \sigma_{*2}/\sigma_{*1},
\quad \kappa^* = \sigma^*_{2}/\sigma^*_{1}.
$$
Assume that the strike of an OTM put, $K_i$, is fixed.
Then, according to the definition of beliefs generated by dominating surfaces, the skewness parameters $(\kappa_*,\kappa^*)$ should to be chosen so that, at any time $t$, there exist $\sigma_{*1}(t)$ and $\sigma^*_1(t)$, such that $\sigma_*(t)=(\sigma_{*1}(t),\kappa_*\sigma_{*1}(t))$ and $\sigma^*(t)=(\sigma^*_{1}(t),\kappa^*\sigma^*_{1}(t))$ satisfy:
\begin{equation}\label{eq.emp.IVconstraint.1}
\Sigma^{\sigma_*(t)}(S_t,K_i,\tau) \geq \Sigma^{mkt}_t(K_i,\tau),\quad \Sigma^{\sigma_*(t)}(S_t,K_j,\tau) \leq \Sigma^{mkt}_t(K_j,\tau),
\end{equation}
\begin{equation}\label{eq.emp.IVconstraint.2}
\Sigma^{\sigma^*(t)}(S_t,K_i,\tau) \leq \Sigma^{mkt}_t(K_i,\tau),\quad \Sigma^{\sigma^*(t)}(S_t,K_j,\tau) \geq \Sigma^{mkt}_t(K_j,\tau),\quad \forall\, K_j>U.
\end{equation}
In the above, we use the fact that $\Sigma^{c\sigma}(S_t,K,\tau) = c\Sigma^{\sigma}(S_t,K,c^2\tau)$.
In theory, the above inequalities have to hold for all $\tau>0$. However, in our empirical analysis, we only require them to hold for the smallest three maturities available in the market (in our sample, they never exceed 2 months). 

Note that (\ref{eq.emp.IVconstraint.1})--(\ref{eq.emp.IVconstraint.2}) are written on a directly observable market quantity -- i.e. the implied volatility surface -- which is one of the main advantages of the proposed family of beliefs. However, as we have restricted the family of dominating surfaces to those produced by PCLVG models, we, essentially, need to calibrate $\kappa_*$ and $\kappa^*$ to a given sample (i.e. time series) of market implied volatility surfaces, $\{\Sigma^{mkt}_t\}$. 
The most direct way to do this is to solve two constrained optimization problems -- for the variables $(\kappa_*,\{\sigma_{*1}(t)\})$ and $(\kappa^*,\{\sigma^*_{1}(t)\})$ -- in which the constraints are given by (\ref{eq.emp.IVconstraint.1})--(\ref{eq.emp.IVconstraint.2}), with $\sigma_{*2}(t)=\sigma_{*1}(t)\kappa_*$ and $\sigma^*_{2}(t)=\sigma^*_{1}(t)\kappa^*$, and the objectives can be chosen as ``$\max \kappa_{*}$" and ``$\min \kappa^*$", respectively.
%Then, the skewness parameters can be chosen as follows:
%$$
%\kappa_* = \min_t \sigma_{*2}(t)/\sigma_{*1}(t),
%\quad \kappa^* = \max_t \sigma^*_{2}(t)/\sigma^*_{1}(t).
%$$
Nevertheless, herein, we choose a slightly different approach. Namely, at any time $t$, and for any time to maturity $\tau$, we find $\sigma(t,\tau)=(\sigma_1(t,\tau),\sigma_2(t,\tau))$, s.t. $\Sigma^{\sigma}(S_t,\cdot,\tau)$ matches the observed market implied volatility, $\Sigma^{mkt}_t(\cdot,\tau)$, as closely as possible.\footnote{As an objective, we choose to minimize the sum of squared differences between the market prices of European options and those produced by the PCLVG model.} Then, the skewness parameters $\kappa_*$ and $\kappa^*$ are defined, respectively, as the $\alpha$- and $(1-\alpha)$-quantiles of the sample $\{\sigma_{2}(t,\tau)/\sigma_{1}(t,\tau)\}_{t,\tau}$, with some small $\alpha\in[0,1]$.\footnote{Herein, we choose $\alpha=0.01$.}
%$$
%\kappa_* = \min_{t,\tau} \sigma_{2}(t,\tau)/\sigma_{1}(t,\tau),
%\quad \kappa^* = \max_{t,\tau} \sigma_{2}(t,\tau)/\sigma_{1}(t,\tau),
%$$
Such approach allows us to decrease the dimension of the optimization variable significantly and to speed up the computations. In addition, it allows us to see how well a PCLVG models can fit the market implied smile, for a single maturity. The downside of this approach is the lack of a guarantee that the resulting $(\kappa_*,\kappa^*)$ admit $(\sigma_*(t),\sigma^*(t))$ satisfying (\ref{eq.emp.IVconstraint.1})--(\ref{eq.emp.IVconstraint.2}), for all $(t,\tau)$ in the sample. The latter may occur if the wings of a PCLVG implied volatility, $\Sigma^{\sigma}(S_t,K,\tau)$, are not monotone w.r.t. each of the two variables, $\sigma_1$ and $\sigma_2$. Namely, it does not seem to be true, in general, that increasing $\sigma_1$ will increase the values of $\Sigma^{\sigma}(S_t,K,\tau)$ for $K<S_t$, while keeping the other values unchanged, and that increasing $\sigma_2$ will increase the values of $\Sigma^{\sigma}(S_t,K,\tau)$ for $K>S_t$, while keeping the other values unchanged. However, such a monotonicity does hold for $\tau\approx0$, as shown by Proposition \ref{prop:IV.asymp}. This observation motivates our approach, which turns out to work well empirically.

Thus, in order to test the empirical performance of RTIS portfolios, we need to calibrate PCLVG models to a historical sample of European options' prices. We use the daily SPX (S\&P 500) European options' prices provided by OptionMetrics, an option database containing historical prices of options and their underlying instruments, in the time period from January 3, 2011, through December 27, 2012. 
On each day of the period, we only keep the options whose best closing bid price and best closing offer price are both available, and take the average of the two prices as the option price. To ensure the validity of all prices, the contracts with zero open interest are excluded.
%As a result, there are roughly 10 to 80 call contracts with valid prices available for each maturity. The log-moneyness (more precisely, the put log-moneyess, defined as $\log(K/S_t)$) of these call options ranges roughly from -0.3 to 0.1, varying for different $t$ and $T$.
Our calibration also requires dividend and interest rate data available on OptionMetrics and the homepage of U.S. Department of Treasury, respectively.
%This dividend yield is recovered from option prices via put-call parity with the method proposed in \cite{AitSahalia1997nonparametricestimation}. 
%On day $t$, we denote the dividend yield by $q_t$, and the risk-free rate between $t$ and $T$ by $r_{t,T}$.
Using these dividend and interest rates, we perform a standard transformation (i.e. discounting) of the market prices and strike values, so that the resulting options' prices are consistent with the assumption of zero interest and dividend rates.
%from now on:
%\begin{align} \label{fo:dtl:transform}
%C^{mkt}_t(T,x) & = e^{q_t(T-t)} \bar{C}_t^{mkt}(T,\bar{x}), \quad \text{with} \quad x = \bar{x} - (r_{t,T}-q_t)(T-t),
%\end{align}
%where $\bar{C}_t^{mkt}(T,\bar{x})$ is the market price of a call option with maturity $T$ and strike $e^{\bar{x}}$.
%The adjusted call prices $C^{mkt}_t(T,x)$, corresponding to maturity $T$ and strike $e^x$, are then consistent with the assumption of zero interest and dividend rates (i.e. they do not contain arbitrage under thee assumptions). In a similar way, we define the adjusted bid and ask prices, $C^{mkt,b}_t$ and $C^{mkt,a}_t$.
Finally, we only keep the options with the three smallest maturities in our sample, and compute the market implied volatilities using OTM options only (i.e. put prices for small strikes and call prices for large strikes). For any day $t$ in the given sample, and any time-to-maturity $\tau$ (among the smallest three available), we calibrate a PCLVG model to the market implied smile by minimizing the sum of squared differences between the market prices of the options and the prices produced by PCLVG model. This calibration produces the parameters $\sigma(t,\tau)=(\sigma_1(t,\tau),\sigma_2(t,
\tau))$.

The results of the calibration are presented in Figures \ref{fig:emp.1}--\ref{fig:emp.5}. In particular, Figure \ref{fig:emp.1} provides examples of market implied smiles and the calibrated PCLVG smiles. We can see that the fit is fairly good (especially given that the model has only two parameters), but not perfect. Nevertheless, the important thing is that a PCLVG model captures the market implied skewness sufficiently well, as follows from the subsequent results.
The calibrated parameters, $\{\sigma_1(t,\tau),\sigma_2(t,\tau)\}$, and the associated skewness sample, $\{\sigma_2(t,\tau)/\sigma_2(t,\tau)\}$, are shown in Figure \ref{fig:emp.2}. From the left hand side of the figure, we can see that much of the variation in the calibrated parameters can be explained by a simple linear dependence. The right hand side of the figure confirms this observation, showing that the   skewness $\sigma_2(t,\tau)/\sigma_1(t,\tau)$ has a much smaller range of values, and exhibits a mean-reverting behavior.\footnote{Strictly speaking, to argue that the ratio has a mean-reverting behavior, one has to consider a subsample of the calibrated skewness parameters, for a single time-to-maturity per day, which can be extracted form the total sample, shown in Figure \ref{fig:emp.2}. Even though we do not plot it herein, we have verified that such a subsample also exhibits a mean-reverting behavior.}

The lower and upper skewness parameters, $\kappa_*$ and $\kappa^*$, respectively, are chosen as the $\alpha$- and $(1-\alpha)$-quantiles of the empirical distribution of $\{\sigma_2(t,\tau)/\sigma_1(t,\tau)\}$, with $\alpha=0.01$. The resulting values are: $\kappa_* = 0.3777$ and $\kappa^* = 0.8949$.
It turns out that such a choice of $(\kappa_*,\kappa^*)$ allows for $\sigma_*(t)$ and $\sigma^*(t)$ satisfying (\ref{eq.emp.IVconstraint.1})--(\ref{eq.emp.IVconstraint.1}), for most days $t$. 
Figure \ref{fig:emp.1.5} illustrates this observation: it shows the lower and upper dominating surfaces, with the prescribed skewness parameters, $\kappa_*$ and $\kappa^*$, respectively, for the put strike $K_i(t)=0.95 S_t$. Note that, in principle, one can start by choosing the upper and lower dominating surfaces (i.e. the analogues of red and yellow lines in Figure \ref{fig:emp.1.5}) directly from the time series of market implied smiles, for a fixed set of times to maturity, and, then, fit a PCLVG model to each of the two surfaces, obtaining $\kappa_*$ and $\kappa^*$. Before the dominating surfaces are constructed, the market implied smiles would need to be normalized (i.e. their values and the times to maturity would need to be rescaled appropriately), so that e.g. the short-term ATM implied volatility is always equal to one, in order to get rid of the oscillations in the overall level of volatility and to highlight the values of skewness. Such an approach might be more consistent with the philosophy of robust methods, as it allows to formulate the beliefs directly from the observed values of implied volatility, avoiding the use of a model. However, we decided not to pursue this approach, due to the additional difficulties in its implementation. Namely, in order to perform the aforementioned normalization, one needs to interpolate the market implied volatility across maturities, and different choices of interpolation may lead to different results. Our method of choosing $\kappa_*$ and $\kappa^*$ was easier to implement, and it turns out that the resulting dominating surfaces work well empirically (even though this was not guaranteed a priori), as demonstrated in Figure \ref{fig:emp.1.5}. 

Figures \ref{fig:emp.3}--\ref{fig:emp.4} show the prices of the two RTIS portfolios, constructed using the skewness parameters $\kappa_*$, $\kappa^*$, as well as $\bar{\kappa}$, which is the average of $\{\kappa(t,\tau)=\sigma_2(t,\tau)/\sigma_1(t,\tau)\}$. When constructing an RTIS portfolio, the put strike is always chosen as $K_i=0.95 S_t$ (more precisely, the closest available strike value is used), and the barrier is $U=S_t$. In addition, we normalize the portfolios so that the dollar amount invested in the put option is equal to one. The left hand side of Figure \ref{fig:emp.3} shows the actual market values of the first RTIS portfolio, given by (\ref{eq.RTIS.def}), in which the call strike is computed as $K_{\underline{j}^{\sigma_*}(i)}$, with $\sigma_{*2}/\sigma_{*1}=\kappa_*$ (recall that $\underline{j}^{\sigma}$ depends only on $\kappa=\sigma_2/\sigma_1$). The portfolios are constructed for every available $(t,\tau)$ in the sample, and we can see that almost all their prices are positive, as predicted by Proposition \ref{prop:RTSIS1}. On the other hand, these pries often approach zero, indicating that the proposed RTIS portfolio is an ``efficient" way to trade the implied skew. 
The right hand side of Figure \ref{fig:emp.3} also shows the values of the first RTIS portfolio, but with the call strike computed as $K_{\underline{j}^{\bar{\sigma}}(i)}$, where $\bar{\sigma}_2/\bar{\sigma}_1=\bar{\kappa}$ is the average of $\{\kappa(t,\tau)=\sigma_2(t,\tau)/\sigma_1(t,\tau)\}$, and $\sigma(t,\tau)$ is the pair of PCLVG parameters calibrated to the market implied smile at time $t$, for the time to maturity $\tau$. As the PCLVG implied volatility corresponding to any such $\bar{\sigma}$, typically, is not a lower dominating surface (i.e. it does not satisfy (\ref{eq.emp.IVconstraint.1})), the value of the latter portfolio does not stay positive.
Figure \ref{fig:emp.4} shows the results of a similar analysis for the second RTIS portfolio, given by (\ref{eq.RTIS.upper.def}). It is worth mentioning that both graphs in Figure \ref{fig:emp.4} are further away from zero than those in Figure \ref{fig:emp.3}. In fact, this phenomenon is a consequence of the fact that the first inequality in (\ref{eq.exactSH.1}) is not asymptotically tight for small $\tau$, unlike the second inequality. This fact is discussed in Remark \ref{re:exactSH.tight}, and, in particular, it implies that the second RTIS portfolio is not as efficient as the first one for trading the implied skew.
Notice also that the sample size in the analysis of the second RTIS portfolio is slightly smaller than the one used for the first portfolio. This is due to the fact that, for some $(t,\tau)$, $K_{\overline{j}^{\sigma^*}(i)}$ becomes larger than the maximum available strike, in which case we do not construct the RTIS portfolio.\footnote{Formally, the RTIS portfolio, in this case, becomes a long position in one share of a put option, which is clearly nonnegative.} The moneyness of the call strikes used in the RTIS portfolios is plotted in Figure \ref{fig:emp.5}. It is easy to see that, in most cases, the call strike remains within few percent away from the spot.

\section{Summary}
\label{se:summary}

In this paper, we present a method of constructing a portfolio of financial assets (i.e. a trading strategy) which achieves a given objective regardless of the validity of a specific model -- i.e. a robust method. The strategy is given explicitly and depends only on one's beliefs about the future values of an observable market indicator. As such, it allows one to use the existing statistical tools to formulate the beliefs, providing a practical interpretation of the more abstract mathematical setting, in which the belies are understood as a family of probability measures.

More specifically, we choose the implied skewness as a relevant market indicator, and show how to construct a static portfolio of European options, whose price is positive whenever the implied skewness remains in the given range (i.e. the beliefs are satisfied). In order to construct the desired portfolio, we make use of two important building blocks. First, we identify a class of models for the financial market (i.e. the PCLVG models), which is rich enough to generate the so-called dominating implied volatility surfaces, and in which the options' prices and the implied volatilities can be easily computed. Second, we show that the weak reflection of a put (or call) option's payoff in this model can be computed explicitly, and that it can be efficiently estimated by convex functions from both sides. In this sense, the model-specific results are used as building blocks to develop a robust method.

We propose two potential applications of the resulting portfolios. First, they can be used to monetize one's beliefs about the future values of the implied skew: i.e. to trade the implied skew. This is similar to the well known method for trading the short-term ATM (or, the overall level of) implied volatility via a portfolio of European option, used e.g. to construct the VIX index. Notice that an analogous method for trading the implied skew is missing to date -- the present work is a step in this direction. In addition, the proposed semi-static portfolios of European options can be used to obtain robust super-replicating strategies for barrier options. These strategies improve the ones proposed in \cite{BHR}, and they succeed whenever the beliefs on implied skewness are satisfied.

Finally, we test the proposed method empirically, demonstrating how the beliefs can be constructed from a historical time series of implied volatilities and computing the prices of target portfolios. In particular, we show that it is not difficult to construct realistic beliefs, and that the resulting portfolios do have positive signs whenever the beliefs are satisfied (as predicted by the theoretical results). Our empirical analysis confirms that the proposed method is an efficient way to trade the implied skew. In addition, we demonstrate that the proposed two-parametric family of auxiliary PCLVG models can match the empirically observed implied volatilities with a surprisingly good accuracy.

\appendix

\section{Technical proofs}
\label{app:proofs}

\begin{proof}[Proof of Lemma \ref{le:Csigma.def}]
It is shown in \cite{CN.LVG} that $\chi = C^{\sigma}(U,\cdot,\tau)$ is continuously differentiable, with an absolutely continuous derivative, and it satisfies
$$
\sigma^2(K) \chi''(K) - \frac{1}{\tau^2}\chi(K) = -\frac{1}{\tau^2}(U-K)^+,
$$
along with the boundary conditions: $\chi(0^+)=U$ and $\chi(\infty)=0$. And it is a unique such function.
Due to the piecewise constant structure of $\sigma$, we search for $\chi$ in the following form:
$$
\chi(K) = (U-K)^+ 
+ c_1 \left(\exp\left(K/(\sigma_1\tau)\right) - \exp\left(-K/(\sigma_1\tau)\right) \right) \bone_{[0,U)}(K)
+ c_2 \exp\left(-K/(\sigma_2\tau)\right) \bone_{[U,\infty)}(K),
$$
where the constants $c_i$ must be chosen so that $\chi(K)$ is continuously differentiable at $K=U$. 
Let us analyze the regularity at $U$:
$$
c_1 \left(\exp\left(U/(\sigma_1\tau)\right) - \exp\left(-U/(\sigma_1\tau)\right) \right)
= c_2 \exp\left(-U/(\sigma_2\tau)\right),
$$
$$
\frac{c_1}{\sigma_1} \left(\exp\left(U/(\sigma_1\tau)\right) + \exp\left(-U/(\sigma_1\tau)\right) \right)
= -\frac{c_2}{\sigma_2} \exp\left(-U/(\sigma_2\tau)\right) 
$$
Solving the above, we obtain:
$$
c_1 =  \frac{ \sigma_2\tau }
{ \left( 1+ \frac{\sigma_2}{\sigma_1}\right) \exp\left(U/(\sigma_1\tau)\right) 
- \left( 1 - \frac{\sigma_2}{\sigma_1}\right) \exp\left(-U/(\sigma_1\tau)\right)},
$$
$$
c_2 = \frac{\sigma_2\tau}{2} \exp\left(U/(\sigma_2 \tau)\right) \left[ 1 +
\frac{ \left( 1 - \frac{\sigma_2}{\sigma_1}\right) \exp\left(U/(\sigma_1\tau)\right) 
- \left( 1+ \frac{\sigma_2}{\sigma_1}\right) \exp\left(-U/(\sigma_1\tau)\right) }
{ \left( 1+ \frac{\sigma_2}{\sigma_1}\right) \exp\left(U/(\sigma_1\tau)\right) 
- \left( 1 - \frac{\sigma_2}{\sigma_1}\right) \exp\left(-U/(\sigma_1\tau)\right)}\right],
$$
which yields the statement of the lemma.
\end{proof}

Before we present the next proof we need a technical result on the asymptotic behavior of the cumulative distribution function of a standard normal distribution. In principle, this behavior is well known, however, the specific estimate we provide for the residual term does not seem to be available in existing literature and it is important for the results that follow.

\begin{lemma}\label{le:BSIV.asymp}
Let $\mathcal{N}$ denote the cumulative distribution function of a standard normal. Then, for any $\varepsilon\in(0,1)$, as $x\rightarrow\infty$, we have:
$$
\mathcal{N}(x) = \frac{1}{\sqrt{2\pi} x} e^{-\frac{x^2}{2}}\left( 1 + \underline{\underline{O}}\left(x^{-2-\varepsilon}\right) \right)
$$
\end{lemma}
\begin{proof}[Proof of Lemma \ref{le:BSIV.asymp}]
For any $x>0$ and any $\varepsilon\in(0,1)$, we have:
$$
\frac{1}{\sqrt{2\pi} x} e^{-\frac{x^2}{2}} - \mathcal{N}(x) = \frac{1}{\sqrt{2\pi}}\int_x^{\infty} \left(\frac{y}{x} - 1\right) e^{-y^2/2} dy
= \frac{x}{\sqrt{2\pi}}\int_1^{\infty} \left(z - 1\right) e^{-x^2 z^2/2} dz,
$$
$$
\left|\int_1^{1+x^{-3/2-\varepsilon/2}} \left(z - 1\right) e^{-x^2 z^2/2} dz\right|
\leq \frac{1}{2} x^{-3-\varepsilon} e^{-x^2/2},
$$
$$
\left|\int_{1+x^{-3/2-\varepsilon/2}}^{\infty} \left(z - 1\right) e^{-x^2 z^2/2} dz\right|
\leq \int_{1+x^{-3/2-\varepsilon/2}}^{\infty} z e^{-x^2 z^2/2} dz
= -\frac{1}{x^2} \int_{1+x^{-3/2-\varepsilon/2}}^{\infty} \left(e^{-x^2 z^2/2}\right)' dz
$$
$$
= \frac{1}{x^2} e^{-x^2(1+x^{-3/2-\varepsilon/2})^2/2}
= \frac{1}{x^2} e^{-x^2/2 - x^{1/2-\varepsilon/2} - x^{-1-\varepsilon}/2}
= e^{-x^2/2} \underline{\underline{O}}\left(x^{-3-\varepsilon}\right).
$$
\end{proof}

\begin{proof}[Proof of Proposition \ref{prop:IV.asymp}]
We will only consider the case $K>U$ -- the other case is treated similarly.
All asymptotic expressions throughout the proof are understood for $\tau\downarrow0$. 
Denote by
$$
d_1(U,K,\tau) = \frac{\log U/K + \frac{1}{2}\tau \left(\Sigma^{\sigma}(U,K,\tau)\right)^2}{\Sigma^{\sigma}(U,K,\tau)\sqrt{\tau}}.
$$
We will often write $\Sigma$ in place of $\Sigma^{\sigma}(U,K,\tau)$, and $d_1$ in place of $d_1(U,K,\tau)$.

Using (\ref{eq.Csigma.largeK.def}), it is easy to see that $C^{\sigma}(U,K,\tau)\rightarrow0$, as $\tau\downarrow0$. Therefore, $\Sigma \sqrt{\tau}\rightarrow0$ and $d_1\rightarrow-\infty$, as $\tau\downarrow0$. 
By the definition of $\Sigma$, we have
$$
C^{\sigma}(U,K,\tau) = U \mathcal{N}\left(d_1(U,K,\tau)\right) - 
K  \mathcal{N}\left(d_1(U,K,\tau)-\Sigma^{\sigma}(U,K,\tau) \sqrt{\tau}\right).
$$
Using (\ref{eq.Csigma.largeK.def}), we represent the left hand side of the above as
$$
C^{\sigma} = \frac{1}{1/\sigma_1 + 1/\sigma_2} \tau \exp\left(-\frac{K-U}{\sigma_2\tau}\right) \left(1 + \overline{\overline{o}}(1)\right).
$$
Next, using Lemma \ref{le:BSIV.asymp}, we find the following asymptotic representation:
$$
U \mathcal{N}\left(d_1\right) - 
K  \mathcal{N}\left(d_1-\Sigma \sqrt{\tau}\right)
$$
$$
= \frac{1}{\sqrt{2\pi}}
\left( \frac{U}{|d_1|} e^{-d^2_1/2} \left(1 + \overline{\overline{o}}\left(|d_1|^{-2}\right)\right) - 
\frac{K}{|d_1-\Sigma\sqrt{\tau}|} e^{-(d_1-\Sigma\sqrt{\tau})^2/2} \left(1 + \overline{\overline{o}}\left(|d_1-\Sigma\sqrt{\tau}|^{-2}\right)\right) \right)
$$
$$
= \frac{\sqrt{KU}}{\sqrt{2\pi}} \frac{1}{|d_1|} \exp\left( -\frac{1}{2} \frac{(\log U/K)^2}{\Sigma^2\tau} - \frac{1}{8} \Sigma^2 \tau \right)
\left( 1 + \overline{\overline{o}}\left(|d_1|^{-2}\right) - 
\frac{|d_1|}{|d_1-\Sigma\sqrt{\tau}|} \left(1 + \overline{\overline{o}}\left(|d_1|^{-2}\right)\right) \right).
$$
Notice that
$$
\frac{d_1}{d_1-\Sigma\sqrt{\tau}} = 1 + \frac{\Sigma^2\tau}{\log U/K - \frac{1}{2} \Sigma^2\tau}
= 1 + d_1^{-2}\frac{(\log U/K + \frac{1}{2} \Sigma^2\tau)^2}{\log U/K - \frac{1}{2} \Sigma^2\tau}
= 1 + d_1^{-2} \log U/K + \overline{\overline{o}}\left(|d_1|^{-2}\right).
$$
Collecting the above, we obtain:
$$
U \mathcal{N}\left(d_1\right) - 
K  \mathcal{N}\left(d_1-\Sigma \sqrt{\tau}\right)
= \frac{\sqrt{KU}}{\sqrt{2\pi}} \frac{1}{|d_1|^3} \exp\left( -\frac{1}{2} \frac{(\log U/K)^2}{\Sigma^2\tau} - \frac{1}{8} \Sigma^2 \tau \right)
\left( 1 + \overline{\overline{o}}\left(1\right) \right).
$$
Equating the right hand side of the above to $C^{\sigma}$, we obtain
$$
\frac{\sqrt{KU} \left( 1/\sigma_1 + 1/\sigma_2 \right)}{\sqrt{2\pi}} \frac{1}{\tau|d_1|^3} \exp\left( -\frac{1}{2} \frac{(\log U/K)^2}{\Sigma^2\tau} - \frac{1}{8} \Sigma^2 \tau + \frac{K-U}{\sigma_2\tau} \right)
= 1 + \overline{\overline{o}}\left(1\right),
$$
%$$
%\exp\left( -\frac{1}{2} \frac{(\log U/K)^2}{\Sigma^2\tau} + \frac{K-U}{\sigma_2\tau}
%- 3 \log |d_1| - \log \tau 
%- \frac{1}{8} \Sigma^2 \tau 
%+ \log \left( \frac{\sqrt{KU} \left( 1/\sigma_1 + 1/\sigma_2 \right)}{\sqrt{2\pi}}\right)
%\right)
%= 1 + \overline{\overline{o}}\left(1\right),
%$$
$$
-\frac{1}{2} \frac{(\log U/K)^2}{\Sigma^2} + \frac{K-U}{\sigma_2}
- 3\tau \log |d_1| - \tau \log \tau 
- \frac{1}{8} \Sigma^2 \tau^2 
+ \tau \log \left( \frac{\sqrt{KU} \left( 1/\sigma_1 + 1/\sigma_2 \right)}{\sqrt{2\pi}}\right)
= \overline{\overline{o}}\left(\tau\right),
$$
%$$
%-\frac{1}{2} \frac{(\log U/K)^2}{\Sigma^2} + \frac{K-U}{\sigma_2}
%+ 3\tau \log \Sigma
%+ \frac{3}{2}\tau \log \tau 
%= \overline{\overline{o}}\left(1\right),
%$$
$$
-\frac{1}{2} \frac{(\log U/K)^2}{\Sigma^2} + \frac{K-U}{\sigma_2}
+ 3\tau \log (\Sigma \sqrt{\tau})
= \overline{\overline{o}}\left(1\right),
$$
$$
\Sigma^2 = \frac{\sigma_2 (\log U/K)^2}{2(K-U)} \left( 1+ \overline{\overline{o}}\left(1\right)\right).
$$
\end{proof}

\begin{proof}[Proof of Proposition \ref{prop:exactHedge}]
Denote by $X^s$ the PCLVG process with the same $\sigma$ and with the initial condition $s$. As discussed in \cite{CN.LVG}, $\{X^s\}$ is a Markov family. Using the Markov property, we obtain
$$
\EE \left( \left.(K-X_T)^+ - g(X_T) \right| \mathcal{F}_{t\wedge \hitu} \right)
$$
$$
= \EE \left( \left. \EE\left( \left.\left((K-X_T)^+ - g(X_T)\right) \bone_{\left\{\hitu<T\right\}}\right| \mathcal{F}_{T\wedge \hitu} \right) \right| \mathcal{F}_{t\wedge \hitu} \right)
$$
\begin{equation}\label{eq.uop.sh.formalpf}
+ \EE \left( \left. \EE\left( \left.\left((K-X_T)^+ - g(X_T)\right) \bone_{\left\{\hitu\geq T\right\}}\right| \mathcal{F}_{T\wedge \hitu} \right) \right| \mathcal{F}_{t\wedge \hitu} \right)
\end{equation}
$$
= \EE \left( \left. 
\left(\EE \left(K - X^s_{\tau}\right)^+ 
- \EE g\left(X^s_{\tau} \right)\right)_{\tau = T-T\wedge \hitu,\,s=X_{T\wedge \hitu}} 
\bone_{\left\{\hitu<T\right\}}
\right| \mathcal{F}_{t\wedge \hitu} \right)
$$
$$
+ \EE \left( \left.\left(K - X_T\right)^+ \bone_{\left\{\sup_{t\in[0,T]}X_t<U\right\}} \right| \mathcal{F}_{t\wedge \hitu} \right).
$$
Notice that, on $\{\hitu<T\}$, $X_{T\wedge \hitu} = U$. Therefore, to establish the desired identity, it suffices to verify that
$$
\EE \left(K - X^U_{\tau}\right) 
= \EE g\left(X^U_{\tau}\right),\,\,\,\,\,\,\,\,\forall \tau>0,
$$
which is equivalent to
$$
C^{\sigma} \left(U,K,\tau\right) - U + K 
= \sum_{n=0}^{\infty} C^{\sigma} \left(U,U + (U(2n+1) - K)\frac{\sigma_2}{\sigma_1},\tau\right)
- C^{\sigma} \left(U,U + (U(2n+1) + K)\frac{\sigma_2}{\sigma_1},\tau\right),
$$
and, in view of (\ref{eq.Csigma.smallK.def})--(\ref{eq.Csigma.largeK.def}), the above is equivalent to
$$
\tau \exp\left(-(U-K)/(\sigma_1\tau)\right)
\frac{ 1 - \exp\left(-2K/(\sigma_1\tau)\right) }
{ \frac{1}{\sigma_1} + \frac{1}{\sigma_2} 
- \left( \frac{1}{\sigma_2} - \frac{1}{\sigma_1}\right) \exp\left(-2U/(\sigma_1\tau)\right)}
$$
$$
= \tau \frac{ \left(1 - \exp\left(-2U/(\sigma_1\tau)\right)\right) \left(\exp\left(K/(\sigma_1 \tau)\right) -  \exp\left(-K/(\sigma_1 \tau)\right)\right) }
{ \frac{1}{\sigma_1}+ \frac{1}{\sigma_2}
- \left( \frac{1}{\sigma_2} - \frac{1}{\sigma_1} \right) \exp\left(-2U/(\sigma_1\tau)\right)}
\sum_{n=0}^{\infty} \exp\left(-U(2n+1)/(\sigma_1 \tau)\right),
$$
which clearly holds.
\end{proof}

\bibliographystyle{abbrv}
\bibliography{RTIS_refs}

\newpage

\begin{figure}
\begin{center}
  \begin{tabular} {cc}
    {
    \includegraphics[width = 0.48\textwidth]{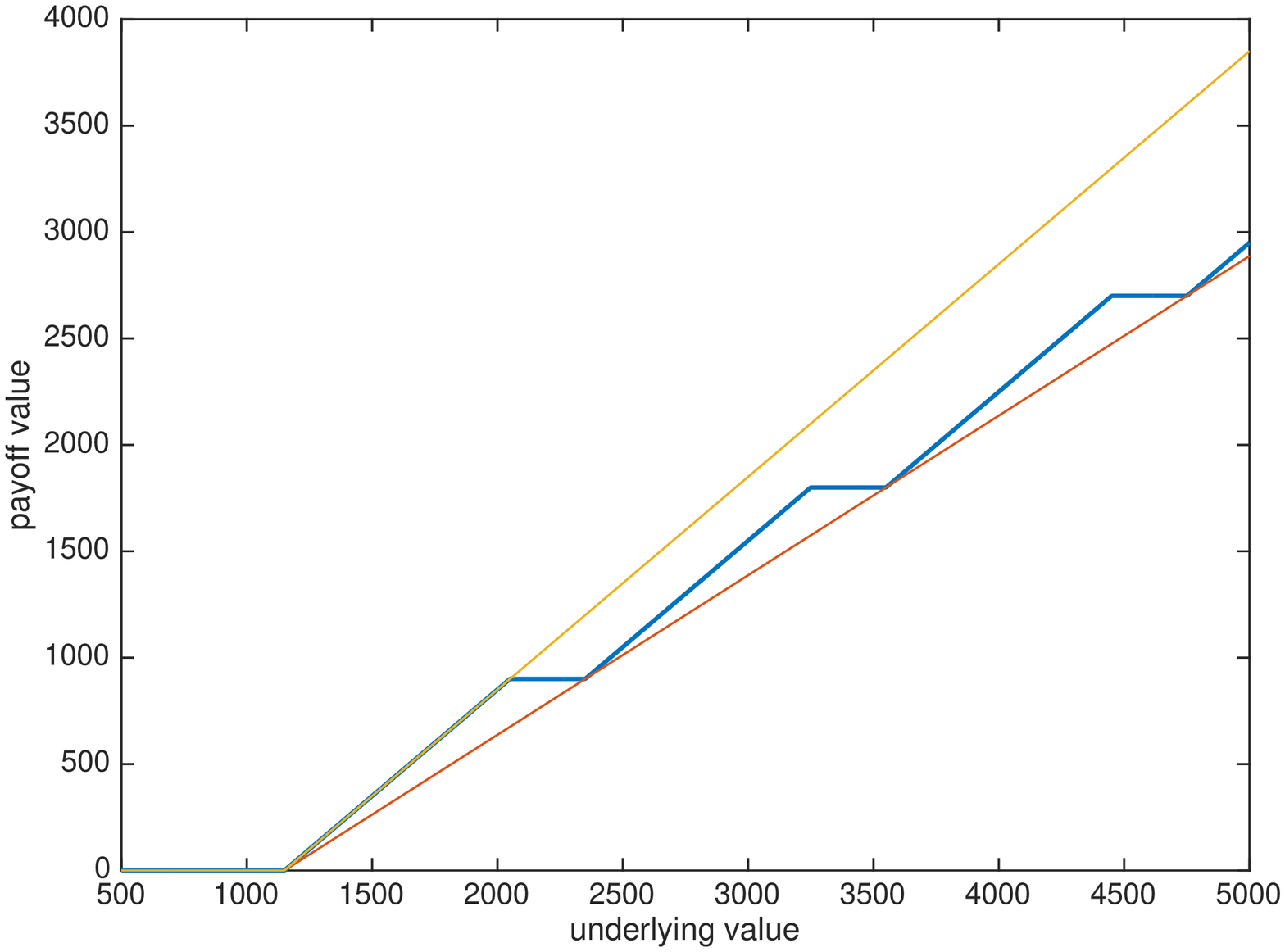}
    } & {
    \includegraphics[width = 0.48\textwidth]{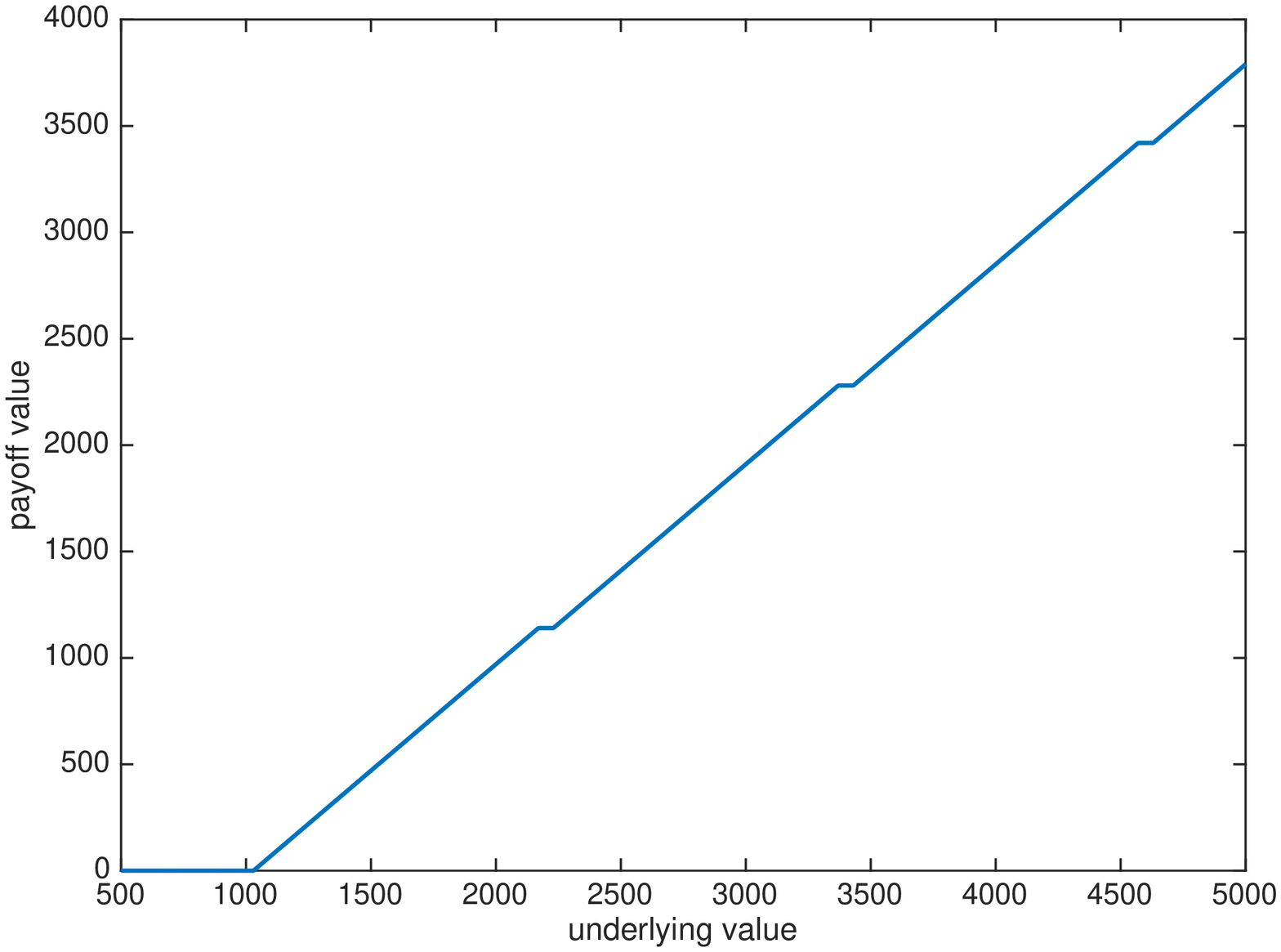}
    }  
   \end{tabular}
  \caption{The payoff functions $g$ (in blue), $\underline{g}$ (in red), and $\overline{g}$ (in yellow). The barrier is $U=1000$, and the skewness parameter is $\sigma_2/\sigma_1 = 0.6$. On the left: $K=0.75 U$. On the right: $K=0.95 U$.}
    \label{fig:g.funcs}
  \end{center}
\end{figure}

\begin{figure}
\begin{center}
  \begin{tabular} {cc}
    {
    \includegraphics[width = 0.48\textwidth]{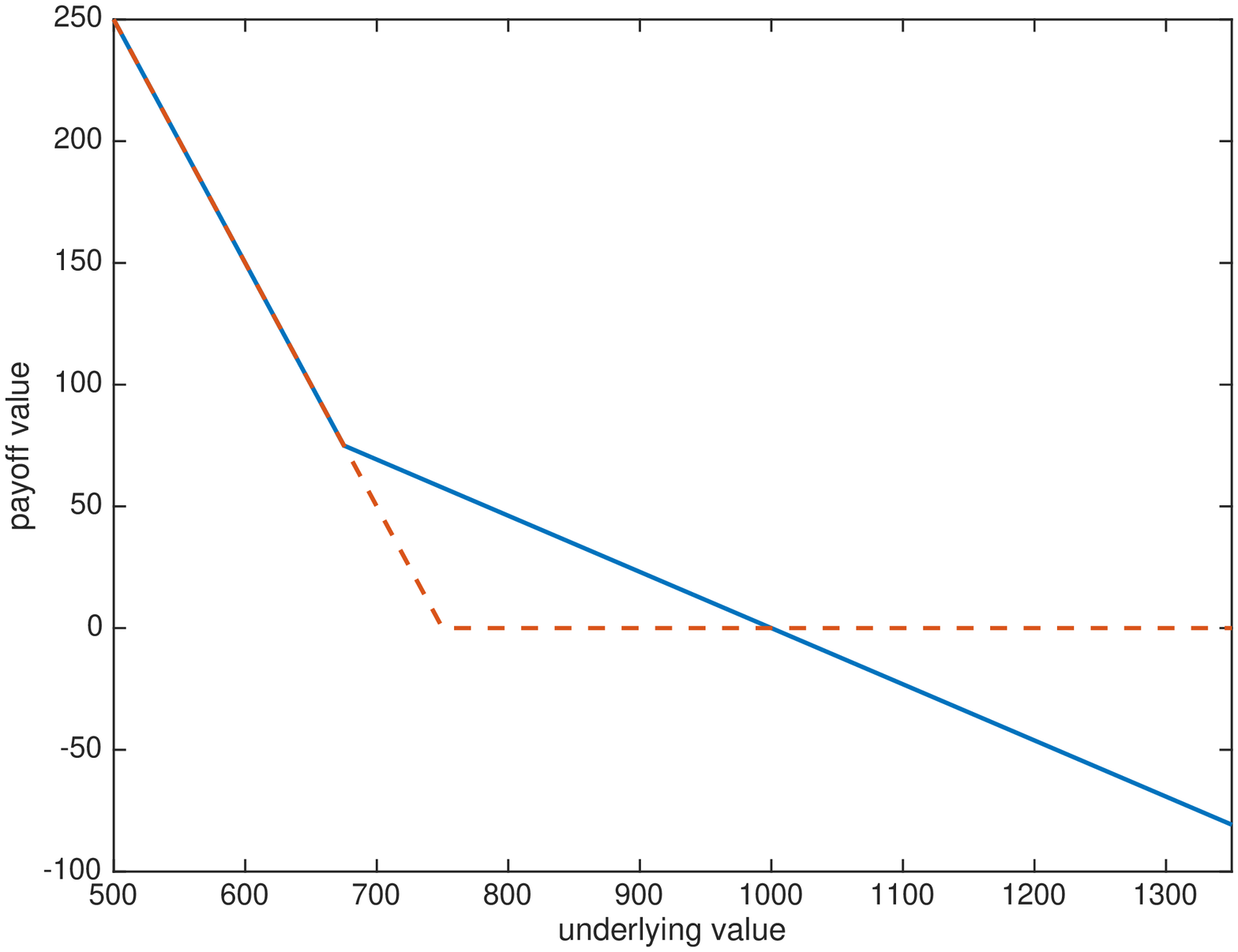}
    } & {
    \includegraphics[width = 0.48\textwidth]{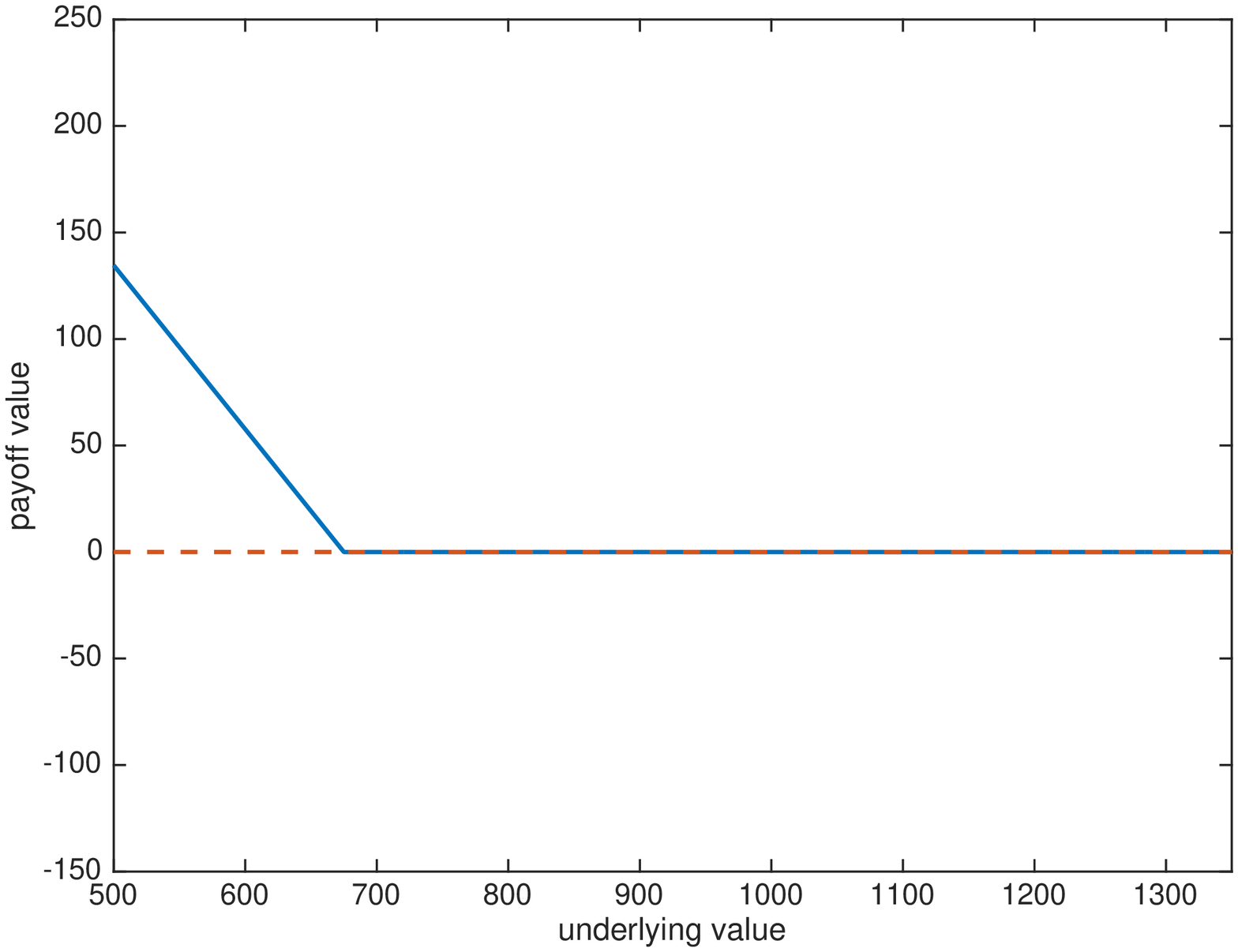}
    }  
   \end{tabular}
  \caption{Payoff functions of the BHR super-replication portfolio (in blue) and of the UOP option (in red). The left hand side corresponds to scenarios in which the underlying stays below the barrier. The right hand side corresponds to scenarios in which the underlying hits the barrier, with the forward position liquidated at the hitting time (at zero cost). Parameters used: $U=1000$, $K=0.75U=750$, $K_i = 0.9K = 675$.}
    \label{fig:BHR}
  \end{center}
\end{figure}

\begin{figure}
\begin{center}
  \begin{tabular} {cc}
    {
    \includegraphics[width = 0.48\textwidth]{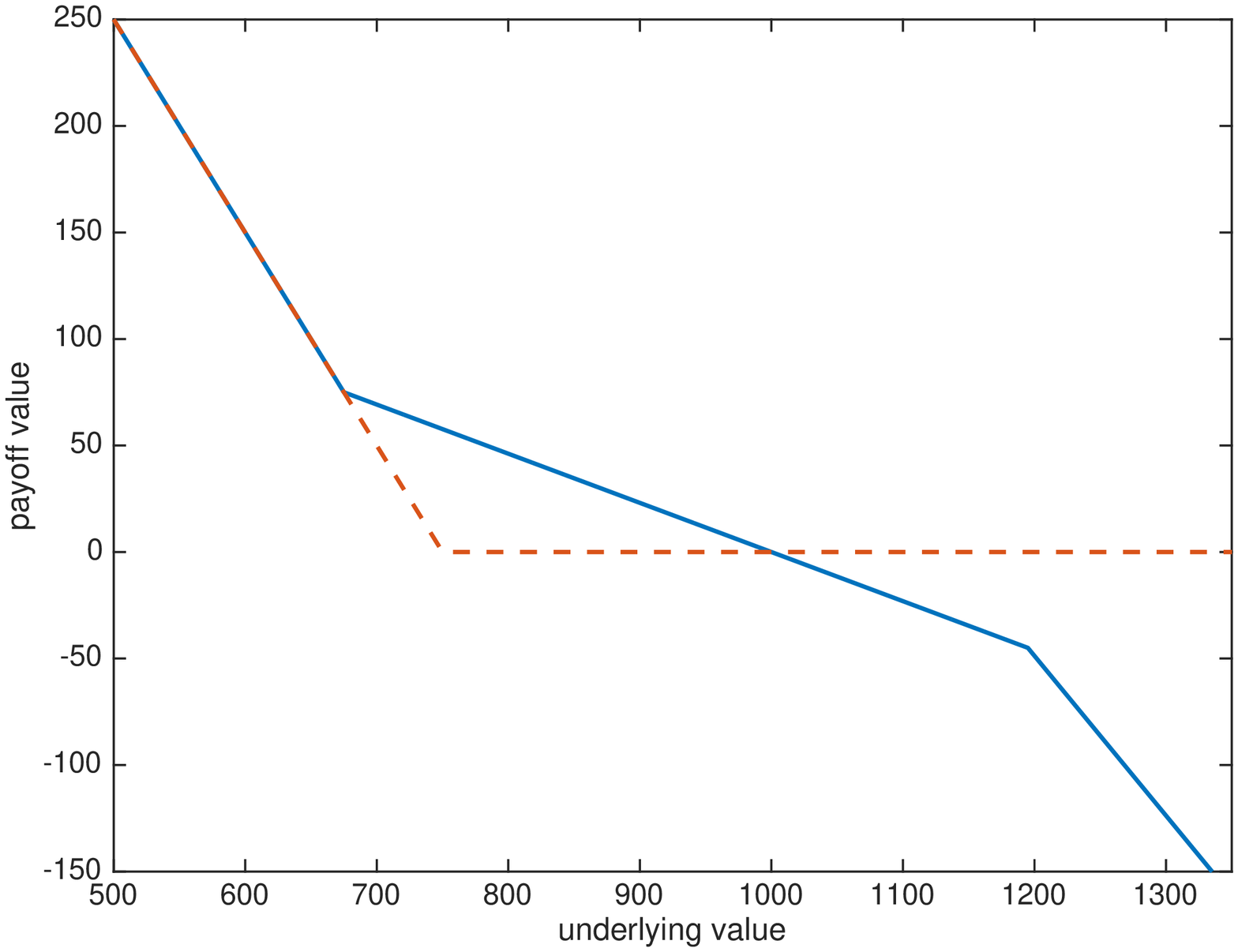}
    } & {
    \includegraphics[width = 0.48\textwidth]{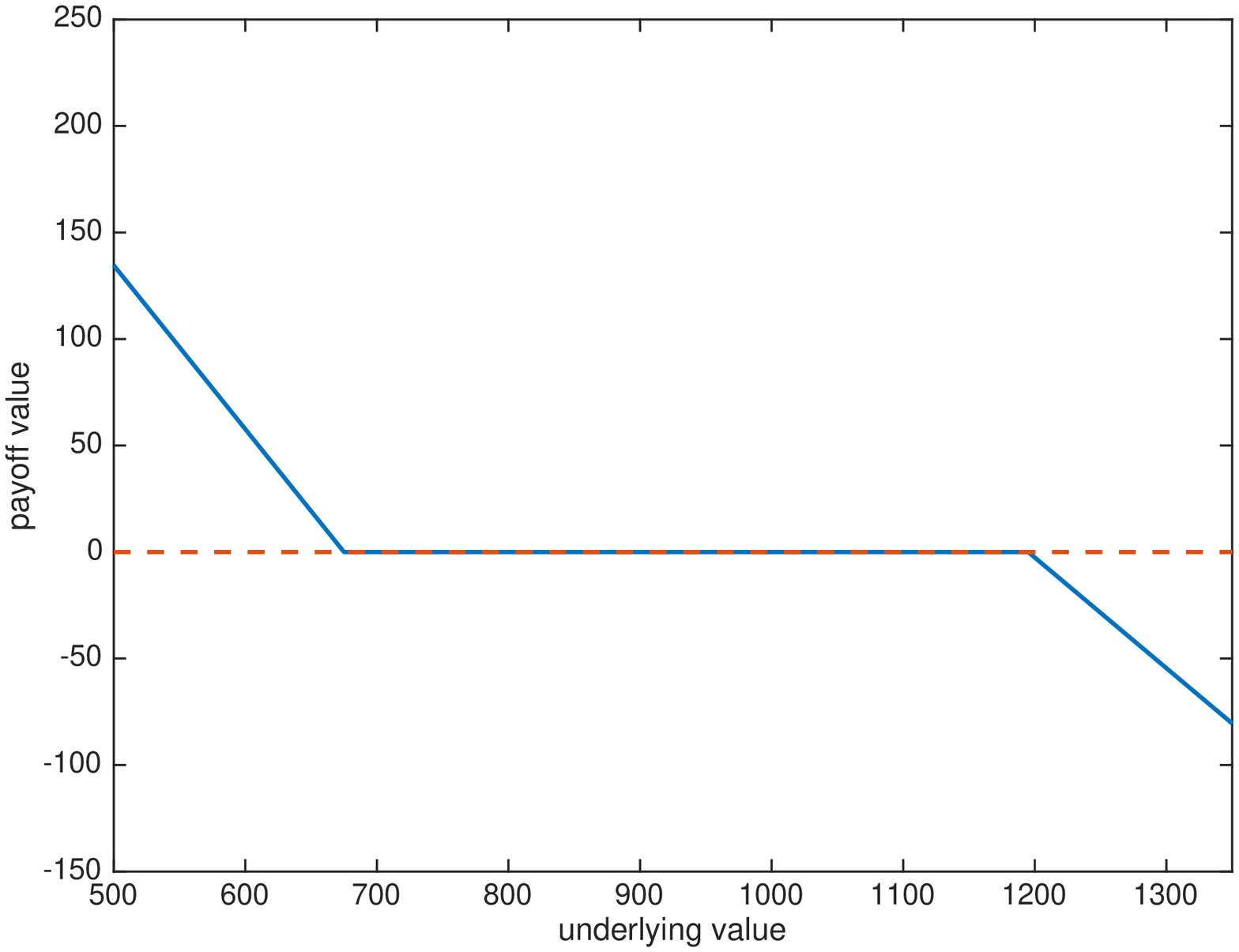}
    }  
   \end{tabular}
  \caption{Payoff functions of the proposed super-replication strategy (in blue) and of the UOP option (in red). The left hand side corresponds to scenarios in which the underlying stays below the barrier. The right hand side corresponds to scenarios in which the underlying hits the barrier, with the forward position liquidated at the hitting time (at zero cost). Parameters used: $U=1000$, $K=0.75U=750$, $K_i = 0.9K = 675$, $\sigma_2/\sigma_1=0.6$ (i.e. $\mathbf{K}^{\sigma,U}(K_i)=1195$).}
    \label{fig:BHR.mod}
  \end{center}
\end{figure}

%\newpage

\begin{figure}
\begin{center}
  \begin{tabular} {cc}
    {
    \includegraphics[width = 0.48\textwidth]{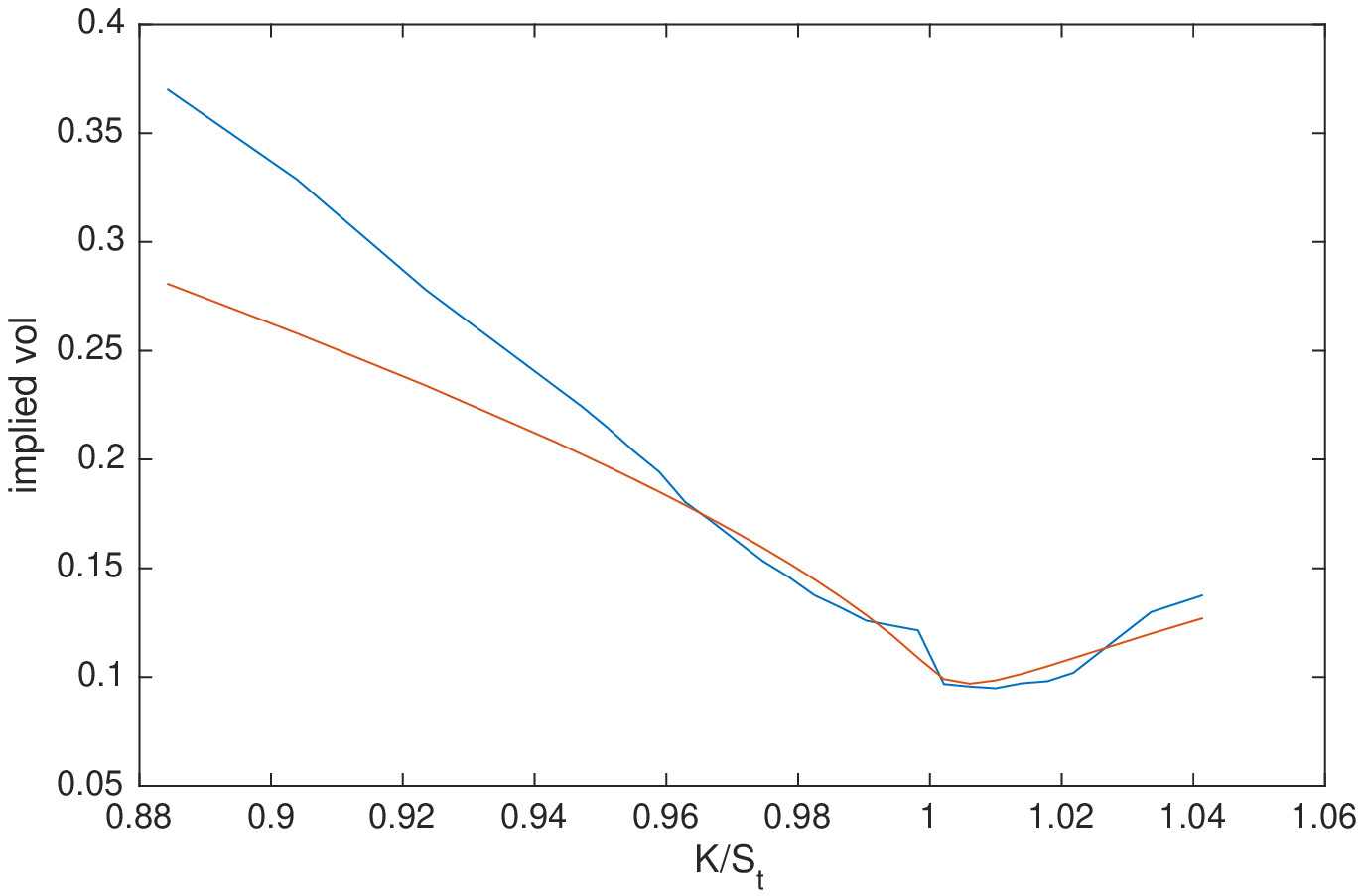}
    } & {
    \includegraphics[width = 0.48\textwidth]{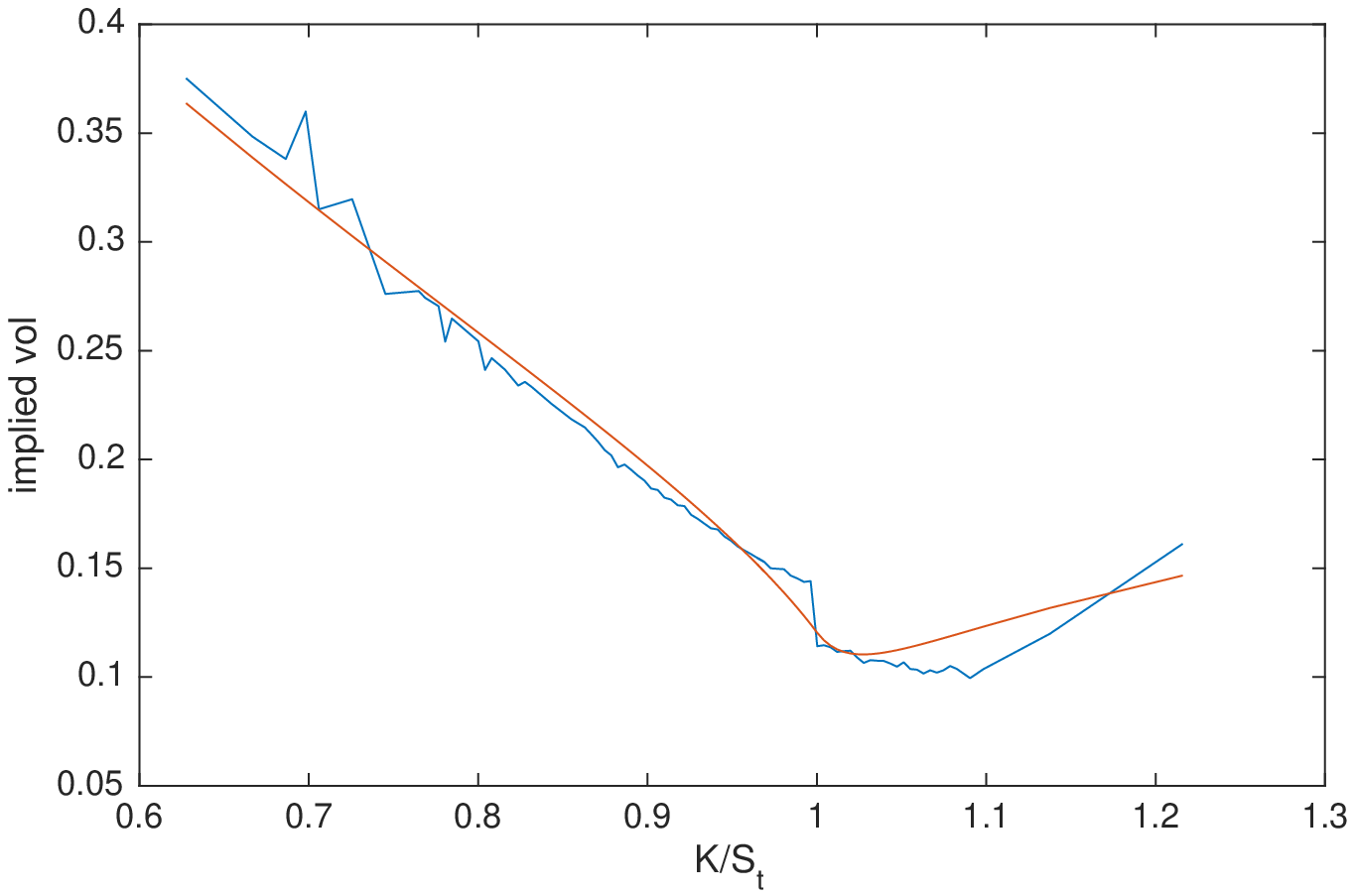}
    }\\
    {
    \includegraphics[width = 0.48\textwidth]{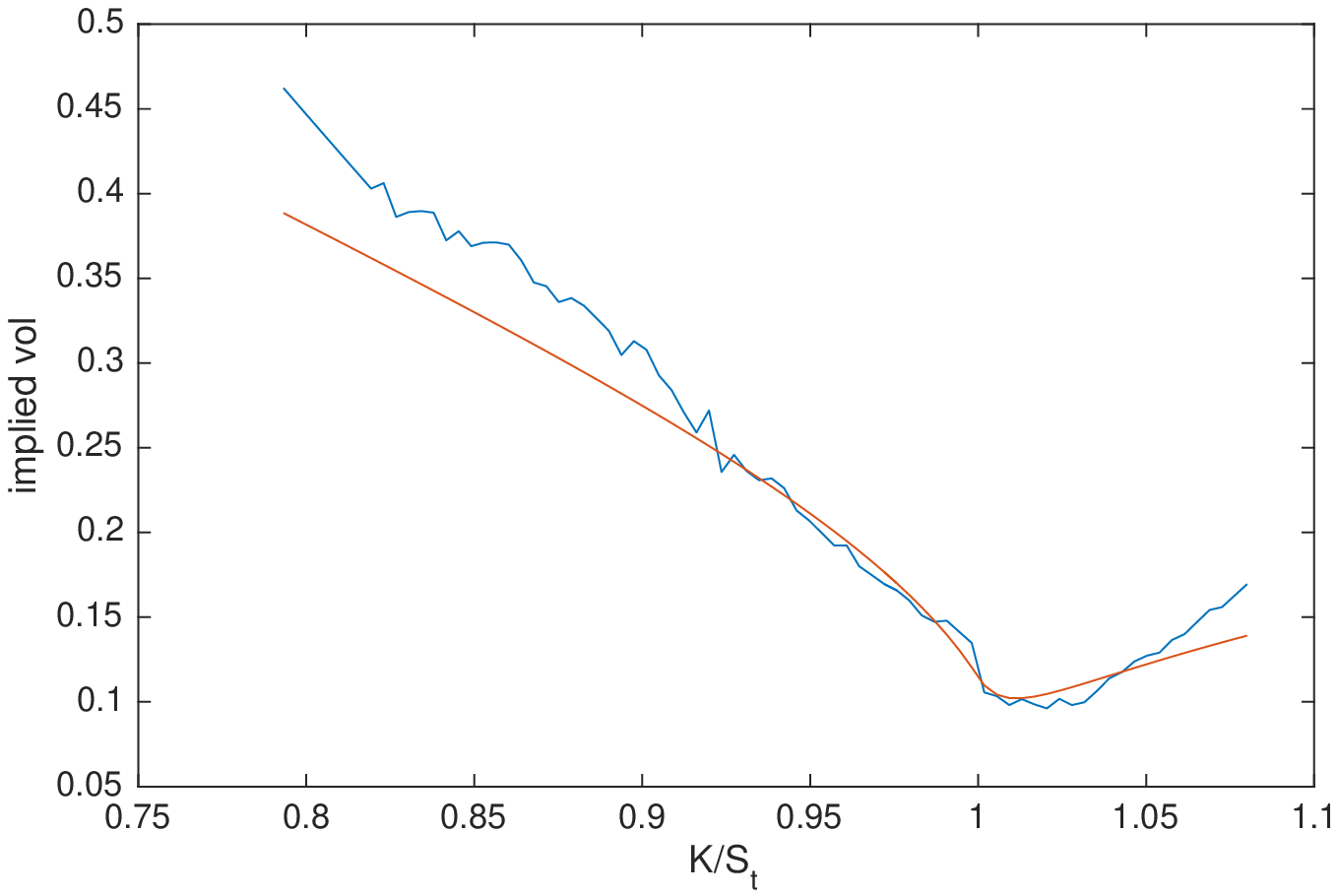}
    } & {
    \includegraphics[width = 0.48\textwidth]{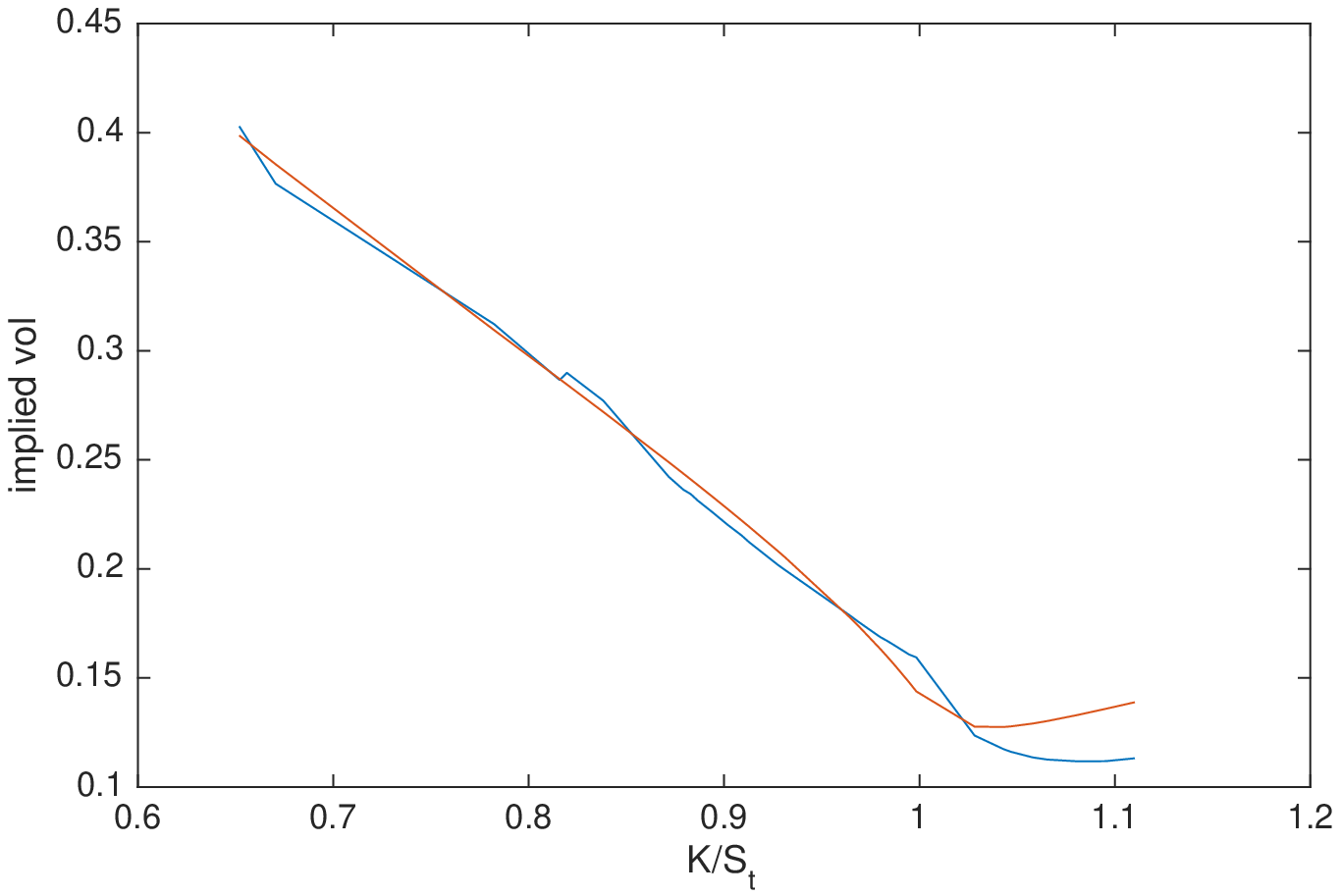}
    }
  \end{tabular}
  \caption{Implied volatility fit for different dates and times to maturity: Jan 3, 2011, 4 days to maturity (top left), Jan 3, 2011, 47 days to maturity (top right), Feb 10, 2012, 8 days to maturity (bottom left), and Feb 10, 2012, 49 days to maturity (bottom right). Blue lines represent the interpolated market implied volatilities, red lines are the calibrated PCLVG implied volatilities.}
    \label{fig:emp.1}
  \end{center}
\end{figure}

\begin{figure}
\begin{center}
  \begin{tabular} {cc}
    {
    \includegraphics[width = 0.48\textwidth]{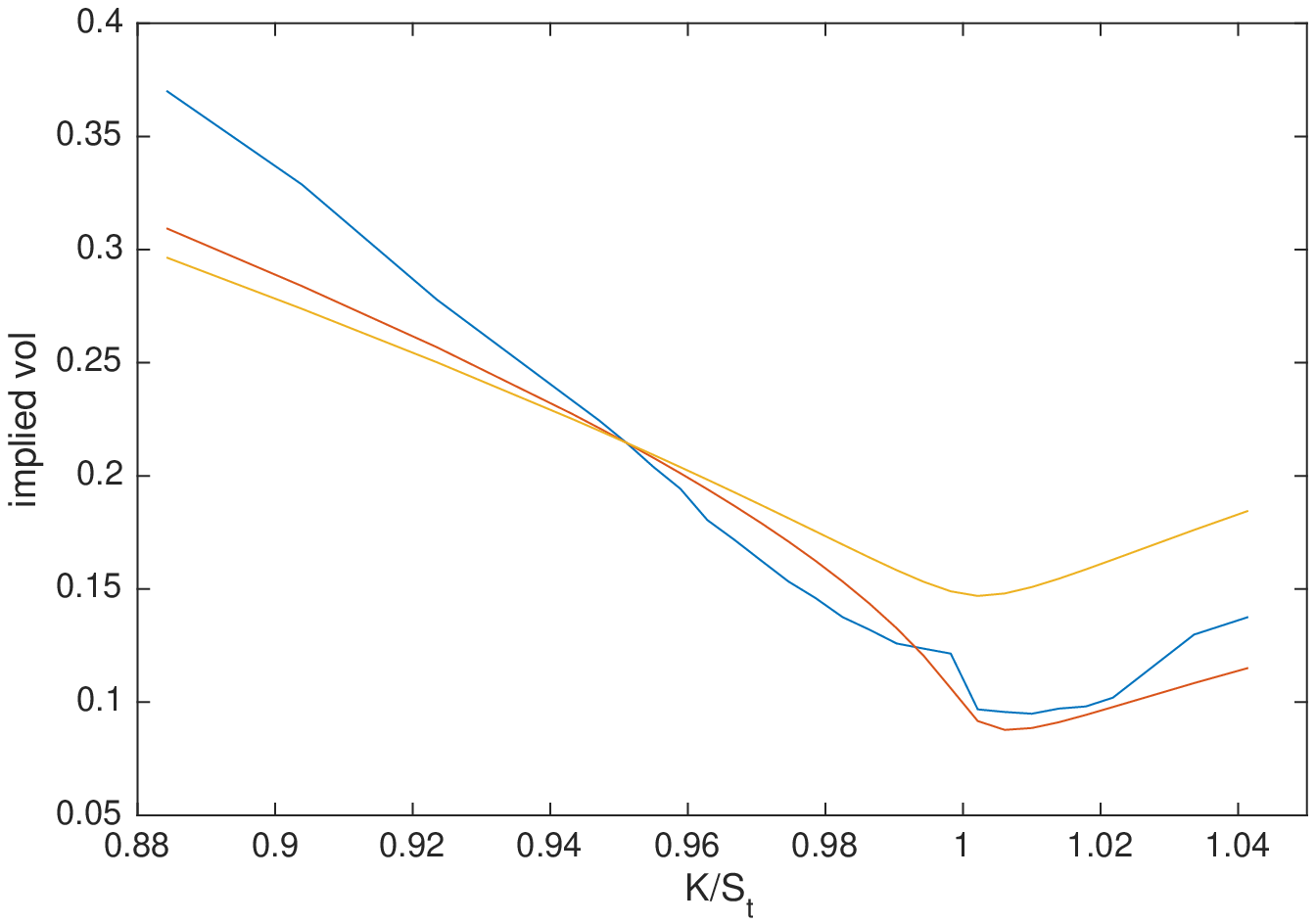}
    } & {
    \includegraphics[width = 0.48\textwidth]{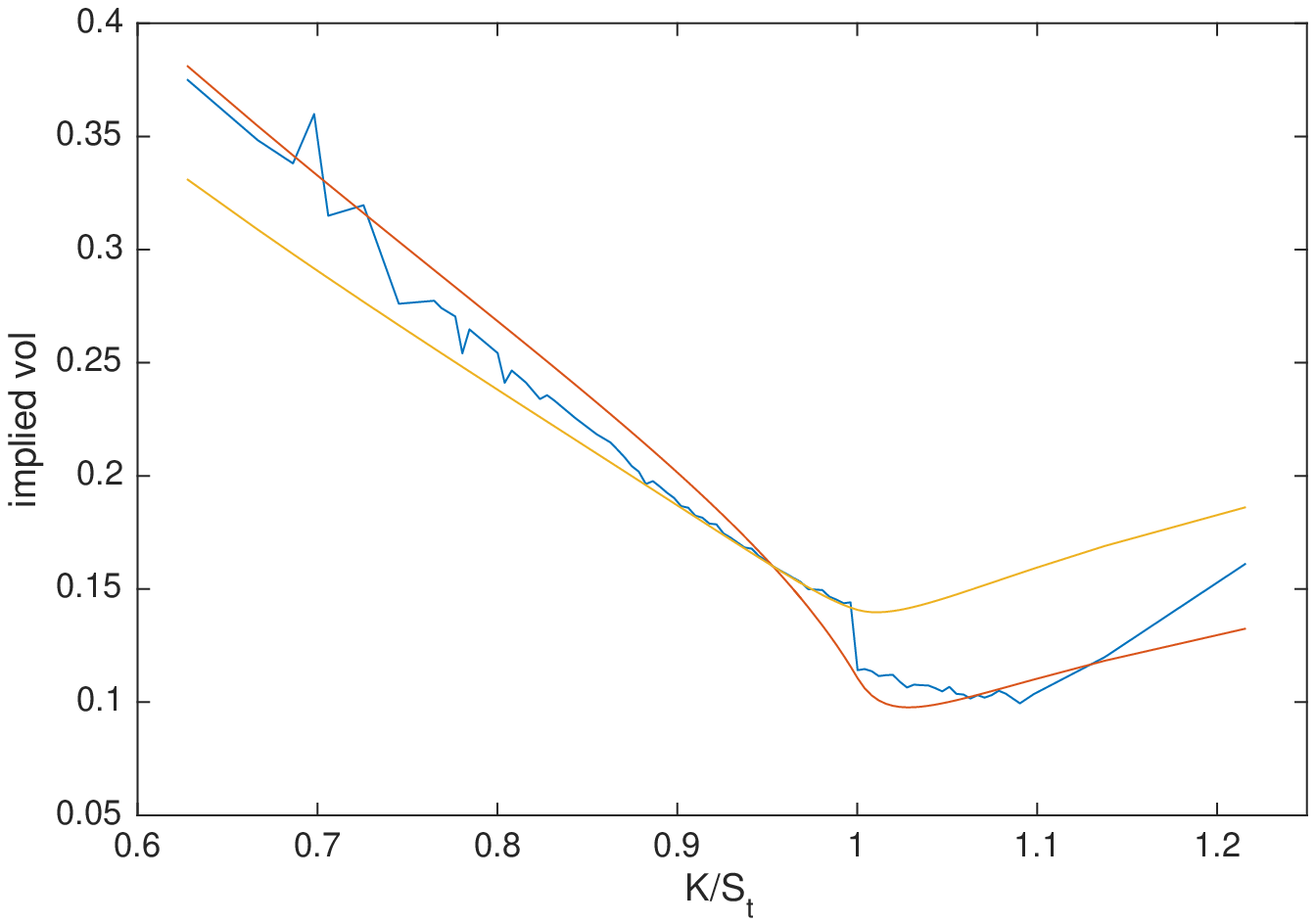}
    }\\
    {
    \includegraphics[width = 0.48\textwidth]{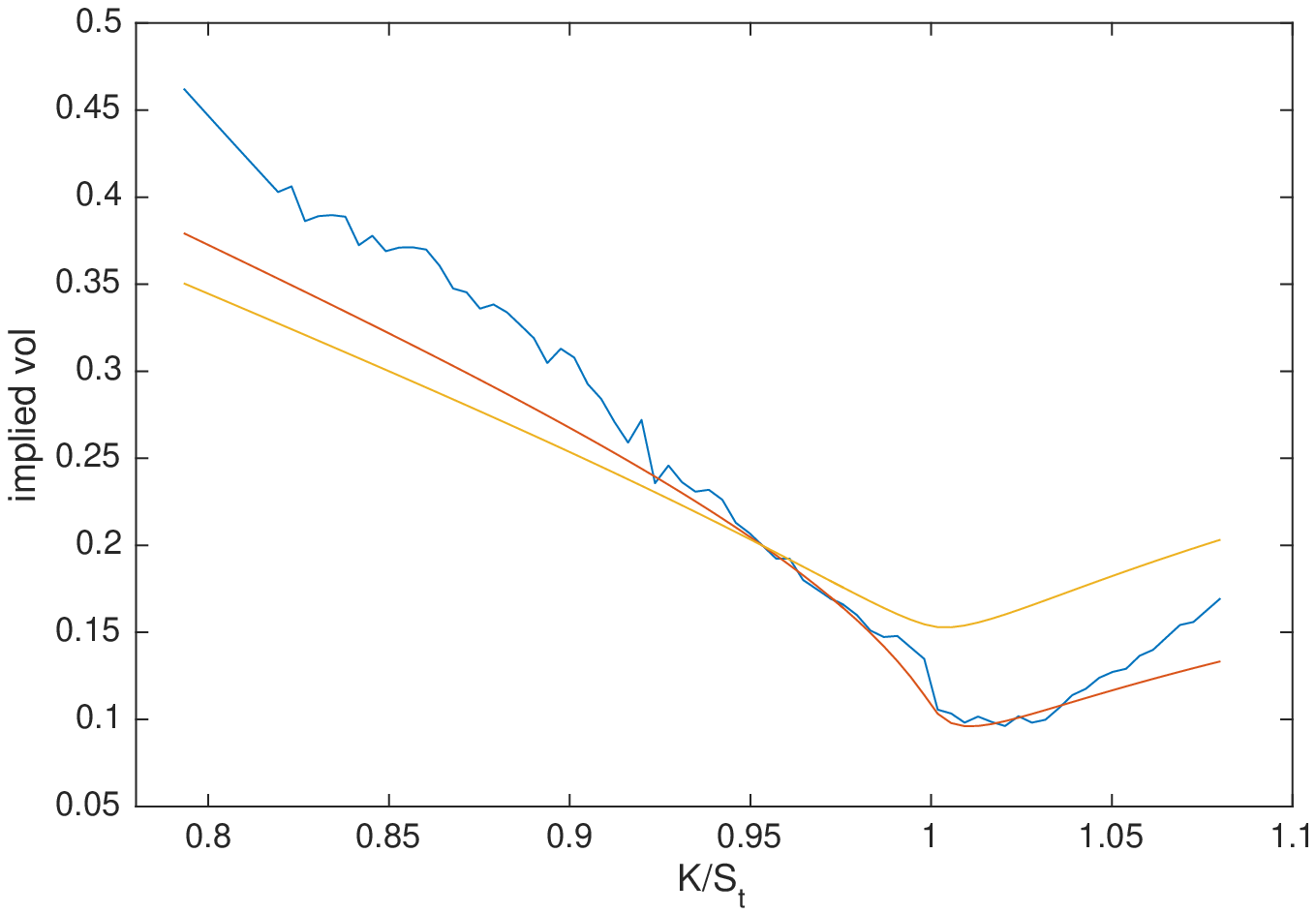}
    } & {
    \includegraphics[width = 0.48\textwidth]{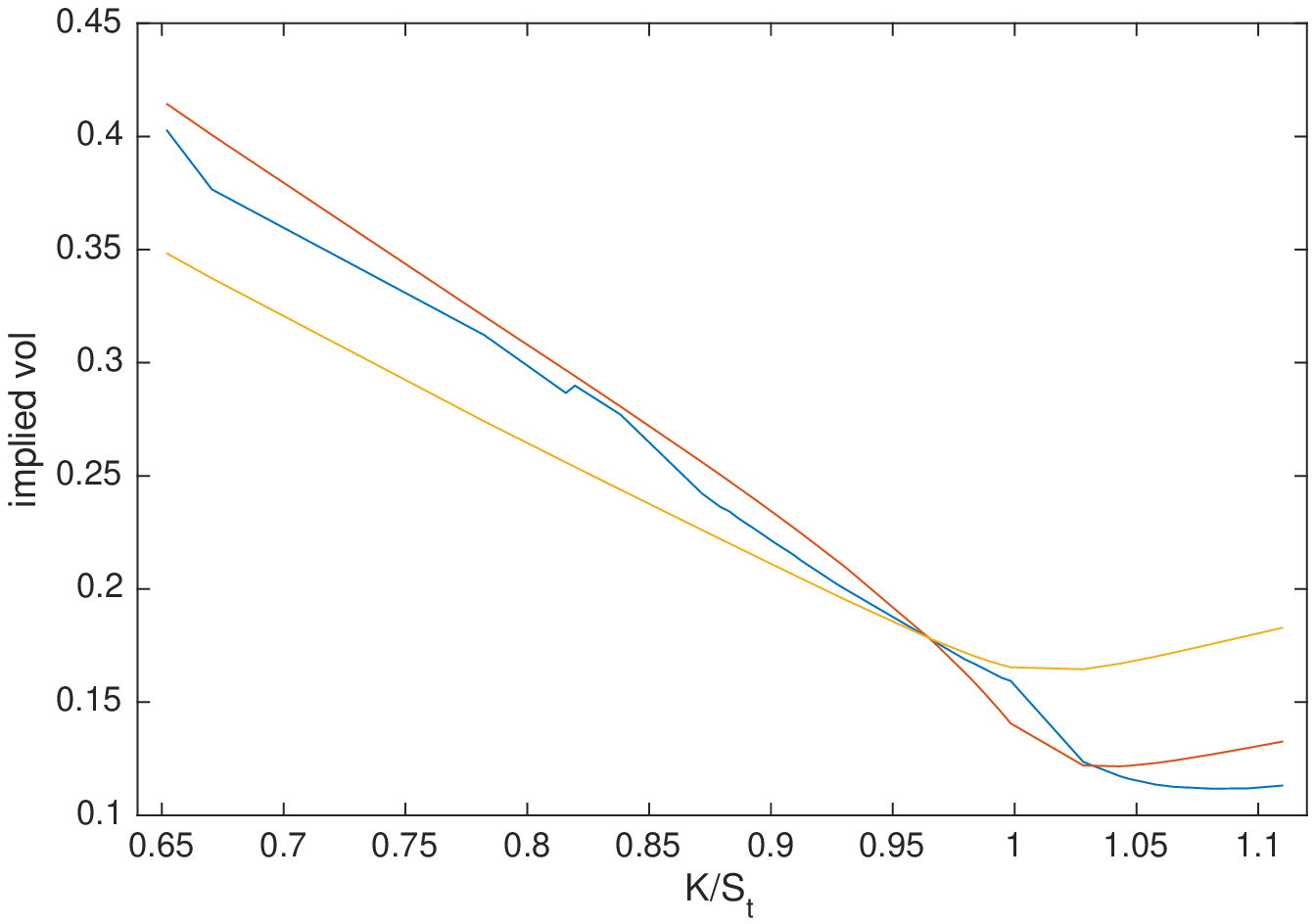}
    }
  \end{tabular}
  \caption{The dominating and market implied volatilities for different dates and times to maturity: Jan 3, 2011, 4 days to maturity (top left), Jan 3, 2011, 47 days to maturity (top right), Feb 10, 2012, 8 days to maturity (bottom left), and Feb 10, 2012, 49 days to maturity (bottom right). Blue lines represent the interpolated market implied volatilities. Red lines are the lower dominating implied volatilities, $\Sigma^{\sigma_*(t)}$, with $\sigma_{*2}(t)/\sigma_{*1}(t)=\kappa_*$. Yellow lines are the upper dominating implied volatilities, $\Sigma^{\sigma^*(t)}$, with $\sigma^*_{2}(t)/\sigma^*_{1}(t)=\kappa^*$. The put strike is always chosen as $K_i(t)=0.95 S_t$.}
    \label{fig:emp.1.5}
  \end{center}
\end{figure}

\begin{figure}
\begin{center}
  \begin{tabular} {cc}
    {
    \includegraphics[width = 0.48\textwidth]{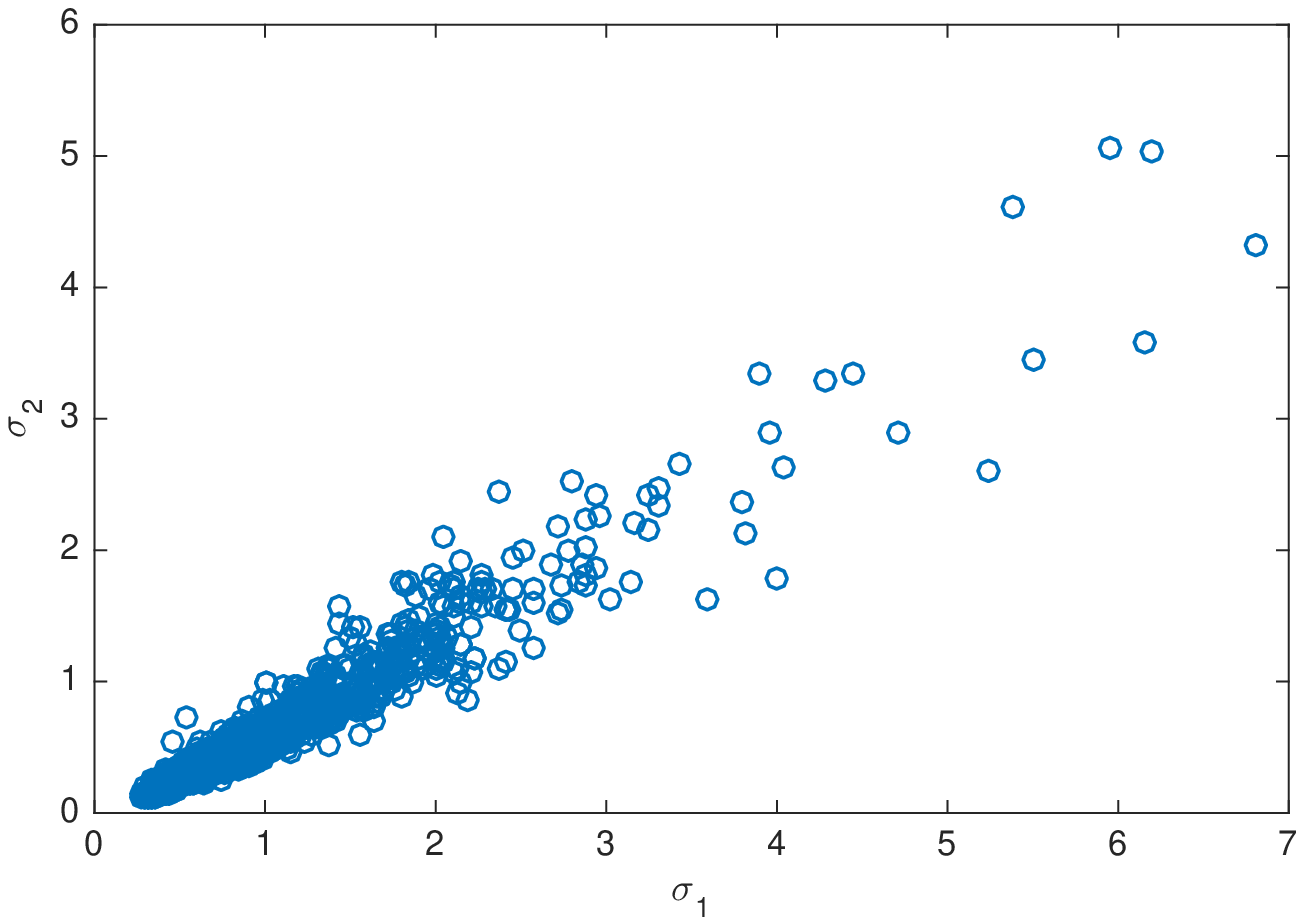}
    } & {
    \includegraphics[width = 0.48\textwidth]{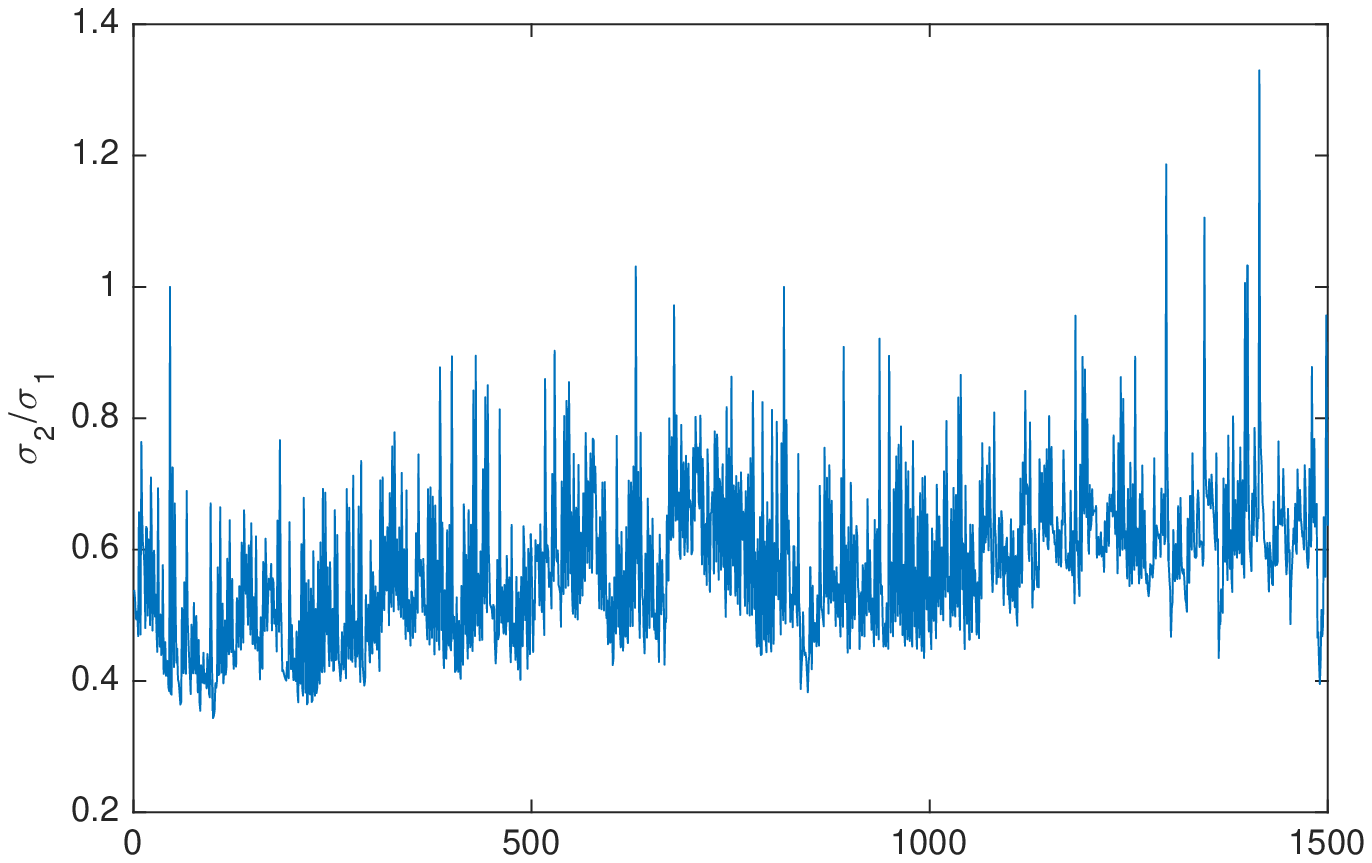}
    }
  \end{tabular}
  \caption{On the left: the values of calibrated $(\sigma_1(t,\tau),\sigma_2(t,\tau))$, every point corresponds to a fixed $(t,\tau)$. On the right: the values of $\sigma_2(t,\tau)/\sigma_1(t,\tau)$, every point corresponds to a fixed $(t,\tau)$, with the values of $(t,\tau)$ appearing in the lexicographical order.}
    \label{fig:emp.2}
  \end{center}
\end{figure}

\begin{figure}
\begin{center}
  \begin{tabular} {cc}
    {
    \includegraphics[width = 0.48\textwidth]{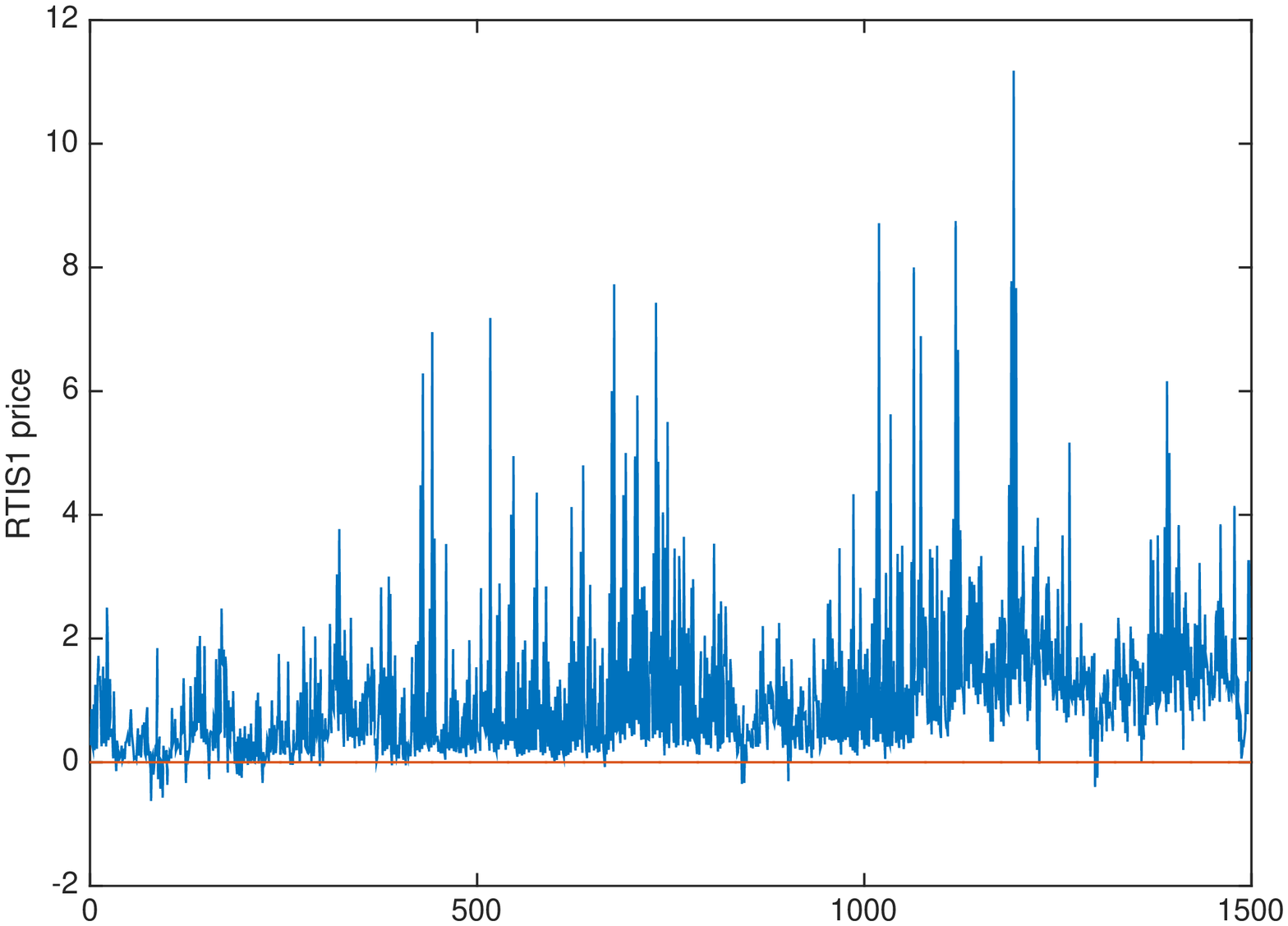}
    } & {
    \includegraphics[width = 0.48\textwidth]{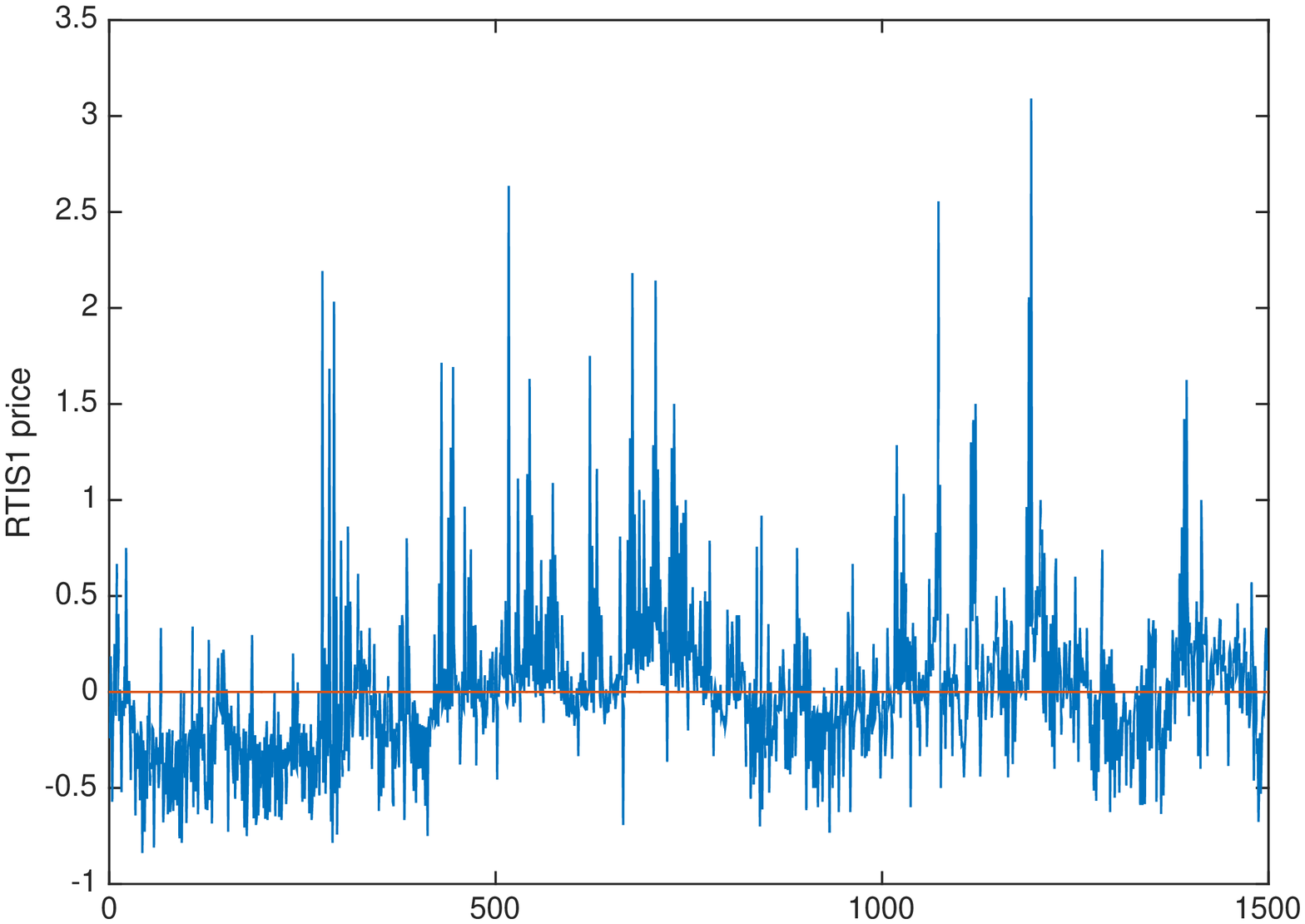}
    }
  \end{tabular}
  \caption{Value of the first RTIS portfolio, given by (\ref{eq.RTIS.def}), for various $(t,\tau)$. The put strike is always chosen as $K_i = 0.95 S_t$, and the barrier is $U=S_t$. The call strikes on the left graph are chosen as $K_{\underline{j}^{\sigma_*}(i)}$, with $\sigma_{*2}/\sigma_{*1}=\kappa_*$. The call strikes on the right graph are chosen as $K_{\underline{j}^{\bar{\sigma}}(i)}$, with $\bar{\sigma}_2/\bar{\sigma}_1$ being the average of $\sigma_2(t,\tau)/\sigma_1(t,\tau)$ over the sample (i.e over all $(t,\tau)$).}
    \label{fig:emp.3}
  \end{center}
\end{figure}

\begin{figure}
\begin{center}
  \begin{tabular} {cc}
    {
    \includegraphics[width = 0.48\textwidth]{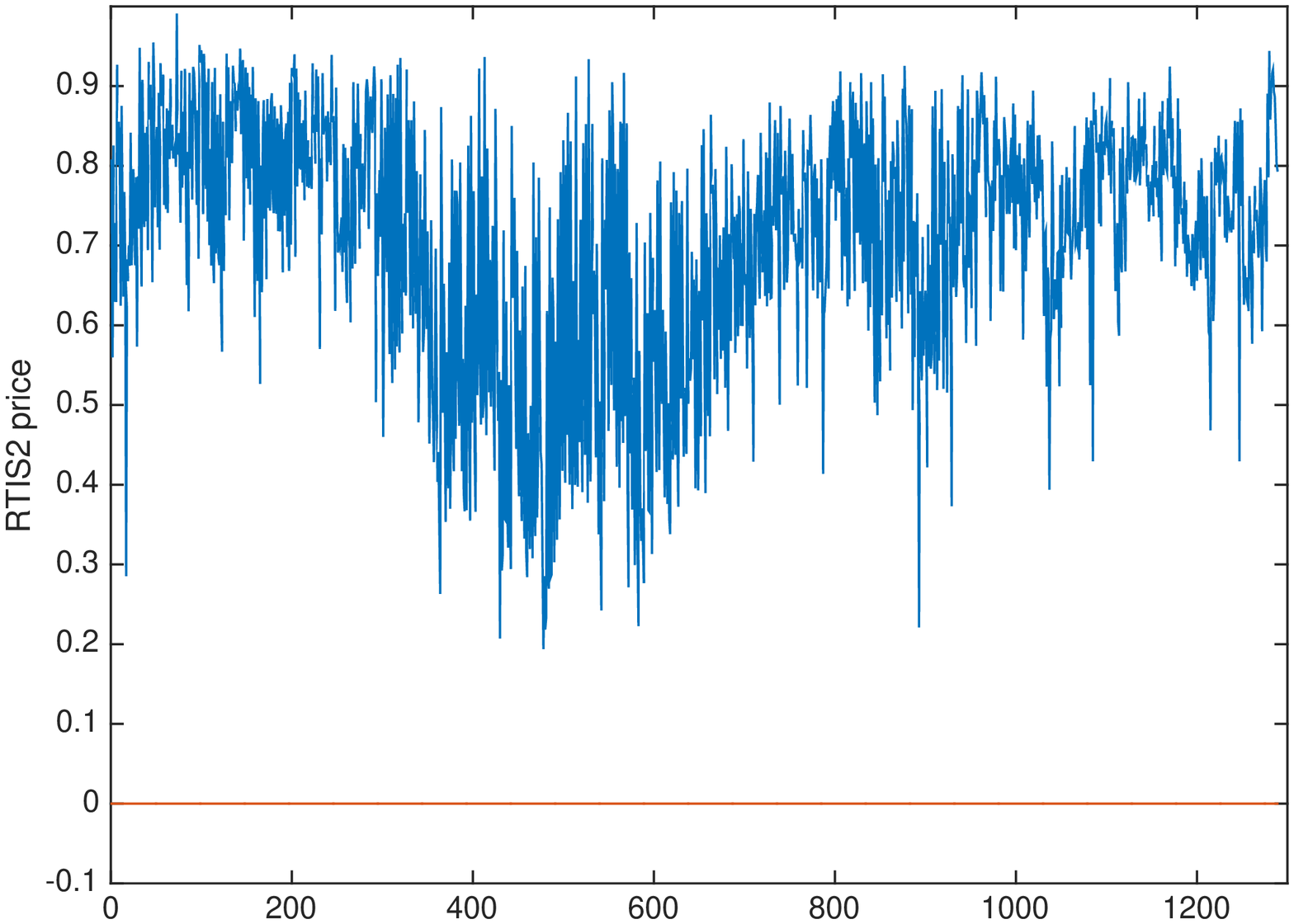}
    } & {
    \includegraphics[width = 0.48\textwidth]{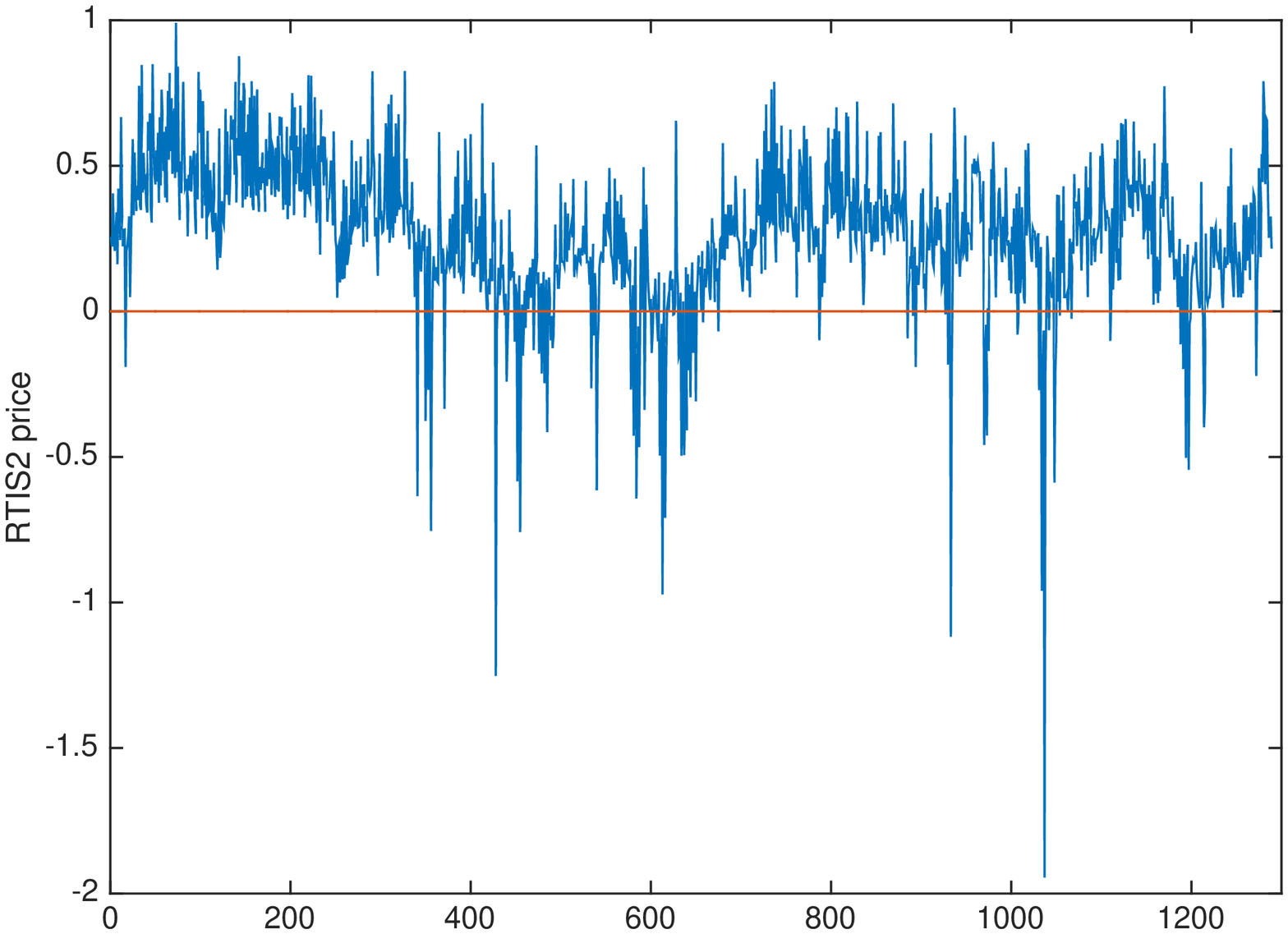}
    }
  \end{tabular}
  \caption{Value of the second RTIS portfolio, given by (\ref{eq.RTIS.upper.def}), for various $(t,\tau)$. The put strike is always chosen as $K_i = 0.95 S_t$, and the barrier is $U=S_t$. The call strikes on the left graph are chosen as $K_{\overline{j}^{\sigma^*}(i)}$, with $\sigma^*_{2}/\sigma^*_{1}=\kappa^*$. The call strikes on the right graph are chosen as $K_{\overline{j}^{\bar{\sigma}}(i)}$, with $\bar{\sigma}_2/\bar{\sigma}_1$ being the average of $\sigma_2(t,\tau)/\sigma_1(t,\tau)$ over the sample (i.e over all $(t,\tau)$).}
    \label{fig:emp.4}
  \end{center}
\end{figure}

\begin{figure}
\begin{center}
  \begin{tabular} {cc}
    {
    \includegraphics[width = 0.48\textwidth]{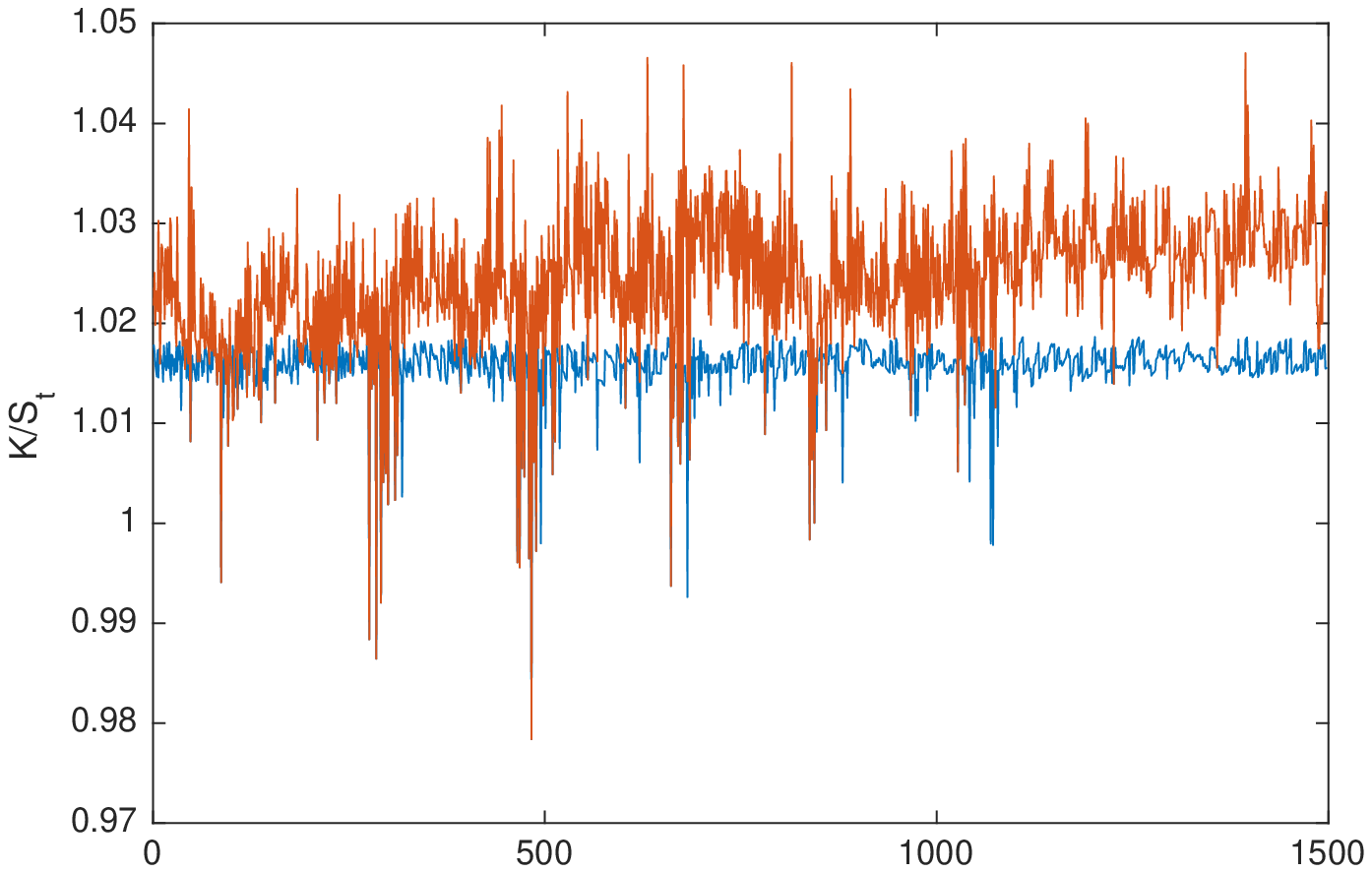}
    } & {
    \includegraphics[width = 0.48\textwidth]{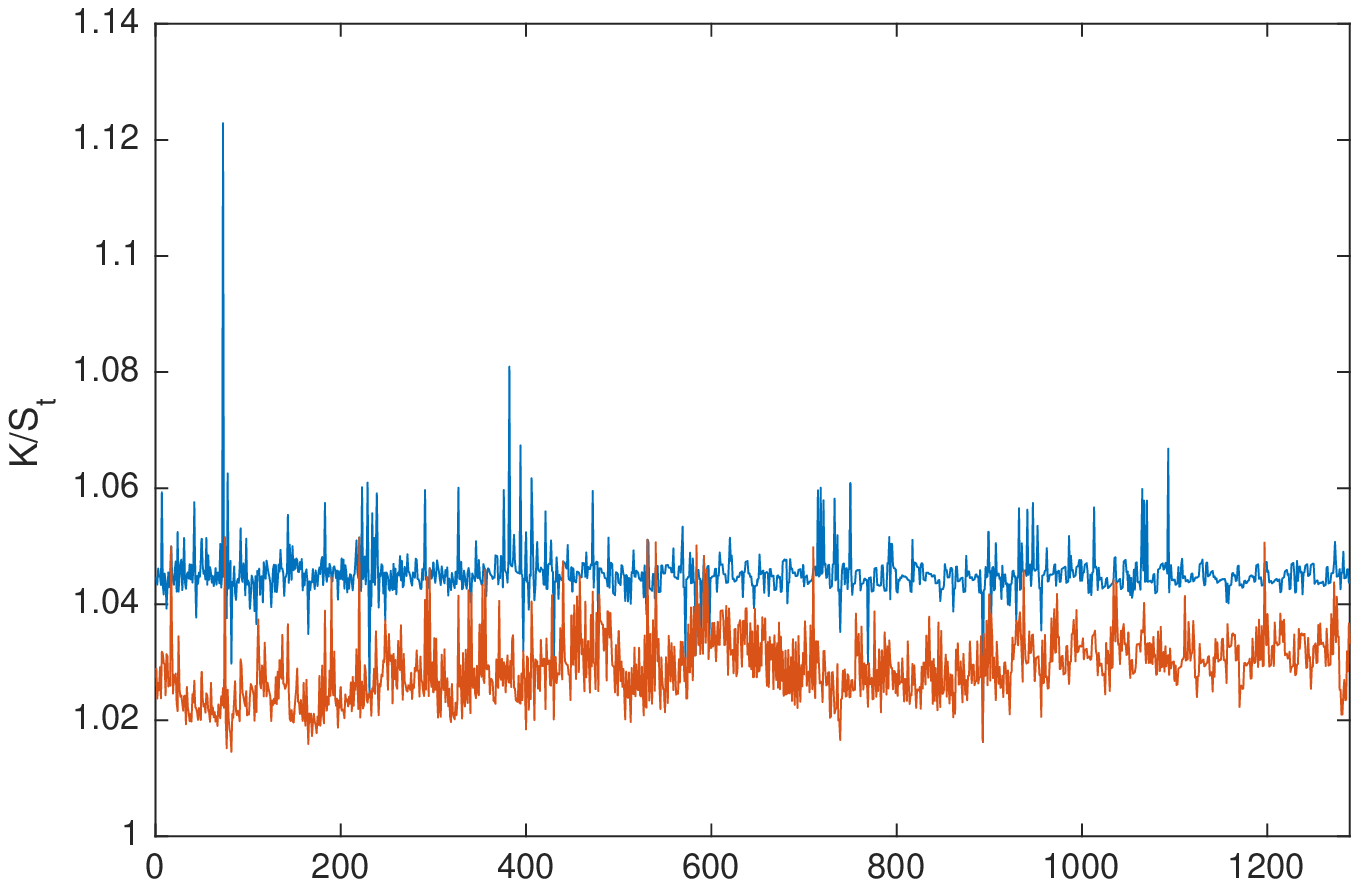}
    }
  \end{tabular}
  \caption{Moneyness of the call strikes used in the RTIS portfolios, across different $(t,\tau)$. The left graph contains $K_{\underline{j}^{\sigma_*}(i)}$ (in blue), with $\sigma_{*2}/\sigma_{*1}=\kappa_*$, and $K_{\underline{j}^{\sigma(t,\tau)}(i)}$ (in orange), with $\sigma(t,\tau)$ being the PCLVG parameter calibrated to the market implied smile at time $t$, for the time to maturity $\tau$. The right graph contains $K_{\overline{j}^{\sigma^*}(i)}$ (in blue), with $\sigma^*_{2}/\sigma^*_{1}=\kappa^*$, and $K_{\overline{j}^{\sigma(t,\tau)}(i)}$ (in orange).}
    \label{fig:emp.5}
  \end{center}
\end{figure}

%\newpage

\end{document}